\documentclass{CSML}

\pdfoutput=1

\usepackage{lastpage}

\def\dOi{14(1:3)2018}
\lmcsheading%
{\dOi}
{1--\pageref{LastPage}}
{}
{}
{Jun.~\phantom01, 2016}
{Jan.~10, 2018}
{}

\usepackage{etex}
\usepackage{stmaryrd}
\usepackage{bussproofs}
\EnableBpAbbreviations
\let\lmcsproof=\proof
\usepackage{proof}
\usepackage{amsmath,amssymb,amstext}
\usepackage{wrapfig}
\usepackage{framed}
\usepackage{url}
\usepackage[hidelinks]{hyperref}

\usepackage{tikz}
\usetikzlibrary{arrows,automata,backgrounds,calc,chains,fadings,folding,decorations.fractals,fit,patterns,positioning,mindmap,shadows,shapes.geometric,shapes.symbols,through,trees,decorations.markings,decorations.text}
\colorlet{dblue}{blue!40!black}

\renewenvironment{proof}{\lmcsproof}{\hfill\qed}

\newcommand{\remove}[1]{}

\theoremstyle{plain}
\newtheorem{theorem}{Theorem}

\newtheorem*{conjecture*}{Conjecture}

\theoremstyle{definition}

\newtheorem{formalisation}[theorem]{Formalisation}

\newcommand{\msf}{\mathsf}
\newcommand{\mit}{\mathit}

\newcommand{\ap}{{\scriptsize@}}

\newcommand{\pairlft}{{\langle}}
\newcommand{\pairrgt}{{\rangle}}
\newcommand{\pairsep}{{,\,}}
\newcommand{\pairstr}[1]{\pairlft#1\pairrgt}
\newcommand{\pair}[2]{\pairstr{#1\pairsep#2}}

\newcommand{\sBNFis}{{{:}{:}{=}}}

\newcommand{\BNFor}{\mathrel{|}}
\newcommand{\coBNFis}{\mathrel{\sBNFis^{\text{co}}}}

\newcommand{\atrs}{\mathcal{R}}

\newcommand{\sred}{{\rightarrow}}

\newcommand{\rerat}[2]{\mathrel{\sred_{#1,#2}}}

\newcommand{\mred}{\to\hspace{-2.8mm}\to}

\newcommand{\ieq}{\stackrel{\infty}{=}}

\newcommand{\harpoondown}{\mathrel{\ooalign{$\leftharpoondown$\cr$\rightharpoondown$\cr}}}

\newcommand{\ieqdown}{\stackrel{\infty}{{\harpoondown}}}%
\newcommand{\ired}{\to^\infty}%
\newcommand{\iredi}{\mathrel{\reflectbox{$\ired$}}}
\newcommand{\ireddown}{\rightharpoondown^\infty}%
\newcommand{\ireddownfin}{\stackrel{\makebox(0,0){\raisebox{3pt}{{\tiny$<$}}}}{\rightharpoondown}^\infty}%
\newcommand{\ireddownmfin}{\stackrel{\makebox(0,0){\raisebox{3pt}{{\tiny($<$)}}}}{\rightharpoondown}^\infty}%
\newcommand{\rstep}{\to_{\varepsilon}}

\newcommand{\ored}{\to^{\le \omega}}%

\newcommand{\cmred}{\to_{m}}
\newcommand{\cired}{\to_{i}}
\newcommand{\cored}{\to_{o}}
\newcommand{\crstep}{\to_{r}}

\newcommand{\ibi}{\stackrel{\infty}{\to}}%
\newcommand{\ibiinv}{\stackrel{\infty}{\leftarrow}}%
\newcommand{\ibidown}{\stackrel{\infty}{\rightharpoondown}}%

\newcommand{\redord}{\to_{\textit{ord}}}
\newcommand{\redorddown}{\rightharpoondown_{\textit{ord}}}
\newcommand{\iredord}{\ired_{\textit{ord}}}

\newcommand{\iredorddown}{\ireddown_{\textit{ord}}}

\newcommand{\fun}[1]{\mathsf{#1}}

\newcommand{\relcomp}{\mathbin{\boldsymbol{;}}}
\newcommand{\funcomp}{\circ}

\newcommand{\nat}{\mathbb{N}}

\newcommand{\sarity}{\mit{ar}}
\newcommand{\arity}[1]{\sarity(#1)}
\newcommand{\avars}{\mathcal{X}}
\newcommand{\vars}[1]{\mathcal{V}\hspace{-.1ex}\mit{ar}(#1)}

\newcommand{\ster}{\mathit{Ter}}
\newcommand{\ter}[2]{\ster(#1,#2)}
\newcommand{\siter}{\ster^{\infty}}
\newcommand{\iter}[2]{\siter(#1,#2)}

\newcommand{\sbis}{{\sim}}
\newcommand{\xbis}{\mathrel{\sbis}}

\newcommand{\posemp}{\varepsilon}
\newcommand{\apos}{p}

\newcommand{\pos}[1]{\mathcal{P}\!os(#1)}

\newcommand{\subtrm}[2]{#1|_{#2}}
\newcommand{\asubst}{\sigma}

\newcommand{\slfp}{\mu}
\newcommand{\lfp}[2]{\slfp {#1}.\,#2}
\newcommand{\sgfp}{\nu}
\newcommand{\gfp}[2]{\sgfp {#1}.\,#2}

\newcommand{\down}[1]{\overline{#1}}

\newcommand{\id}{\mathrm{Id}}%

\newcommand{\tlat}{L}

\newcommand{\Pow}{\mathcal{P}}

\newcommand{\rsplit}{\ensuremath{\msf{split}}}
\newcommand{\rlift}{\ensuremath{\msf{lift}}}
\newcommand{\rid}{\ensuremath{\msf{id}}}%

\newcommand{\nest}{\textit{der}}
\newcommand{\fp}{\textit{fp}}

\newcommand{\boxx}[1]{\colorbox[rgb]{0.99,0.78,0.07}{\kern0.15em#1\kern0.15em}\quad}

\newcommand{\hole}{\raisebox{-2pt}{\scalebox{.7}[1.5]{$\Box$}}}

\newcommand{\sdefd}{:=}
\newcommand{\defd}{\mathrel{\sdefd}}

\newcommand{\aes}{\atrs}%

\renewcommand{\emptyset}{\varnothing}

\newcommand{\fap}[2]{#1(#2)}
\newcommand{\bfap}[3]{\fap{#1}{#2,#3}}

\newcommand{\smetric}{\mathsf{d}}
\newcommand{\metric}{\bfap{\smetric}}

\newcommand{\smktree}{\mathfrak{T}}
\newcommand{\mktree}{\fap{\smktree}}
\newcommand{\smktreedown}{\smktree'}
\newcommand{\mktreedown}{\fap{\smktreedown}}

\newcommand{\smktreedownfin}{\smktree'_{<}}
\newcommand{\mktreedownfin}{\fap{\smktreedownfin}}
\newcommand{\smktreedownmfin}{\smktree'_{(<)}}
\newcommand{\mktreedownmfin}{\fap{\smktreedownmfin}}

\newcommand{\sto}{\rightsquigarrow}

\newcommand{\tail}{\msf{f}}

\newcommand{\proj}[2]{\mathit{proj}_{#1}(#2)}
\newcommand{\rulapp}{\mathit{rulapp}}

\setlist[enumerate]{font=\normalfont,labelindent=*,leftmargin=*,start=1,label=(\roman*)}

\begin{document}

\title[Coinductive Foundations of Infinitary Rewriting and Equational Logic]
{Coinductive Foundations\\ of Infinitary Rewriting\\ and Infinitary Equational Logic\rsuper*}

\author[J.~Endrullis]{J\"{o}rg Endrullis\rsuper a}
\address[]{{\lsuper{a,c}}%
  Department of Computer Science, 
  VU University Amsterdam, The Netherlands
}
\email{\{j.endrullis, r.d.a.hendriks\}@vu.nl}

\author[H.~N.~Hansen]{Helle Hvid Hansen\rsuper b}
\address[]{{\lsuper b}%
  Department of Engineering Systems and Services,
  Delft University of Technology, The Netherlands
}
\email{h.h.hansen@tudelft.nl}

\author[D.~Hendriks]{Dimitri Hendriks\rsuper c} 
\address[]{\vspace{-18 pt}}

\author[A.~Polonsky]{Andrew Polonsky\rsuper d}
\address[]{{\lsuper d}%
  Institut de Recherche en Informatique Fondamentale, Paris Diderot University, France
}
\email{andrew.polonsky@gmail.com}

\author[A.~Silva]{Alexandra Silva\rsuper e}
\address[]{{\lsuper e}%
  Department of Computer Science, University College London, England
}
\email{alexandra.silva@ucl.ac.uk}

\keywords{infinitary rewriting, infinitary equational reasoning, coinduction}

\subjclass{D.1.1, D.3.1, F.4.1, F.4.2, I.1.1, I.1.3}

\titlecomment{{\lsuper*}%
  This is a modified and extended version of~\cite{endr:hans:hend:polo:silv:2015}
  which appeared in the proceedings of RTA~2015.
}

\maketitle
\newcommand{\die}[1]{}
\begin{abstract}
\die{  We present a coinductive framework for developing the infinitary analogues of equational reasoning 
  and term rewriting in a uniform way. We define $\ieq$, a notion of 
  infinitary equational reasoning with respect to an equality relation
  $=_{\aes}$, and $\ired$, the standard notion of infinitary rewriting
  with respect to a reduction relation $\to_{\atrs}$, as follows:}
We present a coinductive framework for defining and reasoning about the infinitary analogues of equational logic
  and term rewriting in a uniform way. We define $\ieq$, the
  infinitary extension of a given equational theory $=_{\aes}$,
  and $\ired$, the standard notion of infinitary rewriting
  associated to a reduction relation $\to_{\atrs}$, as follows:
  \begin{align*}
    {\ieq} &\;\;\defd\;\; \gfp{R}{(=_{\aes} \cup \mathrel{\down{R}})^*} \\ %
    {\ired} &\;\;\defd\;\; \lfp{R}{\gfp{S}{(\to_{\atrs} \cup \mathrel{\down{R}})^*\relcomp \down{S}}} %
  \end{align*}
  Here $\slfp$ and $\sgfp$ are the least and greatest fixed-point operators, respectively,
  and 
  \begin{align*}
    \down{R} \;\defd\; \{\,\pair{f(s_1,\ldots,s_n)}{\,f(t_1,\ldots,t_n)} \mid f \in \Sigma,\, s_1\! \mathrel{R} t_1,\ldots,s_n\! \mathrel{R} t_n\,\} \,\cup\, \id  \,.
  \end{align*}
  The setup captures rewrite sequences of arbitrary ordinal length, 
  but it has neither the need for ordinals nor for metric convergence. 
  This makes the framework especially suitable for formalizations in theorem provers.
\end{abstract}

\section{Introduction}\label{sec:intro}

We present a coinductive framework for defining infinitary equational reasoning and infinitary rewriting
in a uniform way. The framework is free of ordinals, metric convergence and partial orders on terms which have been
essential in earlier definitions of the 
concept of infinitary rewriting~\cite{ders:kapl:plai:1991,kenn:klop:slee:vrie:1995a,klop:vrij:2005,kenn:vrie:2003,kahr:2013,bahr:2010,bahr:2010b,bahr:2012,endr:hend:klop:2012}.

Infinitary rewriting is a generalization of the ordinary finitary rewriting
to infinite terms and infinite reductions (including reductions of ordinal length greater than $\omega$).
For the definition of rewrite sequences of ordinal length,
there is a design choice concerning the exclusion of jumps at limit ordinals,
as illustrated in the ill-formed rewrite sequence 
  \begin{align*}
    \underbrace{\fun{a} \to \fun{a} \to \fun{a} \to \cdots}_{\text{$\omega$-many steps}} \;\fun{b} \to \fun{b}
  \end{align*}
where the rewrite system is $\atrs = \{\,\fun{a} \to \fun{a},\, \fun{b}\to\fun{b}\,\}$.
The rewrite sequence remains for $\omega$ steps at $\fun{a}$ and in the limit step `jumps' to $\fun{b}$.
To ensure connectedness at limit ordinals, the usual choices are:
\begin{enumerate}
  \item \emph{weak convergence} (also called `Cauchy convergence'), where it suffices that the sequence of terms converges towards the limit term, and
  \item \emph{strong convergence}, which additionally requires that the `rewriting activity', i.e., the depth of the rewrite steps,
    tends to infinity when approaching the limit.
\end{enumerate}
The notion of strong convergence incorporates the flavor of `progress', or `productivity',
in the sense that there is only a finite number of rewrite steps at every depth.
Moreover, it leads to a more satisfactory metatheory where redex occurrences can be 
traced over limit steps.

While infinitary rewriting has been studied extensively,
notions of infinitary equational reasoning have not received much attention.
Some of the few works in this area are by Kahrs~\cite{kahr:2013} and by Lombardi, R\'{\i}os and de~Vrijer~\cite{lomb:rios:vrij:2014},
see \textit{Related Work} below.
The reason is that the usual definition of infinitary rewriting is
based on ordinals to index the rewrite steps, 
and hence the rewrite direction is incorporated from the start.
This is different for the framework we propose here, 
which enables us to define several natural notions:
infinitary equational reasoning, bi-infinite rewriting, and the standard concept of infinitary rewriting.
All of these have strong convergence `built-in'.

We define \emph{infinitary equational reasoning} 
with respect to a system of equations $\aes$, as a relation~${\ieq}$
on potentially infinite terms %
by the following 
mutually coinductive rules:
\begin{gather}
  \begin{aligned}
    \infer=
    {s \ieq t}
    {s \mathrel{(=_\aes \cup \ieqdown)^*} t}
    &&\qquad\qquad\qquad&&
    \infer=
    {f(s_1,s_2,\ldots,s_n) \ieqdown f(t_1,t_2,\ldots,t_n)}
    {s_1 \ieq t_1 & \cdots & s_n \ieq t_n}
  \end{aligned}
  \label{rules:intro:ieq}
\end{gather}
The relation ${\ieqdown}$ stands for infinitary equational reasoning below the root.
The coinductive nature of the rules means that the proof trees
need not be well-founded.
Reading the rules bottom-up,
the first rule allows for an arbitrary, but finite, number of rewrite steps at any finite depth (of the term tree).
The second rule enforces that we eventually proceed with the arguments, and hence the activity tends to infinity.

\begin{figure}[h!]
  \begin{framed}
  \vspace{1.5ex}
  \begin{align*}
    \infer=
    {\fun{C}^\omega \ieq \fun{a}}
    {
      \infer=
      {\fun{C}^\omega \ieqdown \fun{C}(\fun{a})}
      {\infer={\fun{C}^\omega \ieq \fun{a}}
          {\makebox(0,0){
            \hspace{-16mm}\begin{tikzpicture}[baseline=15ex,xscale=1.3,yscale=1.4]
            \draw [->,thick,dotted] (0,0) -- (0,1mm) to[out=90,in=80] (-11mm,-1mm) to[out=-100,in=-160,looseness=1.6] (-0mm,-13mm);
            \end{tikzpicture}
          }}
      }
      &
      \fun{C}(\fun{a}) =_\aes \fun{a}
    }
  \end{align*}
  \vspace{-2ex}
  \end{framed}
  \caption{Derivation of $\fun{C}^\omega \ieq \fun{a}$.}
  \label{fig:ieq}
\end{figure}

\begin{exa}\label{ex:Ca:a}
  Let $\aes$ consist of the equation 
  \begin{align*}
    \fun{C}(\fun{a}) = \fun{a} \;.
  \end{align*}
  We write $\fun{C}^\omega$ to denote the infinite term $\fun{C}(\fun{C}(\fun{C}(\ldots)))$,
  the solution of the equation $X = \fun{C}(X)$.
  Using the rules~\eqref{rules:intro:ieq}, we can derive 
  $\fun{C}^\omega \ieq \fun{a}$ as shown in Figure~\ref{fig:ieq}.
  This is an infinite proof tree as indicated by the loop 
  \raisebox{.5ex}{\tikz \draw [->,thick,dotted] (0,0) -- (5mm,0mm);}
  in which the sequence 
  $\fun{C}^\omega \ieqdown \fun{C}(\fun{a}) =_\aes \fun{a}$
  is written 
  by juxtaposing
  $\fun{C}^\omega \ieqdown \fun{C}(\fun{a})$ and $\fun{C}(\fun{a}) =_\aes \fun{a}$.
\end{exa}

Many of the proof trees we consider in this paper are regular trees,
that is, trees having only a finite number of distinct subtrees.
These infinite trees are convenient since they can be depicted
by a `finite tree' enriched with loops \raisebox{.5ex}{\tikz \draw [->,thick,dotted] (0,0) -- (5mm,0mm);}.
However, we emphasise that our framework is not limited to regular trees.

\begin{exa}\label{ex:irregular}
  For an example involving non-regular proof trees, let $\aes$ consist of the equation
  \begin{align*}
    \fun{b}(x) &= \fun{C}(\fun{b}(\fun{S}(x))) \;.
  \end{align*}
  Then we can derive $\fun{b}(x) \ieq \fun{C}^\omega$ by the non-regular proof tree shown in Figure~\ref{fig:nonreg}.
\end{exa}

\begin{figure}[h!]
  \begin{framed}
  \vspace{1.5ex}
  \begin{align*}
    \infer=
    {\fun{b}(x) = \fun{C}^\omega}
    {
      \fun{b}(x) =_\aes \fun{C}(\fun{b}(\fun{S}(x)))
      &
      \infer=
      {\fun{C}(\fun{b}(\fun{S}(x))) \ieqdown \fun{C}^\omega}
      {
        \infer=
        {\fun{b}(\fun{S}(x)) \ieq \fun{C}^\omega}
        {
          \fun{b}(\fun{S}(x)) =_\aes \fun{C}(\fun{b}(\fun{S}(\fun{S}(x))))
          &
          \infer=
          {\fun{C}(\fun{b}(\fun{S}(\fun{S}(x)))) \ieqdown \fun{C}^\omega}
          {
            \infer
            {\fun{b}(\fun{S}(\fun{S}(x))) \ieq \fun{C}^\omega}
            {\vdots}
          }
        }
      }
    }
  \end{align*}
  \vspace{-2ex}
  \end{framed}
  \caption{Non-regular proof tree for $\fun{b}(x) = \fun{C}^\omega$.}
  \label{fig:nonreg}
\end{figure}

\noindent Using the greatest fixed-point constructor~$\sgfp$, we can define $\ieq$ equivalently as follows:
\begin{align}
  {\ieq} &\;\;\defd\;\; \gfp{R}{(=_\aes \cup \mathrel{\down{R}})^*} \,,
  \label{eq:ieq}
\end{align}
where $\down{R}$, corresponding to the second rule in~\eqref{rules:intro:ieq}, is defined by
\begin{align}
  \down{R} \;\defd\; \{\,\pair{f(s_1,\ldots,s_n)}{\,f(t_1,\ldots,t_n)} \mid f \in \Sigma,\; s_1 \mathrel{R} t_1,\,\ldots,s_n \mathrel{R} t_n\,\} \,\cup\, \id  \,.
\end{align}
This is a new and interesting notion of infinitary (strongly convergent) equational reasoning.

%

Now let $\atrs$ be a term rewriting system (TRS).
If we use $\to_\atrs$ instead of $=_\aes$ in the rules~\eqref{rules:intro:ieq}, %
we obtain what we call \emph{bi-infinite rewriting $\ibi$}\,:
\begin{gather}
  \begin{aligned}
    \infer=
    {s \ibi t}
    {s \mathrel{(\to_\atrs \cup \ibidown)^*} t}
    &&\qquad\qquad\qquad&&
    \infer=
    {f(s_1,s_2,\ldots,s_n) \ibidown f(t_1,t_2,\ldots,t_n)}
    {s_1 \ibi t_1 & \cdots & s_n \ibi t_n}
  \end{aligned}
  \label{rules:intro:ibi}
\end{gather}
corresponding to the following fixed-point definition:
\begin{align}
  {\ibi} &\;\;\defd\;\; \gfp{R}{(\to_\atrs \cup \mathrel{\down{R}})^*} \,.
  \label{eq:intro:ibi}
\end{align}
We write $\ibi$ to distinguish bi-infinite rewriting 
from the standard notion $\ired$ of (strongly convergent) infinitary rewriting~\cite{tere:2003}.
The symbol $\infty$ is centered above $\to$ in $\ibi$ 
to indicate that bi-infinite rewriting is `balanced', 
in the sense that it allows rewrite sequences to be extended infinitely forwards, but also infinitely backwards.
Here backwards does \emph{not} refer to reversing the arrow $\leftarrow_{\atrs}$.
For example, for $\atrs = \{\, \fun{C}(\fun{a}) \to \fun{a} \,\}$
we have the backward-infinite rewrite sequence $\cdots \to \fun{C}(\fun{C}(\fun{a})) \to \fun{C}(\fun{a}) \to \fun{a}$
and hence $\fun{C}^\omega \ibi \fun{a}$.
The proof tree for $\fun{C}^\omega \ibi \fun{a}$
has the same shape as the proof tree displayed in Figure~\ref{fig:ieq};
the only difference is that $\ieq$ is replaced by $\ibi$ and $\ieqdown$ by $\ibidown$.
In contrast, the standard notion $\ired$ of infinitary rewriting only takes into account forward limits
and we do \emph{not} have $\fun{C}^\omega \ired \fun{a}$.

We have the following strict inclusions:
\begin{align}
  \label{eq:inclusion1}
  {\ired} \;\;\subsetneq\;\; {\ibi} \;\;\subsetneq\;\; {\ieq} \;\,
\end{align}
In our framework, these inclusions follow directly from the fact 
that the proof trees for $\ired$ (see below)
are a restriction of the proof trees for $\ibi$
which in turn are a restriction of the proof trees for $\ieq$. 
It is also easy to see that each inclusion is strict.
For the first, see above. For the second, just note that $\ibi$ is not symmetric.

Finally, by a further restriction of the proof trees,
we obtain the standard concept of (strongly convergent) infinitary rewriting $\ired$. 
Using least and greatest fixed-point operators, we define:
\begin{align}
  {\ired} &\;\;\defd\;\; \lfp{R}{\gfp{S}{(\to \cup \mathrel{\down{R}})^*\relcomp \down{S}}} \,,
  \label{eq:intro:ired}
\end{align}
where $\relcomp$ denotes relational composition in diagrammatic order, that is:
\begin{align*}
  x \mathrel{(R \relcomp S)} y \iff \exists z.\; x \mathrel{R} z \wedge z \mathrel{S} y\;.
\end{align*}
The greatest fixed point defined using the variable $S$ is a coinductively defined relation. Thus only the last step in the sequence $(\to \cup \mathrel{\down{R}})^*\relcomp \down{S}$ is coinductive.
This corresponds to the following fact about reductions $\sigma$ of ordinal length:
every strict prefix of $\sigma$ must be shorter than $\sigma$ itself,
while strict suffixes may have the same length as $\sigma$.

If we replace $\slfp$ by $\sgfp$ in \eqref{eq:intro:ired}, 
we get a definition equivalent to~$\ibi$ defined by~\eqref{eq:intro:ibi}.
To see that it is at least as strong, note that $\id \subseteq \down{S}$.

Conversely, $\ired$ can be obtained by a restriction of the proof trees obtained by the rules~\eqref{rules:intro:ibi} for $\ibi$.
Assume that in a proof tree using the rules~\eqref{rules:intro:ibi},
we mark those occurrences of $\ibidown$ that 
are followed by another step in the premise of the rule 
(i.e., those that are not the last step in the premise).
Thus we split $\ibidown$ into $\ireddown$ and~$\ireddownfin$.
Then the restriction to obtain the relation $\ired$ is to forbid infinite nesting of marked symbols~$\ireddownfin$.
This marking is made precise in the following rules:
\begin{align}
  \begin{aligned}
  \infer=
  {s \ired t}
  {s \mathrel{(\to \cup \ireddownfin)^*} \relcomp \ireddown t}
  &\hspace{.4cm}&  
  \infer=
  {f(s_1,s_2,\ldots,s_n) \ireddownmfin f(t_1,t_2,\ldots,t_n)}
  {s_1 \ired t_1 & \cdots & s_n \ired t_n}
  &\hspace{.4cm}&
  \infer=
  {s \ireddownmfin s}
  {}
  \end{aligned}
  \label{rules:restrict}
\end{align}
Here $\ireddown$ stands for infinitary rewriting below the root,
and $\ireddownfin$ is its marked version.
The symbol $\ireddownmfin$ stands for both $\ireddown$ and $\ireddownfin$.
Correspondingly, the rule in the middle is an abbreviation for two rules.
The axiom ${s \ireddown s}$ serves to `restore' reflexivity, that is,
it models the identity steps in $\down{S}$ in~\eqref{eq:intro:ired}.
Intuitively, $s \ireddownfin t$ can be thought of as an infinitary rewrite sequence 
below the root, shorter than the sequence we are defining.

We have an infinitary strongly convergent rewrite sequence from $s$ to $t$ 
if and only if $s \ired t$ can be derived by the rules~\eqref{rules:restrict}
in a (not necessarily well-founded) proof tree without infinite nesting of $\ireddownfin$,
that is, proof trees in which all paths (ascending through the proof tree) contain only
finitely many occurrences of $\ireddownfin$.
The depth requirement in the definition of strong convergence
arises naturally in the rules~\eqref{rules:restrict}, 
in particular the middle rule 
pushes the activity to the arguments.
The fact that the rules~\eqref{rules:restrict} 
capture the infinitary rewriting relation $\ired$
is a consequence of a result due to~\cite{kenn:klop:slee:vrie:1995a}
which states that every strongly convergent rewrite sequence
contains only a finite number of steps at any depth $d \in \nat$, 
in particular only a finite number of root steps~$\rstep$.
Hence every strongly convergent reduction is of the form ${(\ireddownfin \relcomp \rstep)^*} \relcomp \ireddown$
as in the premise of the first rule,
where the steps $\ireddownfin$ are reductions of shorter length.

We conclude with an example of a TRS that allows for a rewrite sequence of length beyond $\omega$.

\begin{figure}[h!]
  \begin{framed}
  \begin{align*}
    \infer=
    {\fun{a} \ired \fun{C}^\omega}
    {
      \fun{a} \rstep \fun{C}(\fun{a})
      &
      \infer=
      {\fun{C}(\fun{a}) \ireddown \fun{C}^\omega}
      {\infer={\fun{a} \ired \fun{C}^\omega}
          {\makebox(0,0){
            \hspace{17mm}\begin{tikzpicture}[baseline=13.5ex,xscale=1.3,yscale=1.3]
            \draw [->,thick,dotted] (0,0) -- (0,1mm) to[out=90,in=100] (11mm,-1mm) to[out=-80,in=-20,looseness=1.4] (-0mm,-13mm);
            \end{tikzpicture}
          }}
      }
    }
  \end{align*}
  \end{framed}
  \caption{A reduction $\fun{a} \ired \fun{C}^\omega$.}
  \label{fig:aComega}
\end{figure}

\begin{figure}[h!]
  \begin{framed}
  \begin{align*}
    \infer=
    {\fun{f}(\fun{a},\fun{b}) \ired \fun{D}}
    {
      \infer=
      {\fun{f}(\fun{a},\fun{b}) \ireddownfin \fun{f}(\fun{C}^\omega,\fun{C}^\omega)}
      {
        \infer=
        {\fun{a} \ired \fun{C}^\omega}
        {\text{like Figure~\ref{fig:aComega}}}
        &
        \infer=
        {\fun{b} \ired \fun{C}^\omega}
        {\text{like Figure~\ref{fig:aComega}}}
      }
      & 
      \fun{f}(\fun{C}^\omega,\fun{C}^\omega) \rstep \fun{D}
    }
  \end{align*}
  \end{framed}
  \caption{A reduction $\fun{f}(\fun{a},\fun{b}) \ired \fun{D}$.}
  \label{fig:fab}
\end{figure}
  
\begin{exa}\label{ex:fab}
  We consider the term rewriting system from~\cite{ders:kapl:plai:1991} with the following rules:
  \begin{align*}
    \fun{f}(x,x) &\to \fun{D} & 
    \fun{a} &\to \fun{C}(\fun{a}) &
    \fun{b} &\to \fun{C}(\fun{b}) \,.
  \end{align*}
  We then have $\fun{a} \ired \fun{C}^\omega$, that is, an infinite reduction from $\fun{a}$ to $\fun{C}^\omega$ in the limit:
  \begin{align*}
    \fun{a} \to \fun{C}(\fun{a}) \to \fun{C}(\fun{C}(\fun{a})) \to \fun{C}(\fun{C}(\fun{C}(\fun{a}))) \to \cdots 
    \to^\omega \fun{C}^\omega \,.
  \end{align*}
  Using the proof rules~\eqref{rules:restrict}, we can derive $\fun{a} \ired \fun{C}^\omega$ as shown in Figure~\ref{fig:aComega}.
  The proof tree in Figure~\ref{fig:aComega} can be described as follows:
  We have an infinitary rewrite sequence from $\fun{a}$ to~$\fun{C}^\omega$
  since we have a root step from $\fun{a}$ to $\fun{C}(\fun{a})$, and
  an infinitary reduction below the root from $\fun{C}(\fun{a})$ to $\fun{C^\omega}$.
  The latter reduction $\fun{C}(\fun{a}) \ireddown \fun{C^\omega}$ is in turn witnessed
  by the infinitary rewrite sequence $\fun{a} \ired \fun{C}^\omega$ 
  on the direct subterms.

  We also have the following reduction, now of length $\omega+1$:
  \begin{align*}
    \fun{f}(\fun{a},\fun{b}) 
    \to \fun{f}(\fun{C}(\fun{a}),\fun{b}) 
    \to \fun{f}(\fun{C}(\fun{a}),\fun{C}(\fun{b}))
    \to \cdots \to^\omega \fun{f}(\fun{C}^\omega,\fun{C}^\omega)
    \to \fun{D} \,.
  \end{align*}
  That is, after an infinite rewrite sequence of length $\omega$, 
  we reach the limit term $\fun{f}(\fun{C}^\omega,\fun{C}^\omega)$,
  and we then continue with a rewrite step from $\fun{f}(\fun{C}^\omega,\fun{C}^\omega)$ to $\fun{D}$.
  Figure~\ref{fig:fab} shows how this rewrite sequence \mbox{$\fun{f}(\fun{a},\fun{b}) \ired \fun{D}$}
  can be derived in our setup.
  We note that the rewrite sequence $\fun{f}(\fun{a},\fun{b}) \ired \fun{D}$
  cannot be `compressed' to length $\omega$. 
  So there is no reduction $\fun{f}(\fun{a},\fun{b}) \to^{\le \omega} \fun{D}$.
\end{exa}

\subsubsection*{Related Work}
While 
a coinductive treatment of infinitary rewriting is not new~\cite{coqu:1996,joac:2004,endr:polo:2012b},
the previous approaches 
only capture rewrite sequences of length at most $\omega$. 
The coinductive framework that we present here captures all strongly
convergent rewrite sequences of arbitrary ordinal length.

From the topological perspective, various notions of infinitary rewriting 
and infinitary equational reasoning have been studied in~\cite{kahr:2013}.
%
The closure operator $S_E$ from \cite{kahr:2013} is closely related to our notion of infinitary equational reasoning $\ieq$.
The operator $S_E$ is defined by $S_E(R) = (S \funcomp E)^{\star}(R)$ where 
\begin{enumerate}
 \item $E(R)$ is the equivalence closure of $R$, and
 \item $S(R)$ is the strongly convergent rewrite relation obtained from (single steps) $R$,
 \item and $f^{\star}(R)$ is defined as $\mu x.\,R\cup f(x)$.
\end{enumerate}
Although defined in very different ways, 
the relations $S_E(\to)$ and $\ieq$ typically coincide.
%

In~\cite{lomb:rios:vrij:2014}, Lombardi, R\'{\i}os and de~Vrijer introduce
infinitary equational reasoning based on limits 
to reason about permutation equivalence of infinitary reductions that are modelled by proof terms.

Martijn Vermaat has formalized infinitary rewriting using metric convergence (in place of strong convergence) 
in the Coq proof assistant~\cite{verm:2010}, and proved that weakly orthogonal infinitary rewriting 
does not have the property $\mathrm{UN}$ of unique normal forms, see~\cite{endr:hend:grab:klop:oost:2014}.
While his formalization could be extended to strong convergence,
it remains to be investigated to what extent it can be used 
for the further development of the theory of infinitary rewriting.

Ketema and Simonsen~\cite{kete:simo:2013} 
introduce the notion of `computable infinite reductions'~\cite{kete:simo:2013},
where terms as well as reductions are computable,
and provide a Haskell implementation of the Compression Lemma for this notion of reduction.

This current paper is an extended version of~\cite{endr:hans:hend:polo:silv:2015}. 
The most important changes are:%
\begin{enumerate}
  \item 
    We have introduced a novel notion of permutation equivalence on infinitary rewrite sequences,
    which we call \emph{parallel permutation equivalence}. 
    We show that two rewrite sequences are parallel permutation equivalent if and only if
    they are represented by the same proof tree in our framework, see Section~\ref{sec:permute}.
  \item 
    We have rewritten and extended the description of the Coq formalisation of the Compression Lemma (Section~\ref{sec:coq}).
\end{enumerate}

\subsubsection*{Outline}
In Section~\ref{sec:prelims} we introduce infinitary rewriting in the usual way
based on ordinals,
and with convergence at every limit ordinal.
Section~\ref{sec:coinduction} is a short explanation of (co)induction and fixed-point rules.
The two new definitions of infinitary rewriting~$\ired$ based on
mixing induction and coinduction, as well as their equivalence, are spelled out in Section~\ref{sec:itrs}.
Then, in Section~\ref{sec:equivalence},
we prove the equivalence of these new definitions of infinitary rewriting
with the standard definition.
In Section~\ref{sec:ieq} we present the above introduced relations~$\ieq$ and~$\ibi$
of infinitary equational reasoning and bi-infinite rewriting.
In Section~\ref{sec:venn} we compare the three relations $\ieq$, $\ibi$ and $\ired$.
In Section~\ref{sec:permute} we present our new work on parallel permutation equivalence and canonical proof trees.
As an application, we show in Section~\ref{sec:coq} that our framework
is suitable for formalizations in theorem provers.
We conclude in Section~\ref{sec:conclusion}.

\section{Preliminaries on Term Rewriting}\label{sec:prelims}

We give a brief introduction to infinitary rewriting.
For further reading on infinitary rewriting we refer to
\cite{klop:vrij:2005,tere:2003,bare:klop:2009,endr:hend:klop:2012},
for an introduction to finitary rewriting to
\cite{klop:1992,tere:2003,baad:nipk:1998,bare:1977}.

A \emph{signature $\Sigma$} is a set of symbols $f$ each having a fixed arity $\arity{f} \in \nat$.
Let $\avars$ be an infinite set of variables such that $\avars \cap \Sigma = \emptyset$.
The set $\iter{\Sigma}{\avars}$ of (finite and) \emph{infinite terms} 
over $\Sigma$ and $\avars$ is coinductively defined by the following grammar:
\begin{align*}
  t \coBNFis x \BNFor f(\underbrace{t,\ldots,t}_{\text{$\arity{f}$ times}}) \; \text{($x \in \avars$, $f \in \Sigma$)} \,.
\end{align*}
This means that $\iter{\Sigma}{\avars}$ is defined as the largest set $T$ such that 
for all $t \in T$, either $t \in \avars$ or $t = f(t_1,t_2,\ldots,t_n)$ for some $f \in \Sigma$ with $\arity{f} = n$
and $t_1,t_2,\ldots,t_n \in T$.
So the grammar rules may be applied an infinite number of times, and equality on the terms is bisimilarity.
See further Section~\ref{sec:coinduction} for a brief introduction to coinduction.

We write $\id$ for the identity relation on terms, $\id \defd \{\pair{s}{s} \mid s \in \iter{\Sigma}{\avars}\}$.

\begin{rem}
  Alternatively, the set $\iter{\Sigma}{\avars}$ arises from the set of finite terms, $\ter{\Sigma}{\avars}$, 
  by metric completion,
  using the well-known distance function $\smetric$ 
  defined by
  $\metric{t}{s} = 2^{-n}$ if the $n$-th level of the terms $t,s \in \ter{\Sigma}{\avars}$ (viewed as labeled trees) 
  is the first level where a difference appears, 
  in case $t$ and $s$ are not identical; furthermore, $\metric{t}{t} = 0$.
  It is standard that this construction yields $\pair{\ter{\Sigma}{\avars}}{\smetric}$ as a metric space. 
  Now, infinite terms are obtained by taking the completion of this metric space, 
  and they are represented by infinite trees. 
  We will refer to the complete metric space arising in this way as $\pair{\iter{\Sigma}{\avars}}{\smetric}$, 
  where $\iter{\Sigma}{\avars}$ is the set of finite and infinite terms over~$\Sigma$.
\end{rem}

Let $t \in \iter{\Sigma}{\avars}$ be a finite or infinite term.
The set of \emph{positions $\pos{t}\subseteq \nat^*$ of $t$} is
defined by: $\posemp \in \pos{t}$, and 
$i p \in \pos{t}$ whenever $t = f(t_1,\ldots,t_n)$ 
with $1 \le i \le n$ and $p \in \pos{t_i}$.
For $p \in \pos{t}$, the \emph{subterm} $\subtrm{t}{p}$ of $t$ at position $p$
is defined by
$\subtrm{t}{\posemp} = t$ and
$\subtrm{f(t_1,\ldots,t_n)}{ip} = \subtrm{t_i}{p}$.
The set of \emph{variables $\vars{t}\subseteq \avars$ of $t$} is %
$\vars{t} = \{x \in \avars \mid \exists \, p\in \pos{t}.\,\subtrm{t}{p} = x\}$.

A \emph{substitution $\asubst$} is a map $\asubst : \avars \to \iter{\Sigma}{\avars}$;
its domain 
is extended to
$\iter{\Sigma}{\avars}$ essentially by corecursion: %
$\asubst(f(t_1,\ldots,t_n)) = f(\asubst(t_1),\ldots,\asubst(t_n))$
(cf.~\cite[Example~2.5(iv) and Remark 3.2]{Milius:CIA}).
For a term~$t$ and a substitution~$\asubst$, we write $t\sigma$ for $\sigma(t)$.
We write $x \mapsto s$ for the substitution defined by
$\asubst(x) = s$ and $\asubst(y) = y$ for all $y \ne x$.
Let $\hole$ be a fresh variable.
A \emph{context} $C$ is a term $\iter{\Sigma}{\avars \cup \{\hole\}}$
containing precisely one occurrence of %
$\hole$.
For contexts $C$ and terms $s$ we write $C[s]$ for $C(\hole \mapsto s)$.

A \emph{rewrite rule $\ell \to r$} over $\Sigma$ and $\avars$ is a pair 
$(\ell,r)$
of terms $\ell,r\in\iter{\Sigma}{\avars}$ such that the left-hand side $\ell$ is not a variable ($\ell \not\in \avars$),
and all variables in the right-hand side $r$ occur in $\ell$, $\vars{r} \subseteq \vars{\ell}$.
Note that we require neither the left-hand side nor the right-hand side of a rule
to be finite.

A \emph{term rewriting system (TRS) $\atrs$} over $\Sigma$ and $\avars$
is a set of rewrite rules over $\Sigma$ and $\avars$.
A TRS induces a rewrite relation on the set of terms as follows.
For $\apos \in \nat^\ast$ we define ${\rerat{\atrs}{\apos}} \subseteq \iter{\Sigma}{\avars} \times \iter{\Sigma}{\avars}$, 
a \emph{rewrite step at position $\apos$}, by
$
  C[\ell\sigma] \rerat{\atrs}{\apos} C[r\sigma]
$
if $C$ is a context with $\subtrm{C}{\apos} = \hole$,\; $\ell \to r \in \atrs$, and $\sigma : \avars \to \iter{\Sigma}{\avars}$.
We write $\rstep$ for \emph{root steps}, %
${\rstep} = \{\,(\ell\sigma,r\sigma) \mid \ell \to r \in \atrs,\; \text{$\sigma$ a substitution}\,\}$.
We write $s \to_{\atrs} t$ if $s \rerat{\atrs}{\apos} t$ for some $\apos\in\nat^\ast$.
A \emph{normal form} is a term without a redex occurrence,
that is, a term that is not of the form $C[\ell\sigma]$ for some context $C$, 
rule $\ell \to r\in \atrs$ and substitution $\sigma$.

A natural consequence of this construction is %
the notion of \emph{weak convergence}: 
we say that $t_0 \to t_1 \to t_2 \to \cdots$ is an infinite reduction sequence with limit $t$, 
if $t$ is the limit of the sequence $t_0,t_1,t_2, \ldots$ in the usual sense of metric convergence. 
In contrast, the central notion of \emph{strong convergence}
requires, in addition to weak convergence, that the depth of the redexes contracted in
successive steps tends to infinity when approaching a
limit ordinal from below.
This condition rules out the possibility that the action of redex contraction stays confined at the top, 
or stagnates at some finite level of depth. 

\begin{defi}\label{def:itrs:standard}
  A \emph{transfinite rewrite sequence} (of ordinal length $\alpha$)
  consists of an initial term $t_0$ and a sequence of rewrite steps 
  $(t_{\beta} \rerat{\atrs}{\apos_{\beta}} t_{\beta+1})_{\beta < \alpha}$
  such that for every limit ordinal $\lambda < \alpha$ we have that 
  if $\beta$ approaches $\lambda$ from below, then
  \begin{enumerate}[label=(\roman*)]
    \item\label{item:distance}
      the distance $\metric{t_\beta}{t_{\lambda}}$ tends to $0$ 
      and, moreover,
    \item\label{item:depth}
      the depth of the rewrite action, i.e., 
      the length of the position $\apos_\beta$, 
      tends to infinity.
  \end{enumerate}
  The sequence is called \emph{strongly convergent} 
  if $\alpha$ is a successor ordinal, 
  or there exists a term $t_\alpha$ such that
  the conditions~\ref{item:distance} and~\ref{item:depth}
  are fulfilled for every limit ordinal $\lambda \leq \alpha$;
  we then write $t_0\iredord t_\alpha$.
  The subscript $\mit{ord}$ is used in order to distinguish 
  $\iredord$ from the equivalent %
  relation $\ired$ 
  as defined in Definition~\ref{def:ired:fixedpoint}. 
  We sometimes write $t_0\redord^{\alpha} t_\alpha$ 
  to explicitly indicate the length $\alpha$ of the sequence.
  The sequence is called \emph{divergent} if it is not strongly convergent.
\end{defi}

There are several reasons why strong convergence is beneficial; 
the foremost being that in this way we can define the notion of \emph{descendant} 
(also \emph{residual}) over limit ordinals. 
Also the well-known Parallel Moves Lemma
and the Compression Lemma
fail for weak convergence, see~\cite{simo:2004} and \cite{ders:kapl:plai:1991} respectively.

\section{(Co)induction, Fixed Points and Relations}\label{sec:coinduction}

We briefly introduce the relevant concepts from
(co)algebra and (co)induction that will be used later throughout this
paper. 
For a more thorough introduction, we refer to \cite{jaco:rutt:2011}.
There will be two main points where coinduction will play a role, in the definition of terms and in the definition 
of term rewriting. 

Terms are usually defined with respect to a type constructor~$F$. 
For instance, consider the type of lists with elements in a given set $A$, given in a functional programming style:
\begin{verbatim}
  type List a = Nil | Cons a (List a)
\end{verbatim}
The above grammar corresponds to the type constructor
$F(X) = 1 + A \times X$ where the $1$ is used as a placeholder for the empty list {\tt Nil} 
and the second component represents the {\tt Cons} constructor.
Such a grammar can be interpreted in two ways: 
The \emph{inductive} interpretation yields as terms the set of finite lists,
and corresponds to the \emph{least fixed point} of~$F$.
The \emph{coinductive} interpretation yields as terms
the set of all finite or infinite lists,
and corresponds to the \emph{greatest fixed point} of~$F$.
More generally, the inductive interpretation of a type constructor 
yields closed finite terms (with well-founded syntax trees), and
dually, the coinductive interpretation 
yields closed possibly infinite terms.
For readers familiar with the categorical definitions of algebras
and coalgebras, these two interpretations amount to defining
closed finite terms as the \emph{initial $F$-algebra}, 
and closed possibly infinite terms as the \emph{final $F$-coalgebra}.

In order to formally define finite and infinite terms
over a signature $\Sigma$ and a set of variables $\avars$, consider
the associated type constructor $G_{\Sigma,\avars}(Y) = X + F_\Sigma(Y)$
where $F_\Sigma(Y) = \{ f(y_1, \ldots, y_n) \mid f \in \Sigma, y_1,\ldots,y_n \in Y, n=\arity{f}\}$.
Then $\ter{\Sigma}{\avars}$ is the least fixed point of $G_{\Sigma,\avars}$
and $\iter{\Sigma}{\avars}$ is the greatest fixed point of $G_{\Sigma,\avars}$.

Equality on finite terms is the expected syntactic/inductive definition. 
Equality of possibly infinite terms is \emph{bisimilarity}. 
For instance, in the above example, two finite or infinite lists are equal 
if and only if they are related by a {\tt List}-bisimulation, 
which is a relation $R \subseteq $ {\tt List a $\times$ List a} 
such that for all pairs in $R$ are of the form
\begin{enumerate}
  \item $(\texttt{Nil},\texttt{Nil})$, or
  \item $(\mathtt{Cons}\, a \, \sigma,\mathtt{Cons}\, b \,\tau)$ such that $a = b$ and $(\sigma,\tau) \in R$.
\end{enumerate}\medskip

\noindent Throughout this paper, we define and reason about relations on the set
$T := \iter{\Sigma}{\avars}$ of terms.
Such relations are elements of the powerset of $T \times T$, which we view as a
complete lattice $\tlat := \Pow(T \times T)$ in which joins and meets are given by
unions and intersections of relations.
Relations on terms can thus be defined as least and greatest fixed points
of monotone operators on $\tlat$, using the Knaster-Tarski fixed point theorem.
In $L$, an \emph{inductively defined relation} is 
a least fixed point $\lfp{X}{F(X)}$ of a monotone $F : \tlat \to \tlat$.
Dually, a \emph{coinductively defined relation} is
a greatest fixed point $\gfp{X}{F(X})$ of a monotone $F : \tlat \to \tlat$.
We will make frequent use of the fact
that $\gfp{Y}{F(Y)}$ is the greatest post-fixed point of $F$, that is,
\begin{gather}
\gfp{Y}{F(Y)} = \bigcup\,\{\, X \in L \mid X \subseteq F(X) \,\},
\label{eq:gfp-as-greatest-post-fp}
\end{gather}
and $\lfp{Y}{F(Y)}$ is the least pre-fixed point of $F$, that is,
\begin{gather}
\lfp{Y}{F(Y)} = \bigcap\,\{\, X \in L \mid F(X) \subseteq X \,\}
\label{eq:lfp-as-last-pre-fp}
\end{gather}
The above properties can be expressed as the following fixed point rules:
\begin{gather}
  \begin{aligned}
    \frac{X \subseteq F(X)}{X \subseteq \gfp{Y}{F(Y)}}(\nu\text{-rule}) 
    &&\qquad\qquad\qquad&&
    \frac{F(X) \subseteq X}{\lfp{Y}{F(Y)} \subseteq X}(\mu\text{-rule})
  \end{aligned}
  \label{eq:coind-ind-rules}
\end{gather}
These proof rules, in fact, show the connection to the more abstract categorical notions of induction and coinduction. This can be seen by viewing $L$ as a partial order $(L,\subseteq)$. A partial order $(P, \leq)$ can, in turn, be seen as a category in which the objects are the elements of $P$ and there is a unique arrow $X \to Y$ if $X \leq Y$. A functor on $(P,\leq)$ is then nothing but a monotone map $F$; an $F$-coalgebra $X \to F(X)$ is a post-fixed point of $F$; and a final $F$-coalgebra is a greatest fixed point of $F$. Dually, an $F$-algebra $F(X) \to X$ is a pre-fixed point of $F$, and an initial $F$-algebra is a least fixed point of $F$. The two proof rules express the universal properties of these final and initial objects.

We will use a number of basic operations on relations. These include
union $(\cup)$, reflexive, transitive closure ($^*$),
relation composition in diagrammatic order $(\relcomp)$,
and relation lifting which we define now.
\begin{defi}\label{def:lifting}
  For a relation $R \subseteq T \times T$ %
  we define its \emph{lifting $\down{R}$} (with respect to $\Sigma)$ by
  \begin{align*}
    \down{R} \;\defd\; \{\,\pair{f(s_1,\ldots,s_n)}{\,f(t_1,\ldots,t_n)} \mid f \in \Sigma ,\, \arity{f} = n\,, s_1 \mathrel{R} t_1,\ldots,s_n \mathrel{R} t_n\,\}
                \,\cup\, \id \,.
  \end{align*}
\end{defi}\medskip

\noindent It is straightforward to verify that all these operations are monotone (in all arguments). Hence any map $F\colon L \to L$ built from these operations will have a unique least and greatest fixed point.

\section{New Definitions of Infinitary Term Rewriting}\label{sec:itrs}

We present two new definitions of infinitary rewriting \mbox{$s \ired t$},
based on mixing induction and coinduction, and prove their equivalence.
In Section~\ref{sec:equivalence} we show they are equivalent to the standard definition based on ordinals.
We summarize the definitions:
\begin{enumerate}[label=\emph{\Alph*}.]
  \item \emph{Derivation Rules.}
        First, we define $s \ired t$ via a syntactic restriction on the proof trees 
        that arise from the coinductive rules~\eqref{rules:restrict}.
        The restriction excludes all proof trees that contain ascending paths
        with an infinite number of marked symbols.
  \medskip
  \item \emph{Mixed Induction and Coinduction.}
        Second, we define $s \ired t$ based on mutually mixing induction and coinduction,
        that is, least fixed points $\mu$ and greatest fixed points $\nu$.
\end{enumerate}
\noindent
In contrast to previous coinductive definitions~\cite{coqu:1996,joac:2004,endr:polo:2012b}, 
the setup proposed here captures all strongly convergent rewrite sequences (of arbitrary ordinal length).

Throughout this section, we fix a signature $\Sigma$ and a term
rewriting system $\atrs$ over $\Sigma$. 
We also abbreviate $T \defd \iter{\Sigma}{\avars}$.
\begin{nota}\label{not:transitivity}
  Instead of introducing separate derivation rules for transitivity,
  we write a reduction of the form
  $s_0 \rightsquigarrow s_1 \rightsquigarrow \cdots \rightsquigarrow s_n$
  as a sequence of single steps: 
  \begin{align*}
    \infer=
    {\text{conclusion}}
    {s_0 \rightsquigarrow s_1 \quad s_1 \rightsquigarrow s_2 \quad\cdots\quad s_{n-1} \rightsquigarrow s_n}
  \end{align*}
  \noindent
  This allows us to write the subproof immediately above a single step.
\end{nota}

\subsection{Derivation Rules}

\begin{defi}\label{def:ired:restrict}
  We define the relation ${\ired} \subset T \times T$ as follows.
  We have $s \ired t$ if there exists a 
  (finite or infinite) proof tree $\delta$ deriving $s \ired t$ using the following five rules:
  \begin{align*}
    \infer=[\rsplit]
    {s \ired t}
    {s \mathrel{(\ireddownfin \cup \rstep)^*} \relcomp \ireddown t}
    &&  
    \infer=[\rlift]
    {f(s_1,s_2,\ldots,s_n) \ireddownmfin f(t_1,t_2,\ldots,t_n)}
    {s_1 \ired t_1 & \cdots & s_n \ired t_n}
    &&
    \infer=[\rid]
    {s \ireddownmfin s}
    {}
  \end{align*}
  such that $\delta$ does not contain an infinite nesting of $\ireddownfin$,
  that is, such that there exists no path ascending 
  through the proof tree that meets an infinite number of symbols $\ireddownfin$. 
  The symbol $\ireddownmfin$ stands for $\ireddown$ or $\ireddownfin$;
  so the second rule is an abbreviation for two rules; similarly for the third rule.
\end{defi}
In the above definition, we tacitly assume that the root steps are derived by 
axioms of the form
\begin{align}
  \infer=[\ell \to r \in \atrs,\; \text{$\sigma$ a substitution}]{
    \ell\sigma \rstep r\sigma
  }{
  } \label{eq:axiom:rstep}
\end{align}
For keeping the proof trees compact, 
we will just write $\ell\sigma \rstep r\sigma$ in the proof trees
not mentioning rule and substitution.

We give some intuition for the rules in Definition~\ref{def:ired:restrict}.
The relations $\ireddownfin$ and $\ireddown$ are infinitary reductions below the root.
We use $\ireddownfin$ for constructing parts of the prefix (between root steps), and
$\ireddown$ for constructing a suffix of the reduction that we are defining.
When thinking of ordinal indexed rewrite sequences $\sigma$, 
a suffix of $\sigma$ can have length equal to $\sigma$,
while the length of every prefix of $\sigma$ must be strictly smaller than the length of $\sigma$.
The five rules (\rsplit, and the two versions of \rlift\ and \rid) can be interpreted as follows:
\begin{enumerate}
  \item 
    The \rsplit-rule:
    the term $s$ rewrites infinitarily to $t$, $s \ired t$, 
    if $s$ rewrites to $t$ using a finite sequence of (a) root steps,
    and (b) infinitary reductions $\ireddown$ below the root
    --- where infinitary reductions preceding root steps must be shorter than the derived reduction.
    \smallskip
  \item 
    The \rlift-rules: 
    the term $s$ rewrites infinitarily to $t$ below the root, $s \ireddownmfin t$, %
    if the terms are of the shape $s = f(s_1,s_2,\ldots,s_n)$ and $t = f(t_1,t_2,\ldots,t_n)$
    and there exist reductions %
    between the arguments:
    $s_1 \ired t_1$, \ldots, $s_n \ired t_n$.
    \smallskip
  \item 
    The \rid-rules allow for the rewrite relations~$\ireddownmfin$ %
    to be reflexive,
    and this in turn yields reflexivity of $\ired$.
    For variable-free terms, reflexivity can already be derived using the other %
    rules.
    For terms with variables, this %
    rule is needed 
    (unless we treat variables as constant symbols).
\end{enumerate}
For an example of a proof tree, 
we refer to Example~\ref{ex:fab} in the introduction. %

\subsection{Mixed Induction and Coinduction}

The next definition is based on mixing induction and coinduction. 
The inductive part is used to model the restriction 
to finite nesting of $\ireddownfin$ in the derivations of Definition~\ref{def:ired:restrict}.
The induction corresponds to a least fixed point $\slfp$,
while a coinductive rule to a greatest fixed point $\sgfp$.

\begin{defi}\label{def:ired:fixedpoint}
  We define the relation ${\ired} \subseteq T \times T$ by
  \begin{align}
    \label{eq:ired:defn}
      {\ired} \;\;\defd\;\; \lfp{R}{\gfp{S}{(\rstep \cup \mathrel{\down{R}})^*\relcomp \down{S}}} \,.
  \end{align}
\end{defi}

We argue why ${\ired}$ is well-defined.
Let $L \defd \Pow(T\times T)$ be the set of all relations on terms.
Define functions $G : L \times L \to L$ and $F : L \to L$ by
\begin{align}
  G(R,S) \defd (\rstep \cup \mathrel{\down{R}})^*\relcomp \down{S}
    \quad\text{ and }\quad
  F(R) \defd \gfp{S}{G(R,S)} = \gfp{S}{(\rstep \cup \mathrel{\down{R}})^*\relcomp \down{S}} \,.
  \label{eq:G:F}
\end{align}
It can easily be verified that $F$ and $G$ are monotone, in all their arguments,
with respect to set-theoretic inclusion.
Hence $F$ and $G$ have unique least and greatest fixed points.

In particular, the relation ${\ired}$ given by \eqref{eq:ired:defn} is well-defined.
\subsection{Equivalence}

We show equivalence of Definitions~\ref{def:ired:restrict} and~\ref{def:ired:fixedpoint}.
Intuitively, the $\slfp R$ in the fixed point definition 
corresponds to the nesting restriction in the definition using derivation rules. 
If one thinks of Definition~\ref{def:ired:fixedpoint} as $\lfp{R}{F(R)}$
with $F(R) = \gfp{S}{G(R,S)}$
(see equation~\eqref{eq:G:F}), then $F^{n+1}(\emptyset)$ 
are all infinite rewrite sequences that can be derived 
using proof trees where the nesting depth of the marked symbol $\ireddownfin$ is at most $n$. 

To avoid confusion we write $\ired_{\nest}$ for the relation $\ired$ defined in Definition~\ref{def:ired:restrict},
and $\ired_{\fp}$ for the relation $\ired$ defined in Definition~\ref{def:ired:fixedpoint}.
We show ${\ired_{\nest}} = {\ired_{\fp}}$.
Definition~\ref{def:ired:restrict} requires that the nesting structure of $\ireddownfin_\nest$
in proof trees is well-founded. As a consequence, we can associate to every proof tree
a (countable) ordinal that allows to embed the nesting structure in an order-preserving way.
We use $\omega_1$ to denote the first uncountable ordinal,
and we view ordinals as the set of all smaller ordinals
(then the elements of $\omega_1$ are all countable ordinals).

\begin{defi}
  Let $\delta$ be a proof tree as in Definition~\ref{def:ired:restrict},
  and let $\alpha$ be an ordinal.
  An \emph{$\alpha$-labeling of $\delta$} 
  is a labeling of all symbols $\ireddownfin_\nest$ in $\delta$ with elements from $\alpha$
  such that
  each label is strictly greater than all labels occurring in the subtrees (all labels above).
\end{defi}

\begin{lem}\label{lem:nest}
  Every proof tree as in Definition~\ref{def:ired:restrict}
  has an $\alpha$-labeling for some $\alpha \in \omega_1$.
\end{lem}

\begin{proof}
  Let $\delta$ be a proof tree and 
  let $L(\delta)$ be the set of positions of the symbol
  $\ireddownfin_\nest$ in $\delta$.
  For positions $p,q \in L(\delta)$ we write $p < q$ if $p$ is a strict prefix of $q$.
  Then we have that $<^{-1}$ is well-founded, that is, 
  there is no infinite sequence $p_0 < p_1 < p_2 < \cdots$ with $p_i \in L(\delta)$ ($i \ge 0$)
  as a consequence of the nesting restriction on
  $\ireddownfin_\nest$. 

  By transfinite recursion, the well-founded order on $L(\delta)$
  extends to a well-order, isomorphic to some ordinal $\alpha$ --- 
  and $\alpha < \omega_1$ since $L(\delta)$ is a countable set.
\end{proof} 

\begin{defi}
  Let $\delta$ be a proof tree as in Definition~\ref{def:ired:restrict}.
  We define the \emph{nesting depth} of $\delta$ as 
  the least ordinal $\alpha \in \omega_1$ such that $\delta$ admits an $\alpha$-labeling.
  For every $\alpha \le \omega_1$, we define a relation
  ${\ired_{\alpha,\nest}}  \subseteq {\ired_{\nest}}$
  as follows:
  $s \ired_{\alpha,\nest} t$ whenever $s \ired_\nest t$
  can be derived using a proof with nesting depth $< \alpha$.
  Likewise we define relations
  ${\ireddown_{\alpha,\nest}}$ and
  ${\ireddownfin_{\alpha,\nest}}$\,.
\end{defi}

As a direct consequence of Lemma~\ref{lem:nest} we have:
\begin{cor}\label{cor:nest:omega1}
  We have ${\ired_{\omega_1,\nest}} = {\ired_\nest}$.
\end{cor}

\begin{thm}\label{thm:equiv:der:fp}
  Definitions~\ref{def:ired:restrict} and~\ref{def:ired:fixedpoint} 
  define the same relation, ${\ired_{\nest}} = {\ired_{\fp}}$. %
\end{thm}

\begin{proof}
  We begin with ${\ired_{\fp}} \subseteq {\ired_{\nest}}$.
  Recall that $F(\ired_\nest)$ is the greatest fixed point of $G(\ired_\nest,\_)$, see~\eqref{eq:G:F}.
  Also, we have ${\ireddown_{\nest}} = {\ireddownfin_{\nest}} = \down{\ired_{\nest}\vphantom{i}}$, 
  and hence
  \begin{align}
    F({\ired_{\nest}}) &= (\rstep \cup \mathrel{\down{\ired_{\nest}\vphantom{i}}})^* \relcomp \down{F({\ired_{\nest}})} 
    = (\rstep \cup \mathrel{\ireddownfin_{\nest}})^* \relcomp \down{F({\ired_{\nest}})}\\
    \down{F({\ired_{\nest}})} &= \id \cup \{\,\pair{f(\vec{s})}{\,f(\vec{t})} \mid \vec{s} \,\mathrel{F(\ired_{\nest})}\, \vec{t}\,\}
    \label{eq:id:or:not}
  \end{align}
  where $\vec{s}$, $\vec{t}$ abbreviate $s_1,\ldots,s_n$ and $t_1,\ldots,t_n$, respectively,
  and we write $\vec{s} \mathrel{R} \vec{t}$
  if we have $s_1 \mathrel{R} t_1,\ldots,s_n \mathrel{R} t_n$.
  Employing the $\mu$-rule from~\eqref{eq:coind-ind-rules},
  it suffices to show that $F({\ired_{\nest}}) \subseteq {\ired_{\nest}}$.
  Assume $\pair{s}{t} \in F({\ired_{\nest}})$.
  Then $\pair{s}{t} \in (\rstep \cup \mathrel{\ireddownfin_{\nest}})^* \relcomp \down{F({\ired_{\nest}})}$.
  Then there exists $s'$ such that $s \mathrel{(\rstep \cup \mathrel{\ireddownfin_{\nest}})^*} s'$
  and $s' \mathrel{\down{F({\ired_{\nest}})}} t$.
  Now we distinguish cases according to~\eqref{eq:id:or:not}:
  \begin{align*}
    \infer=[\rsplit] {s \ired t} {s \mathrel{(\rstep \cup \mathrel{\ireddownfin_{\nest}})^*} t & 
      \infer=[\rid] {t \ireddown t } {} }
    &&
    \infer=[\rsplit] {s \ired t} {s \mathrel{(\rstep \cup \mathrel{\ireddownfin_{\nest}})^*} s' & 
      \infer=[\rlift] {s' \ireddown t } {\delta_1 & \cdots & \delta_n} }
  \end{align*}
  Here, for $i \in \{1,\ldots,n\}$, $\delta_i$ is the proof tree for $s_i \ired t_i$
  obtained from $s_i \mathrel{F(\ired_{\nest})} t_i$ by corecursively applying the same procedure.

  Next we show that ${\ired_{\nest}} \subseteq {\ired_{\fp}}$.
  By Corollary~\ref{cor:nest:omega1} %
  it suffices to show ${\ired_{\omega_1,\nest}} \subseteq {\ired_{\fp}}$.
  We prove by well-founded induction on $\alpha \le \omega_1$ that
  ${\ired_{\alpha,\nest}} \subseteq {\ired_{\fp}}$.
  Since $\ired_{\fp}$ is a fixed point of $F$,
  we obtain ${\ired_{\fp}} = F(\ired_{\fp})$, and since $F(\ired_{\fp})$ 
  is the greatest fixed point of $G(\ired_{\fp},\_)$,
  using the $\nu$-rule from~\eqref{eq:coind-ind-rules},
  it suffices to show the inclusion
  \begin{itemize}[label=$(*)$]
    \item [$(*)$] ${\ired_{\alpha,\nest}} \subseteq G(\ired_{\fp},\ired_{\alpha,\nest}) \defd (\rstep \cup \mathrel{\down{\ired_{\fp}}})^*\relcomp \down{\ired_{\alpha,\nest}}$ \;.
  \end{itemize}\medskip

  Thus assume that $s \ired_{\alpha,\nest} t$,
  and let $\delta$ be a proof tree of nesting depth $\le \alpha$ deriving $s \ired_{\alpha,\nest} t$.
  The only possibility to derive $s \ired_{\nest} t$ is an application of the \rsplit-rule
  with the premise $s \mathrel{(\rstep \cup \ireddownfin_{\nest})^*} \relcomp \ireddown_{\nest} t$.
  Since $s \ired_{\alpha,\nest} t$, we have
  $s \mathrel{(\rstep \cup \ireddownfin_{\alpha,\nest})^*} \relcomp \ireddown_{\alpha,\nest} t$.
  Let $\tau$ be one of the steps $\ireddownfin_{\alpha,\nest}$ displayed in the premise.
  Let $u$ be the source of $\tau$ and $v$ the target,
  so $\tau : u \ireddownfin_{\alpha,\nest} v$.
  The step $\tau$ is derived either via the \rid-rule or the \rlift-rule.
  The case of the \rid-rule is not interesting since we then can drop $\tau$ from the premise.
  Thus let the step $\tau$ be derived using the \rlift-rule.
  Then the terms $u,v$ are of form $u = f(u_1,\ldots,u_n)$ and $v = f(v_1,\ldots,v_n)$
  and for every $1 \le i \le n$ we have $u_i \ired_{\beta,\nest} v_i$ for some $\beta < \alpha$.
  Thus by induction hypothesis we obtain $u_i \ired_{\fp} v_i$ for every $1 \le i \le n$,
  and consequently $u \mathrel{\down{\ired_{\fp}\vphantom{i}}} v$.
  We then have $s \mathrel{(\rstep \cup \mathrel{\down{\ired_{\fp}\vphantom{i}})^*}} \relcomp \ireddown_{\alpha,\nest} t$,
  and hence $s \mathrel{G(\ired_{\fp},\ired_{\alpha,\nest})} t$.
  This concludes the proof.
\end{proof}

\section{Equivalence with the Standard Definition}\label{sec:equivalence}

In this section we prove the equivalence of the coinductively defined 
infinitary rewrite relations~$\ired$ from
Definitions~\ref{def:ired:restrict} (and~\ref{def:ired:fixedpoint})
with the standard definition
based on ordinal length rewrite sequences with metric and strong convergence at every limit ordinal
(Definition~\ref{def:itrs:standard}).
The crucial observation is the following theorem from~\cite{klop:vrij:2005}:
\begin{thm}[\normalfont{Theorem 2 of~\cite{klop:vrij:2005}}]\label{thm:finite}
  A transfinite reduction is divergent if and only if for some $n \in \nat$
  there are infinitely many steps at depth $n$.
\end{thm}

We are now ready to prove the equivalence of both notions:
\begin{thm}\label{thm:ired:equiv}
  We have ${\ired} = {\iredord}$.
\end{thm}

\begin{proof}
  We write $\iredorddown$ to denote a reduction $\iredord$ without root steps,
  and we write $\redord^\alpha$ and $\redorddown^\alpha$ to indicate the ordinal length $\alpha$.

  We begin with the direction ${\iredord} \subseteq {\ired}$.
  We define a function $\smktree$ (and $\smktreedownmfin$) 
  by guarded corecursion~\cite{coqu:1994},
  mapping rewrite sequences $s \redord^\alpha t$
  (and $s \redorddown^\alpha t$)
  to infinite proof trees derived using the rules from Definition~\ref{def:ired:restrict}.
  This means that every recursive call produces a constructor, contributing to the construction of the infinite tree.
  Note that the arguments of $\smktree$ (and $\smktreedownmfin$)
  are not required to be structurally decreasing.

  We do case distinction on the ordinal $\alpha$.
  If $\alpha = 0$, then $t = s$ and we define %
  \begin{align*}
    &\mktree{s \redord^0 s} 
    \;\;=\;\; \infer=[\rsplit] {s \ired s}{\mktreedown{s \redorddown^0 s}} \\[1ex]
    &\mktreedownmfin{x \redorddown^0 x} 
    \;\;=\;\; \infer=[\rid] {x \ireddownmfin x}{} \\
    &\mktreedownmfin{f(t_1,\ldots,t_n) \redorddown^0 f(t_1,\ldots,t_n)} 
    \;\;=\;\; \infer=[\rlift] {f(t_1,\ldots,t_n) \ireddownmfin f(t_1,\ldots,t_n)}{\mktree{t_1 \redord^0 t_1} & \cdots & \mktree{t_n \redord^0 t_n}}
  \end{align*}

\noindent If $\alpha > 0$, then, by Theorem~\ref{thm:finite} the rewrite sequence $s \redord^\alpha t$
  contains only a finite number of root steps.
  As a consequence, it is of the form:
  \begin{align*}
    s = s_0 \sto s_1 \cdots \sto s_{2n} \sto s_{2n+1} = t
  \end{align*}
  where for every $i \in \{0,\ldots,2n\}$:
  \begin{enumerate}
    \item for even $i$, $s_{i} \sto s_{i+1}$ is an infinite reduction below the root $S_i : s_i \redorddown^{\beta_i} s_{i+1}$, and
    \item for odd $i$, $s_{i} \sto s_{i+1}$ is a root step $s_i \rstep s_{i+1}$,
  \end{enumerate}
  where $\beta_i < \alpha$ if $i < 2n$ and $\beta_i \le \alpha$ if $i = 2n$.
  For $i < 2n$ we have $\beta_i < \alpha$  since every strict prefix must be shorter than the sequence itself.
  We define
  \begin{align*}
    \mktree{s \redord^\alpha t}
    \;\;=\;\; 
    \infer=[\rsplit] {s \ired t} { \delta_0 & \delta_1 & \cdots & \delta_{2n} }
  \end{align*}
  where, for $0 \le i < n$,
  \begin{align*}
    \delta_i = 
    \begin{cases}
      s_i \rstep s_{i+1} & \text{for odd $i$,} \\
      \mktreedownfin{S_i : s_i \redorddown^{\beta} s_{i+1}} & \text{for even $i$ with $i < 2n$,} \\
      \mktreedown{S_i : s_i \redorddown^{\beta} s_{i+1}} & \text{for even $i$ with $i = 2n$.}
    \end{cases}
  \end{align*} 
  
  For reductions $S : s \redorddown^\alpha t$ with $\alpha > 0$
  we have $s = f(s_1,\ldots,s_n)$ and $t = f(t_1,\ldots,t_n)$ for some $f \in \Sigma$ of arity~$n$
  and terms $s_1,\ldots,s_n,t_1,\ldots,t_n \in \iter{\Sigma}{\avars}$.
  Moreover, for every $i$ with $1 \le i \le n$, 
  there are rewrite sequences $S_i : s_i \redord^{\le \alpha} t_i$
  obtained by selecting from $S$ the subsequence of steps on the $i$-th argument.
  These steps are not necessarily consecutive, but selecting them nonetheless gives rise to a well-defined reduction.
  We define:
  \begin{align*}
    \mktreedownmfin{s \redorddown^\alpha t}
    \;\;=\;\;
    \infer=[\rlift] {s \ireddownmfin t}{\mktree{S_1 : s_1 \redord^{\le \alpha} t_1} & \cdots & \mktree{S_n : s_n \redord^{\le \alpha} t_n}}
  \end{align*}
  
  The obtained proof tree $\mktree{s \redord^\alpha t}$ derives $s \ired t$.
  To see that the requirement that there is no ascending path through this tree 
  containing an infinite 
  number of symbols $\ireddownfin$ is fulfilled,
  we note the following.
  The symbol $\ireddownfin$ is produced by $\mktreedownfin{s \redorddown^\beta t}$
  which is invoked in $\mktree{s \redord^\alpha t}$ for a $\beta$ that is strictly smaller than $\alpha$.
  By well-foundedness of $<$ on ordinals, no such path exists.

  We now show ${\ired} \subseteq {\iredord}$.
  We prove by well-founded induction on $\alpha \le \omega_1$ that
  ${\ired_{\alpha}} \subseteq {\iredord}$.
  This suffices since ${\ired} = {\ired_{\omega_1}}$.
  Let $\alpha \le \omega_1$ and assume that $s \ired_{\alpha} t$.
  Let $\delta$ be a proof tree of nesting depth $< \alpha$ witnessing $s \ired_{\alpha} t$.
  The only possibility to derive $s \ired t$ is an application of the \rsplit-rule
  with the premise $s \mathrel{(\rstep \cup \ireddownfin)^*} \relcomp \ireddown t$.
  Since $s \ired_{\alpha} t$, we have
  $s \mathrel{(\rstep \cup \ireddownfin_{\alpha})^*} \relcomp \ireddown_{\alpha} t$.
  By induction hypothesis we have 
  $s \mathrel{(\rstep \cup \iredord)^*} \relcomp \ireddown_{\alpha} t$,
  and thus
  $s \iredord \relcomp \ireddown_{\alpha} t$.
  We have ${\ireddown_{\alpha}} = {\down{\ired_{\alpha}\vphantom{i}}}$, and consequently
  $s \iredord s_1 \mathrel{\down{\ired_{\alpha}\vphantom{i}}} t$ for some term $s_1$.  
  Repeating this argument on
  $s_1 \mathrel{\down{\ired_{\alpha}\vphantom{i}}} t$, we get 
  $s \iredord s_1 \mathrel{\down{\iredord\vphantom{i}}} s_2 \mathrel{\down{\down{\ired_{\alpha}\vphantom{i}}}} t$.  
  After $n$ iterations, we obtain
  \begin{align*}
    s \iredord s_1 
    \mathrel{\down{\iredord\vphantom{i}}} s_2 
    \mathrel{\down{\down{\iredord\vphantom{i}}}} s_3
    \mathrel{\down{\down{\down{\iredord\vphantom{i}}}}} s_4
    \cdots \mathrel{({\iredord})^{-(n-1)}} s_n
    \mathrel{({\ired_{\alpha}})^{-n}} t
  \end{align*}
  where $({\ired_{\alpha}})^{-n}$
  denotes the $n$th iteration of $x \mapsto \down{x}$ on $\ired_{\alpha}$.
  
  Clearly, the limit of $\{s_n\}$ is $t$.  Furthermore, each of the reductions $s_n \iredord s_{n+1}$ 
  are strongly convergent and take place at depth greater than or equal to $n$.
  Thus, the infinite concatenation of these reductions yields a strongly convergent reduction from $s$ to $t$
  (there is only a finite number of rewrite steps at every depth $n$).
\end{proof}

\section{Infinitary Equational Reasoning and Bi-Infinite Rewriting}\label{sec:ieq}

\subsection{Infinitary Equational Reasoning}%

\begin{samepage}
\begin{defi}\label{def:ieq:rules}
  Let $\aes$ be a TRS over $\Sigma$, and let $T = \iter{\Sigma}{\avars}$.
  We define \emph{infinitary equational reasoning} as the relation~%
  ${\ieq} \subseteq T \times T$
  by the %
  mutually coinductive rules:
  \begin{align*}
    \infer=
    {s \ieq t}
    {s \mathrel{(\leftarrow_{\varepsilon} \cup \rstep \cup \ieqdown)^*} t}
    &&  
    \infer=
    {f(s_1,s_2,\ldots,s_n) \ieqdown f(t_1,t_2,\ldots,t_n)}
    {s_1 \ieq t_1 & \cdots & s_n \ieq t_n}
  \end{align*}
  where ${\ieqdown} \subseteq T \times T$ stands for infinitary equational reasoning below the root.
\end{defi}
\end{samepage}

Note that, in comparison with the rules~\eqref{rules:intro:ieq} for $\ieq$ from the introduction, 
we now have used $\leftarrow_\varepsilon \cup \rstep$ instead of $=_\aes$.
It is not difficult to see that this gives rise to the same relation.
The reason is that we can `push' non-root rewriting steps $=_\aes$ 
into the arguments of~$\ieqdown$.

\begin{exa}\label{ex:ieq}
  Let $\aes$ be a TRS consisting of the following rules:
  \begin{align*}
    \fun{a} & \to \fun{f}(\fun{a})  &
    \fun{b} & \to \fun{f}(\fun{b}) &
    \fun{C}(\fun{b}) & \to \fun{C}(\fun{C}(\fun{a})) \,.
  \end{align*}
  Then we have $\fun{a} \ieq \fun{b}$ as derived in Figure~\ref{fig:ieq:fagb} (top),
  and $\fun{C}(\fun{a}) \ieq \fun{C}^\omega$ as in Figure~\ref{fig:ieq:fagb} (bottom).
  \begin{figure}[h!]
    \begin{framed}\vspace{2ex}
    \begin{gather*}
    \infer=
    { \fun{a} \ieq \fun{b} }
    {
      \fun{a} \to_{\varepsilon} \fun{f}(\fun{a})
      &
      \infer=
      {\fun{f}(\fun{a}) \ieqdown \fun{f}^\omega}
      {
        \infer=
        {\fun{a} \ieq \fun{f}^\omega}
        {
          \fun{a} \to_{\varepsilon} \fun{f}(\fun{a})
          &
          \infer=
          {\fun{f}(\fun{a}) \ieqdown \fun{f}^\omega}
          {{\makebox(0,0){
                \hspace{13mm}\begin{tikzpicture}[baseline=15ex,yscale=1.4]
                \draw [->,thick,dotted] (-4mm,0) -- (-4mm,1mm) to[out=90,in=100] (8mm,-1mm) to[out=-80,in=-20,looseness=1.4] (-6mm,-13mm);
                \end{tikzpicture}
              }}}
        }
      }
      &&&&
      \infer=
      {\fun{f}^\omega \ieqdown \fun{f}(\fun{b})} 
      {
        \infer=
        {\fun{f}^\omega \ieq \fun{b}}
        {
          \infer=
          {\fun{f}^\omega \ieqdown \fun{f}(\fun{b})}
          {{\makebox(0,0){
                \hspace{-13mm}\begin{tikzpicture}[baseline=15ex,yscale=1.4]
                \draw [->,thick,dotted] (4mm,0) -- (4mm,1mm) to[out=90,in=80] (-8mm,-1mm) to[out=-100,in=-160,looseness=1.6] (6mm,-13mm);
                \end{tikzpicture}
              }}
          }
          &
          \fun{f}(\fun{b}) \leftarrow_{\varepsilon} \fun{b}
        }
      }
      &
      \fun{f}(\fun{b}) \leftarrow_{\varepsilon} \fun{b}
    }
    \\[2ex] %
    \infer=
    {\fun{C}(\fun{a}) \ieq \fun{C}^\omega}
    {
      \infer=
      {\fun{C}(\fun{a}) \ieqdown \fun{C}(\fun{b})} 
      {\infer={\fun{a} \ieq \fun{b}}{\text{(as above)}}}
      &
      \fun{C}(\fun{b}) \rstep \fun{C}(\fun{C}(\fun{a}))
      &
      \infer=
      {\fun{C}(\fun{C}(\fun{a})) \ieqdown \fun{C}^\omega}
      {\infer={\fun{C}(\fun{a}) \ieq \fun{C}^\omega}
            {\makebox(0,0){
              \hspace{22mm}\begin{tikzpicture}[baseline=17ex,xscale=1.3,yscale=1.4]
              \draw [->,thick,dotted] (0,0) -- (0,1mm) to[out=90,in=100] (15mm,-1mm) to[out=-80,in=-25,looseness=1.4] (-0mm,-14mm);
              \end{tikzpicture}
            }}
      }
    }
    \end{gather*}
    \end{framed}
    \caption{An example of infinitary equational reasoning, 
      deriving $\fun{C}(\fun{a}) \ieq \fun{C}^\omega$ in the TRS~$\aes$ of Example~\ref{ex:ieq}.
      Recall Notation~\ref{not:transitivity}.}
    \label{fig:ieq:fagb}
  \end{figure}
\end{exa}

Definition~\ref{def:ieq:rules} of \,$\ieq$\, can also be defined using 
a greatest fixed point as follows: 
\begin{align*}
  {\ieq} \;\;\defd\;\; \gfp{R}{(\leftarrow_{\varepsilon} \cup \rstep \cup \mathrel{\down{R}})^*}  \,,
\end{align*}
where $\down{R}$ was defined in Definition~\ref{def:lifting}.
The equivalence of these definitions 
can be established in a similar way as in Theorem~\ref{thm:equiv:der:fp}.
As remarked at the end of section~\ref{sec:coinduction}, the map
$R \mapsto (\leftarrow_{\varepsilon} \cup \rstep \cup \mathrel{\down{R}})^*$
is monotone, and consequently the greatest fixed point exists. 

We note that, in the presence of collapsing rules (i.e., rules $\ell \to r$ where $r \in \avars$),
everything becomes equivalent: $s \ieq t$ for all terms $s,t$. 
For example, having a rule $\tail(x) \to x$ we obtain that
$s \ieq \tail(s) \ieq \tail^2(s) \ieq \cdots \ieq \tail^\omega$ for every term $s$.
This can be overcome by forbidding certain infinite terms and certain infinite limits.

\subsection{Bi-Infinite Rewriting}%

Another notion that arises naturally in our setup
is that of bi-infinite rewriting,
allowing rewrite sequences to extend infinitely forwards and backwards.
We emphasize that each of the steps $\rstep$ 
in such sequences is a forward step.

\begin{defi}\label{def:ibi:rules}
  Let $\atrs$ be a term rewriting system over $\Sigma$,
  and let $T= \iter{\Sigma}{\avars}$.
  We define the \emph{bi-infinite rewrite relation ${\ibi} \subseteq T \times T$} 
  by the following coinductive rules
  \begin{align*}
    \infer=
    {s \ibi t}
    {s \mathrel{(\rstep \cup \ibidown)^*} t}
    &&  
    \infer=
    {f(s_1,s_2,\ldots,s_n) \ibidown f(t_1,t_2,\ldots,t_n)}
    {s_1 \ibi t_1 & \cdots & s_n \ibi t_n}
  \end{align*}
  where ${\ibidown} \subseteq T \times T$ stands for bi-infinite rewriting below the root.
\end{defi}

If we replace $\ieq$ and $\ired$ by $\ibi$, and $\ieqdown$ and $\ireddown$ by $\ibidown$,
then Examples~\ref{ex:Ca:a} and~\ref{ex:fab} %
are illustrations of this rewrite relation.

Again, like $\ieq$, the relation $\ibi$\, can also be defined using 
a greatest fixed point: %
\begin{align*}
  {\ibi} &\;\;\defd\;\; \gfp{R}{(\rstep \cup \mathrel{\down{R}})^*} \,.
\end{align*}
As remarked at the end of section~\ref{sec:coinduction},
$R \mapsto (\rstep \cup \mathrel{\down{R}})^*$
is monotone, and hence the greatest fixed point exists. 
Also, the equivalence of Definition~\ref{def:ibi:rules} with this $\sgfp$-definition 
can be established similarly. %
\newpage
\section{Relating the Notions}\label{sec:venn}

\begin{lem}
  Each of the relations $\ired$, $\ibi$ and $\ieq$ is reflexive and transitive.
  The relation $\ieq$ is also symmetric.
\end{lem}
\begin{proof}
  Follows immediately from
  the fact that the relations are defined using the reflexive-transitive closure in each of their first rules. 
\end{proof}

\newcommand{\ibii}{\mathrel{\reflectbox{$\ibi$}}}
\begin{thm}
  For every TRS $\atrs$ we have the following inclusions:
  \begin{center}
    \begin{tikzpicture}[node distance=20mm]
      \node (ired) {$\ired$};
      \node (ibi) [right of=ired] {$\ibi$};
      \node (iredconv) [right of=ired,yshift=-6mm] {$({\iredi} \cup {\ired})^*$};
      \node (ibiconv) [right of=ibi,node distance=25mm] {$({\ibii} \cup {\ibi})^*$};
      \node (ieq) [right of=ibiconv,node distance=25mm] {$\ieq$};
      
      \node at ($(ired)!.5!(ibi)$) {$\subseteq$};
      \node at ($(ibi)!.4!(ibiconv)$) {$\subseteq$};
      \node at ($(ired)!.6!(iredconv.west)$) [rotate=-30] {$\subseteq$};
      \node at ($(iredconv.east)!.4!(ibiconv.west)$) [rotate=20] {$\subseteq$};
      \node at ($(ibiconv)!.6!(ieq)$) {$\subseteq$};
    \end{tikzpicture}
  \end{center}
  Moreover, for each of these inclusions there exists a TRS for which the inclusion is strict. 
\end{thm}

\begin{proof}
  The inclusions ${\ired} \subsetneq {\ibi} \subsetneq {\ieq}$ have
  already been established in the introduction,
see equation \eqref{eq:inclusion1}.
  The inclusion ${\ired} \subsetneq {({\iredi} \cup {\ired})^*}$ is well-known (and obvious).
  Also ${\ibi} \subsetneq {({\ibii} \cup {\ibi})^*}$ is immediate since $\ibi$ is not symmetric.
  
  The inclusion ${({\iredi} \cup {\ired})^*} \subseteq {({\ibii} \cup {\ibi})^*}$
  is immediate since ${\ired} \subseteq {\ibi}$.
  Example~\ref{ex:Ca:a} witnesses strictness of this inclusion.
  The reason is that, for this example, ${\ired} = {\to^*}$ as the system does not admit any forward limits.
  Hence ${({\iredi} \cup {\ired})^*}$ is just finite conversion on potentially infinite terms.
  Thus $\fun{C}^\omega \ibi \fun{a}$, but not $\fun{C}^\omega \mathrel{({\iredi} \cup {\ired})^*} \fun{a}$.

  The inclusion ${({\ibii} \cup {\ibi})^*} \subseteq {\ieq}$
  follows from the fact that $\ieq$ includes $\ibi$ and is symmetric and transitive.
  Example~\ref{ex:ieq} witnesses strictness:
  $\fun{C}(\fun{a}) = \fun{C}^\omega$ can only be derived by
  a rewrite sequence of the form:
  \begin{align*}
    \fun{C}(\fun{a}) \ibi \fun{C}(\fun{f}^\omega) \ibiinv \fun{C}(\fun{b}) \to \fun{C}(\fun{C}(\fun{a}))
    \ibi \fun{C}(\fun{C}(\fun{f}^\omega)) \ibiinv \fun{C}(\fun{C}(\fun{b})) \to \fun{C}(\fun{C}(\fun{C}(\fun{a}))) \ibi \cdots
  \end{align*}
  and hence we need to change rewriting directions infinitely often  %
  whereas ${({\ibii} \cup {\ibi})^*}$ allows to change the direction only a finite number of times.
\end{proof}

\newcommand{\ieqclose}[1]{\stackrel{\infty}{T}(#1)}
\begin{defi}
  For relations $S \subseteq \iter{\Sigma}{\avars} \times \iter{\Sigma}{\avars}$ we define 
  \begin{align*}
    \ieqclose{S} \;\;\defd\;\; \gfp{R}{(S^{-1} \cup S \cup {\mathrel{\down{R}}})^*}\;.
  \end{align*}  
\end{defi}

\begin{lem}\label{lem:ieq:closed}
  We have ${\ieqclose{S}} = {\ieqclose{\ieqclose{S}}}$ for every $S \subseteq \iter{\Sigma}{\avars} \times \iter{\Sigma}{\avars}$.
\end{lem}

\begin{proof}
  For every relation $S$ we have
  $S \subseteq (S^{-1} \cup S \cup {\mathrel{\down{R}}})^*$
  and hence $S \subseteq {\ieqclose{S}}$ 
  by \eqref{eq:coind-ind-rules}.
  Hence it follows that ${\ieqclose{S}} \subseteq {\ieqclose{\ieqclose{S}}}$.
  For ${\ieqclose{\ieqclose{S}}} \subseteq {\ieqclose{S}}$ we note that
  \begin{align*}
    {\ieqclose{\ieqclose{S}}} 
    &= {(\;{\ieqclose{S}^{-1}} \cup {\ieqclose{S}} \cup {\down{\ieqclose{\ieqclose{S}}}}\;)^*} && \text{by definition}\\
    &= {(\;{\ieqclose{S}} \cup {\down{\ieqclose{\ieqclose{S}}}}\;)^*} && \text{by symmetry of $\ieqclose{S}$} \\
    &= {(\;(\;{S^{-1}} \cup S \cup {\down{\ieqclose{S}}}\;)^* \cup {\down{\ieqclose{\ieqclose{S}}}}\;)^*} && \text{by definition} \\
    &= {(\;{S^{-1}} \cup S \cup {\down{\ieqclose{S}}} \cup {\down{\ieqclose{\ieqclose{S}}}}\;)^*}  \\
    &= {(\;{S^{-1}} \cup S \cup {\down{\ieqclose{\ieqclose{S}}}}\;)^*} && \text{since ${\ieqclose{S}} \subseteq {\ieqclose{\ieqclose{S}}}$} 
  \end{align*}
  Thus ${\ieqclose{\ieqclose{S}}}$ is a fixed point of $R \mapsto (S^{-1} \cup S \cup {\mathrel{\down{R}}})^*$,
  and hence ${\ieqclose{\ieqclose{S}}} \subseteq {\ieqclose{S}}$.
\end{proof}

It follows immediately that $\ieq$ is closed under $\ieqclose{\cdot}$.
\begin{cor} \label{cor:Tinf}
  We have ${\ieq} \;=\; {\ieqclose{\ieq}}$ for every TRS $\atrs$.
\end{cor}

\begin{proof}
  We have 
  ${\ieq} \;=\; {\ieqclose{\to_{\varepsilon}}} 
          \;=\; {\ieqclose{\ieqclose{\to_{\varepsilon}}}} 
          \;=\; {\ieqclose{\ieq}}$.
\end{proof}

The work~\cite{kahr:2013} introduces various notions of infinitary rewriting.
We comment on the notions that are closest to the relations $\ibi$ and $\ieq$ introduced in our paper.
First, we note that it is not difficult to see that ${\ibi} \subsetneq {\mred_t}$ where
$\mred_t$ is the topological graph closure of $\to$.
The paper~\cite{kahr:2013} also introduces a notion of infinitary equational reasoning
with a strongly convergent flavour, namely:
\begin{align*}
  S_E(R) = (S \funcomp E)^\star(R)
\end{align*}
where 
\begin{enumerate}
  \item $E(R)$ is the equivalence closure of $R$, and
  \item $S(R)$ is the strongly convergent rewrite relation obtained from (single steps) $R$,
  \item and $f^{\star}(R)$ is defined as $\mu x.\,R\cup f(x)$.
\end{enumerate}

\begin{lem}
  We have ${S_E(\to)} \;\subseteq\; {\ieq}$ for every TRS $\atrs$. 
\end{lem}

\begin{proof}
 The following containments are immediate:
  \begin{enumerate}
    \item ${\to} \,\subseteq\, {\ieq}$,
    \item ${E(\ieq)} \,=\, {\ieq}$, and
    \item ${S(\ieq)} \,\subseteq\, {\ieqclose{\ieq}} \,=\, {\ieq}$
      (Corollary \ref{cor:Tinf}).
  \end{enumerate} 
From the definition of $S_E(\cdot)$ as a least fixed point, the claim follows.
\end{proof}

It could be reasonable to conjecture that $S_E(\to)$ and ${\ieq}$
coincide.  We now show that this is not the case.

\begin{exa}
\newcommand{\mcS}{\mathcal{S}}
\newcommand{\setof}[1]{\{#1\}}
\newcommand{\semof}[1]{\llbracket#1\rrbracket}
\newcommand{\infof}[1]{\mathrm{NF}^\infty(#1)}
Consider the iTRS $\atrs$ consisting of the rules
\begin{align*}
  c(b(x)) &\to a(a(x))\\
  c(a(x)) &\to b(b(x))
\end{align*}

Notice that $a^\omega \ieq b^\omega$ in $\atrs$.  One possible
derivation $\delta$ of this fact is given below
where $\overline{\delta}$ is the same as $\delta$, but with all pairs mirrored
and premises of the split rule are listed in reverse order.
We also use that $b(b^\omega) = b^\omega$ and $a(a^\omega) = a^\omega$.

{\small  \begin{align*}
    \infer=
    {a^\omega \ieq b^\omega}
    { a(a(a^\omega)) {\leftarrow_\varepsilon} c(b(a^\omega)) \quad
      \infer=
      { c(b(a^\omega)) {\ieqdown} c(b(b^\omega)) }
      {\delta}
      &
      \infer=
      { c(b^\omega) {\ieqdown} c(a^\omega) }
      { \overline{\delta} }
      &
      {\infer=
      { c(a(a^\omega)) {\ieqdown} c(a(b^\omega))}
      {\delta}
      }
      \quad c(a(b^\omega)) {\rstep} b(b(b^\omega))
    }
  \end{align*}
  }
One does not have $\pair{a^\omega}{b^\omega} \in S_E(\atrs)$, however.
Let us sketch a proof of this.
First, notice that, for any relation $R$, $S_E(R)$ can alternatively
be described as
\begin{align}
\label{eq:S_E}
S_E(R) := \mu x. R \cup S(E(x)) = \mu x. R \cup E(x) \cup S(x)
= \mu x. R \cup E(S(x))
\end{align}
This is because a prefixed point of the
composition $S \circ E$ is a prefixed point of both
monotone operators $E$ and $S$, and vice versa.
We are particularly interested in the operator appearing on the right-hand side of
\eqref{eq:S_E}.  After a single iteration, it yields the usual concept
of infinitary conversion in the iTRS $\atrs$.

Observe that $\atrs$ is in fact a string rewrite system.  Being
orthogonal, and having no collapsing rules, we know that $\atrs$
satisfies both iCR and iSN.  Therefore, infinitary conversion in
$\atrs$ is characterized by canonical semantics $\mcS$,
consisting of infinitary normal forms.  It is easy to see
that these are precisely
\begin{align*}
\mcS &= \setof{w \in \setof{a,b}^m \mid m \le \omega} \sqcup \setof{w c^n \mid w \in
  \setof{a,b}^*, n \le \omega}\\
&= \setof{wc^n \mid w \in \setof{a,b}^m, m+n \le \omega}
\end{align*}

It is a curious fact that $E(S(\cdot))$ does not yet stabilize at
$(=_\mcS)$, the equality of infinitary normal forms.  But it does
stabilize after one more iteration --- without relating
$a^\omega$ and $b^\omega$.

To see this, consider a sequence of $=_\mcS$-steps
\begin{align}
s_0 \ired \cdot \iredi s_1 \ired \cdot \iredi s_2 \ired \cdot \cdot \cdot\qquad
&\text{with }\lim_{n \to \infty} s_n  =  s_\infty
\label{eq:iconv}
\end{align}
By infinitary rewriting theory, for each $n$, we can find standard,
$\omega$-compressed reductions
\begin{align}
\label{eq:nf-red}
\rho_n : s_n \ired s := \infof{s_0} \in \mcS
\end{align}
Let us say that a finite prefix $w \leq s$ is \emph{stable} if we can find
a number $N$, a prefix $v \leq s_\infty$, and a reduction $\nu : v
\to^* w$ such that, for all $n \ge N$:
\begin{itemize}
\item $s_n = v s_n'$
\item $\rho_n$ factors as $\rho_n = \nu \circ \rho_n'$, where $\rho_n'
  : s_n' \to s'$ and $s=ws'$.
\end{itemize}

If every prefix of $s$ is stable, then it is easy to see that
$s_\infty \ired s$; then the infinite conversion \eqref{eq:iconv}
yields no new pairs in $S_E(\atrs)$.

Otherwise, there is a maximal stable prefix $w$ --- which may or may
not be empty.
Fixing this $w=w_{\max{}}$, with respective $N$, $v$, and $\nu$, we find that
\begin{itemize}
\item 
$s = w s'$, $s_\infty = v s_\infty'$, $\nu : v \to w$;
\item 
$s_N' \ired \cdot \iredi s_{N+1}' \ired \cdot \cdot \cdot$,
with $\lim_{n \to \infty} s_n' = s_\infty'$;
\item $s'_n \ired s'$ for $n \ge N$;
\end{itemize}

We claim that $s'_\infty=c^\omega$.
For suppose
$s'_\infty \in \setof{c^kxu \mid k \ge 0, x \in \setof{a,b}, u \in \setof{a,b,c}^{\le \omega}}$.
After $k$ steps of outermost reduction, the outermost letter becomes
$y \in \setof{a,b}$, and there is no longer a redex present at the root.
By continuity of the sequence $\setof{s'_n}$, this even happens at
each $n \ge M$, for some $M \ge N$.
But now $y$ becomes a stable prefix of $s'$, and $wy$ a
stable prefix of $s$ --- contradicting maximality of $w$.
So $s'_\infty = c^\omega$.
Now, unless $s' = c^\omega$ as well,
trivializing the whole thing,
the reductions $\rho'_n$ must be non-trivial, yielding subterms of
form $aa(x)$ or $bb(x)$.
Then $s'$ has a prefix resulting from a reduction of a term of form
$c^kau$ or $c^kbu$ to normal form --- for arbitrarily large $k$.  A
cursory examination of the rules reveals that the only two
possibilities for such prefixes are $a(ab)^k$ and $b(ba)^k$.  Since
the $k$ indeed is unbounded as $s'_n \to s'_\infty$, we conclude that
$s' \in \setof{a(ab)^\omega,b(ba)^\omega}$.

Thus, the only pairs added to $(=_\mcS)$ by the operator
$S(\cdot)$ are those of the form $\pair{w\alpha}{vc^\omega}$, where $v \to^* w$
and $\alpha \in \setof{a(ab)^\omega,b(ba)^\omega}$.
The equivalence generated by these relations corresponds to infinitary
conversion in the augmented iTRS:
\begin{align*}
  \atrs^+ := \atrs \cup \begin{cases}
a(ab)^\omega \to c^\omega\\
b(ba)^\omega \to c^\omega
\end{cases}
\end{align*}
\newcommand{\eqplus}{{=_{\mcS^+}}}
Let us denote this conversion by $\eqplus$.
It remains to show that $S_E(\eqplus)$ coincides with $\eqplus$.
We proceed as before, starting with a chain of $\atrs^+$-conversions
\begin{align}
s_0 \ired \cdot \iredi s_1 \ired \cdot \iredi s_2 \ired \cdot \cdot \cdot\qquad
&\text{with }\lim_{n \to \infty} s_n  =  s_\infty
\label{eq:iconv'}
\end{align}
We remark that $\atrs^+$ is still confluent, owing to lack of overlap.

Even though compression fails due to the presence of rules with
infinite left sides, this failure happens to be completely innocuous:
if any of these rules are ever used in a reduction sequence, they will
replace an infinite part of the term by a normal form which cannot
interact with anything --- and the finite prefix which remains is
strongly normalizing (since $\atrs^+$ is finitarily SN).
In particular, for any $\atrs^+$-reduction $\rho:u \ired u'$, the following are
equivalent:
\begin{itemize}
\item $\rho$ cannot be compressed to length $\omega$;
\item $\rho$ factors as $u \ired w \alpha \to w c^\omega$, where
$\alpha \in \setof{a(ab)^\omega,b(ba)^\omega}$ and $u \ired w\alpha$ is
is an infinite reduction.
\end{itemize}
%
We thus again obtain standard, $(\omega+1)$-compressed reductions
\begin{align}
\label{eq:nf-red'}
\rho_n : s_n \ired s := \infof{s_0} \in \mcS
\end{align}
If it so happens that, no reduction
$\rho_n$ fires any of the new rules, then it is evident that
$(s,s_\infty)$ is already included in $\eqplus$.
Otherwise, if one of the new rules is ever fired, then $s = w
c^\omega$.
If $t$ is any term, and $\rho$ is a reduction from $t$ to a normal
form of the shape $w c^\omega$, then $\rho$ factors as $\rho^f \circ
\rho^i \circ \rho^!$, where
\begin{itemize}
\item $t = t^0 t^i$;
\item $\rho^f : t_0 \to^* w$;
\item $\rho^i : t_i \ired \alpha$ are reductions in $\atrs$, where $\alpha \in
\setof{a(ab)^\omega,b(ba)^\omega}$;
\item $\rho^! : \alpha \to c^\omega$.
\end{itemize}

Applying this observation for each $\rho_n$ from \eqref{eq:nf-red'},
we conclude that the initial part $\rho^f_n$ must eventually stabilize
(due to stabilization of prefixes as
$s_n \stackrel{n \to \infty}{\longrightarrow} s_\infty$).
For the same reason, we have that $\rho_n^!$ must eventually settle on
one of the two new rules, with the target of $\rho^i$ converging to its
left side.
The remaining reductions $\rho^i$, being pure $\atrs$-reductions, are
covered by our earlier analysis, and so we conclude that $s_\infty =
s_\infty^0 s_\infty^i$, with $s_\infty^0 \eqplus w$, and $s_\infty^i
\eqplus c^\omega$.

Finally, let us remark why it suffices to consider limits of length
$\omega$.  This is settled by induction on the sequence length $\beta$.

As decisively settled in TeReSe, a strongly convergent sequence is
necessarily of at-most-countable length $\beta$.

If $\beta$ is a successor ordinal, then we conclude by finite
induction from the greatest limit ordinal less than $\beta$.

Otherwise, $\beta$ will be a limit ordinal, and we shall be able to
produce a sequence $s_0, s_{\beta(1)}, s_{\beta(2)},\dots$ as before,
with $s_{\beta(i)} \ired \cdot \iredi s_0$. 

The same analysis applies, and we deduct that $s_\beta \ired_{\atrs^+}
\mathsf{NF}(s_0)$. \qedhere
\end{exa}

\section{Correspondence of Proof Trees and Rewrite Sequences}\label{sec:permute}

\newcommand{\therule}[2]{\mathit{rul}(#1,#2)}
\newcommand{\thepos}[2]{\mathit{pos}(#1,#2)}
\newcommand{\thesub}[2]{\mathit{sub}(#1,#2)}

In this section, we investigate the correspondence between ordinal-indexed rewrite sequences
and coinductive proof trees.
We define a \emph{correspondence relation} that makes precise
when a rewrite sequence is represented by a certain proof tree.
In general, this correspondence is a many-to-many relation:
a proof tree represents a class of rewrite sequences, and 
a rewrite sequence can be represented by different proof trees.

We then define \emph{canonical proof trees} for $\ired$ 
in such a way that every ordinal-indexed rewrite sequence has a unique representative among the canonical proof trees.
More precisely, there is a many-to-one correspondence between rewrite sequences and canonical proof trees.
To characterise the class of rewrite sequences represented by the same canonical proof tree, 
we introduce a notion of equivalence on infinitary rewrite sequences,
called \emph{parallel permutation equivalence}.
Thereby two rewrite sequences are considered equivalent 
if they differ only in the order of steps in parallel subtrees.

\begin{nota}
  For a rewrite sequence $S : s_0 \to^\alpha_\atrs s_\alpha$ 
  consisting of steps $(s_{\beta} \to s_{\beta+1})_{\beta < \alpha}$
  arising from the application of the rule $\ell_\beta \to r_\beta$ 
  with substitution $\sigma_\beta$ at position $p_\beta$, respectively,
  we introduce the following notation
  \begin{align*}
    \therule{S}{\beta} &= \ell_\beta \to r_\beta\\
    \thepos{S}{\beta} &= p_\beta\\
    \thesub{S}{\beta} &= \sigma_\beta\\
  \end{align*}
  for every $\beta < \alpha$.
\end{nota}

\subsection{The Correspondence Relation}
Assume that we have a term $f(s_1,\ldots,s_n)$ and rewrite sequences on the direct subterms
$S_1 : s_1 \to^{\alpha_1} t_1$, \ldots, $S_n : s_n \to^{\alpha_n} t_n$.
As these rewrite sequences occur in parallel subterms, any interleaving of them 
gives rise to a rewrite sequence $f(s_1,\ldots,s_n) \to^\beta f(t_1,\ldots,t_n)$.
The following definition introduces the notion of \emph{interleaving}
on the basis of a monotonic bijective embedding of the disjoint union $\alpha_1 \uplus \ldots \uplus \alpha_n$ into $\beta$.

\begin{defi}
  Let $f \in \Sigma$ of arity $n$.
  Let $S_i : s_i \to^{\alpha_i} t_i$ be rewrite sequences of length $\alpha_i$ for every $i \in \{\,1,\ldots,n\,\}$.
  A rewrite sequence $T : f(s_1,\ldots,s_n) \to^\beta f(t_1,\ldots,t_n)$ of length $\beta$
  is called an \emph{interleaving of $S_1,\ldots,S_n$ with root $f$} if there exists
  a bijection 
  \begin{align*}
    \xi : (\{1\}\times \alpha_1 \cup \ldots \cup \{n\}\times \alpha_n) \to \beta
  \end{align*}
  such that for every $i \in \{\,1,\ldots,n\,\}$ and every $\gamma < \alpha_i$ we have:
  \begin{enumerate}
    \item $\thepos{T}{\xi(i,\gamma)} = i\cdot \thepos{S_i}{\gamma}$ (\emph{corresponding position in the $i$-th argument}),
    \item $\therule{T}{\xi(i,\gamma)} = \therule{S_i}{\gamma}$ (\emph{same rule}),
    \item $\thesub{T}{\xi(i,\gamma)} = \thesub{S_i}{\gamma}$ (\emph{same substitution}), and
    \item for every $\gamma' < \gamma$ it holds that $\xi(i,\gamma') < \xi(i,\gamma)$ (\emph{monotonic embedding}).
  \end{enumerate}
\end{defi}

The following definition introduces the correspondence between coinductive proof trees
and ordinal-indexed rewrite sequences.

\begin{defi}\label{def:correspondence}
  Let $\atrs$ be a term rewriting system.
  We define the \emph{correspondence relation} between proof trees 
  (with respect to Definition~\ref{def:ired:restrict})
  and ordinal-indexed rewrite sequences 
  as the largest relation such that the following conditions hold:
  \begin{enumerate}
    \item 
      A proof tree of the form 
      \begin{align*}
        \infer=[\rsplit]{s \ired t}{\delta_1 & \delta_2 & \cdots & \delta_n}
      \end{align*}
      corresponds to a rewrite sequence $S : s \to^\alpha t$ 
      if $S$ is the concatenation of rewrite sequences $S_1,\ldots,S_n$
      such that $\delta_i$ corresponds to $S_i$ for every $i \in \{\,1,\ldots,n\,\}$.
    \item 
      A proof tree of the form 
      $s \rstep t$
      only corresponds to the rewrite sequence $s \rstep t$.
    \item 
      A proof tree of the form 
      \begin{align*}
        \infer=[\rlift]{f(s_1,s_2,\ldots,s_n) \ireddownmfin f(t_1,t_2,\ldots,t_n)}{\delta_1 & \delta_2 & \cdots & \delta_n}
      \end{align*}
      corresponds to a rewrite sequence $S : s \to^\alpha t$ 
      if $S$ is an interleaving of rewrite sequences $S_1,\ldots,S_n$ with root $f$
      such that the proof tree $\delta_i$ corresponds to the rewrite sequence $S_i$ for every $i \in \{\,1,\ldots,n\,\}$.
    \item 
      A proof tree of the form 
      \begin{align*}
        \infer=[\rid]
        {s \ireddownmfin s}
        {}
      \end{align*}
      only corresponds to the empty rewrite sequence $s \to^0 s$.
  \end{enumerate}
\end{defi}

\begin{rem}\label{rem:lift:choice}
  Note that a proof tree corresponds to more than one rewrite sequence if and only if 
  it contains an application of the lift-rule with (at least) two premises that do not correspond to empty rewrite sequences.
  The lift-rule introduces choice in the `construction' of the rewrite sequence 
  by allowing for an arbitrary interleaving of the rewrite sequences on the arguments.
\end{rem}

The following example illustrates that a proof term can
correspond to an infinite number of ordinal-indexed rewrite sequences.
\begin{exa}
  We consider the proof trees in Figures~\ref{fig:aComega} and~\ref{fig:fab}:
  \begin{enumerate}
    \item 
      The proof tree for $\fun{a} \ired \fun{C}^\omega$
      corresponds to the only rewrite sequence $\fun{a} \to^\omega \fun{C}^\omega$.
    \item 
      The proof tree for $\fun{b} \ired \fun{C}^\omega$
      corresponds to the only rewrite sequence $\fun{b} \to^\omega \fun{C}^\omega$.
    \item 
      The proof tree for $\fun{f}(\fun{a},\fun{b}) \ireddownfin \fun{f}(\fun{C}^\omega,\fun{C}^\omega)$
      corresponds to all possible interleavings of $\fun{a} \to^\omega \fun{C}^\omega$ and $\fun{b} \to^\omega \fun{C}^\omega$
      applied to the respective subterms $\fun{f}(\fun{a},\fun{b})$.
      Note that there are continuum many rewrite sequences
      that all have length $\omega$ or $\omega \cdot 2$.
  \end{enumerate}
\end{exa}

\noindent The next example shows that some rewrite sequences can be represented by 
multiple proof trees.
\begin{exa}
  There are multiple proof trees for the rewrite sequence $\fun{a} \to^\omega \fun{C}^\omega$,
  for example the proof trees shown in Figures~\ref{fig:aComega} and~\ref{fig:aComega:alternative}. 
\end{exa}

\begin{figure}[h!]
  \begin{framed}\vspace{4ex}
  \begin{align*}
    \infer=
    {\fun{a} \ired \fun{C}^\omega}
    {
      \fun{a} \rstep \fun{C}(\fun{a})
      &
      \infer=
      {\fun{C}(\fun{a}) \ireddown \fun{C}(\fun{C}(\fun{a}))}
      {
        \infer=
        {\fun{a} \ired \fun{C}(\fun{a})}
        {\fun{a} \rstep \fun{C}(\fun{a})}
      }
      &
      \infer=
      {\fun{C}(\fun{C}(\fun{a})) \ireddown \fun{C}^\omega}
      {
        \infer=
        {\fun{C}(\fun{a}) \ired \fun{C}^\omega}
        {
          \infer=
          {\fun{C}(\fun{a}) \ireddown \fun{C}^\omega}
          {
            \infer={\fun{a} \ired \fun{C}^\omega}
            {
              \makebox(0,0){
              \hspace{20mm}\begin{tikzpicture}[baseline=25ex,xscale=1.6,yscale=2.5]
              \draw [->,thick,dotted] (0,0) -- (0,1mm) to[out=90,in=100] (11mm,-1mm) to[out=-80,in=-20,looseness=1.4] (-0mm,-13mm);
              \end{tikzpicture}
              }
            }
          }
        }
      }
    }
  \end{align*}
  \end{framed}
  \caption{A second proof tree for $\fun{a} \ired \fun{C}^\omega$.}
  \label{fig:aComega:alternative}
\end{figure}

\begin{exa}
  Let $\atrs$ consist of the rule $\fun{f}(x) \to \fun{g}(x)$.
  Figures~\ref{fig:fomega:gomega:1} and~\ref{fig:fomega:gomega:2} show
  proof trees corresponding to rewrite sequences $\fun{f}^\omega \to^\omega \fun{g}^\omega$.
  The proof tree in Figure~\ref{fig:fomega:gomega:1} corresponds to
  the rewrite sequence
  \begin{align*}
    \fun{f}^\omega \to \fun{g}(\fun{f}^\omega) \to \fun{g}(\fun{g}(\fun{f}^\omega)) \to \cdots \to^\omega \fun{g}^\omega\;,
  \end{align*} 
  and the proof tree in Figure~\ref{fig:fomega:gomega:2} corresponds to
  \begin{align*}
    \fun{f}^\omega 
    \to \fun{g}(\fun{f}^\omega) 
    \to \fun{g}(\fun{f}(\fun{g}(\fun{f}^\omega))) 
    \to \fun{g}^3(\fun{f}^\omega) 
    \to \fun{g}^4(\fun{f}^\omega) 
    \to \fun{g}^5(\fun{f}^\omega) 
    \to \cdots \to^\omega \fun{g}^\omega\;.
  \end{align*} 
  Note that, by Remark~\ref{rem:lift:choice}, these rewrite sequences are unique.
  Both proof trees correspond to precisely one rewrite sequence since they do not contain 
  applications of the lift-rule (rules with conclusion $\ireddown$) with multiple premises.
\end{exa}

\begin{figure}[h!]
  \begin{framed}
  \begin{align*}
    \infer=
    {\fun{f}^\omega \ired \fun{g}^\omega}
    {
      \fun{f}^\omega \rstep \fun{g}(\fun{f}^\omega)
      &
      \infer=
      {\fun{g}(\fun{f}^\omega) \ireddown \fun{g}^\omega}
      {\infer={\fun{f}^\omega \ired \fun{g}^\omega}
          {\makebox(0,0){
            \hspace{17mm}\begin{tikzpicture}[baseline=13.5ex,xscale=1.3,yscale=1.3]
            \draw [->,thick,dotted] (0,0) -- (0,1mm) to[out=90,in=100] (11mm,-1mm) to[out=-80,in=-20,looseness=1.4] (-0mm,-13mm);
            \end{tikzpicture}
          }}
      }
    }
  \end{align*}
  \end{framed}
  \caption{A proof tree for $\fun{f}^\omega \to \fun{g}(\fun{f}^\omega) \to \fun{g}(\fun{g}(\fun{f}^\omega)) \to \cdots \to^\omega \fun{g}^\omega$.}
  \label{fig:fomega:gomega:1}
\end{figure}

\begin{figure}[h!]
  \begin{framed}
  \begin{align*}
    \infer=
    {\fun{f}^\omega \ired \fun{g}^\omega}
    {
      \fun{f}^\omega \rstep \fun{g}(\fun{f}^\omega)
      &
      \infer=
      {\fun{g}(\fun{f}^\omega) \ireddown \fun{g}^\omega}
      {
        \infer=
        {\fun{f}^\omega \ired \fun{g}^\omega}
        {
          \infer=
          {\fun{f}^\omega \ireddown \fun{f}(\fun{g}(\fun{f}^\omega))}
          {
            \infer=
            {\fun{f}^\omega \ired \fun{g}(\fun{f}^\omega)}
            {
              \fun{f}^\omega \rstep \fun{g}(\fun{f}^\omega)
            }
          }
          &
          \fun{f}(\fun{g}(\fun{f}^\omega)) \rstep \fun{g}(\fun{g}(\fun{f}^\omega))
          &
          \infer=
          {\fun{g}(\fun{g}(\fun{f}^\omega)) \ireddown \fun{g}^\omega}
          {
            \infer=
            {\fun{g}(\fun{f}^\omega) \ired \fun{g}^\omega}
            {
              \infer=
              {\fun{g}(\fun{f}^\omega) \ireddown \fun{g}^\omega}
              {
                \infer=
                {\fun{f}^\omega \ired \fun{g}^\omega}
                {
                  \text{Figure~\ref{fig:fomega:gomega:1}}
                }
              }
            }
          }
        }
      }
    }
  \end{align*}
  \end{framed}
  \caption{A proof tree for 
    $\fun{f}^\omega 
    \to \fun{g}(\fun{f}^\omega) 
    \to \fun{g}(\fun{f}(\fun{g}(\fun{f}^\omega))) 
    \to \fun{g}^3(\fun{f}^\omega) 
    \to \fun{g}^4(\fun{f}^\omega) 
    \to \cdots \to^\omega \fun{g}^\omega$.}
  \label{fig:fomega:gomega:2}
\end{figure}

\subsection{Canonical Proof Trees and Parallel Permutation Equivalence}
In order to have a unique proof tree for every ordinal-indexed rewrite sequence, 
we introduce `canonical' proof trees for $\ired$.
We show that the correspondence of canonical proof trees and rewrite sequences is 
a one-to-many relationship,
and we characterise the class of rewrite sequences represented by 
a canonical proof tree in terms of `parallel permutation equivalence'.
\begin{defi}
  A proof tree for $\ired$ is called \emph{canonical} if
  \begin{enumerate}
    \item 
      every application of the \rsplit{} rule
      is an instance of the \emph{canonical} \rsplit{} rule
      \medskip
      \begin{align*}
        \begin{aligned}
        \infer=[\textit{canonical }\rsplit]
        {s \ired t}
        {s \mathrel{(\ireddownfin \relcomp \rstep)^*} \relcomp \ireddown t}
        \end{aligned} &\text{\quad, and}
      \end{align*}
    \item 
      every application of the $\rid$ rule is an instance of
      \begin{align*}
        \infer=[\textit{canonical }\rid]
        {x \ireddownmfin x}
        {}
      \end{align*}
  \end{enumerate}
\end{defi}

In contrast with the \rsplit{} rule from Definition~\ref{def:ired:restrict},
the \emph{canonical} \rsplit{} rule replaces
\begin{align*}
  (\ireddownfin \cup \rstep)^* \text{\quad by \quad} 
  (\ireddownfin \relcomp \rstep)^*\;.
\end{align*}
Thereby the canonical form enforces that $\ireddownfin$ and $\rstep$ alternate.

The canonical \rid{} rule replaces
\begin{align*}
  s \ireddownmfin s \text{\quad by \quad} 
  x \ireddownmfin x\;.
\end{align*}
Thereby it enforces unique proof trees for empty rewrite sequences (using rules $\rlift$ and canonical $\rid$).

\begin{defi}
  Let $p,q \in \nat^*$ be positions.
  We define
  \begin{enumerate}
    \item $p \le q$ if $pr = q$ for some $r \in \nat^*$,
    \item $p \parallel q$ if $p \not\le q$ and $q \not\le p$.
  \end{enumerate}
  If $p \parallel q$, then we say that $p$ and $q$ are \emph{parallel} (to each other).
\end{defi}

Recall that we consider ordinals $\alpha$ to be the set of all smaller ordinals:
$\alpha = \{\beta \mid \beta < \alpha\}$.
This allows us to speak about functions $f : \alpha \to \beta$.

\begin{defi}\label{def:permute}
  Let $\atrs$ be a term rewriting system.
  Let $S : s \to^\infty_\atrs t_1$ and $T : s \to^\infty_\atrs t_2$ 
  be strongly convergent rewrite sequences of length $\alpha$ and $\beta$, respectively.

  The rewrite sequence $S$ is called \emph{parallel permutation equivalent} to $T$
  if there exists a bijection $f : \alpha \to \beta$ such that
  \begin{enumerate}
    \item 
      $\therule{S}{\gamma} = \therule{T}{f(\gamma)}$ and 
      $\thepos{S}{\gamma} = \thepos{T}{f(\gamma)}$ 
      for every $\gamma < \alpha$, and
    \item 
      $\thepos{S}{\gamma_1} \parallel \thepos{S}{\gamma_2}$
      whenever $\gamma_1 < \gamma_2 < \alpha$ and $f(\gamma_1) > f(\gamma_2)$.
  \end{enumerate}
\end{defi}

\noindent
Observe that, the bijective mapping $f : \alpha \to \beta$ defines 
a permutation of the steps in the sequence~$S$.
The condition (i) guarantees that the step indexed by $\gamma$ in $S$ 
corresponds to the step indexed by $f(\gamma)$ in $T$ as follows:
both steps must arise from the same rule applied at the same position.
The condition (ii) ensures that steps that have been permuted
(changed their relative order in the sequence),
arise from contractions at parallel positions.

In the following definition we select a subsequence of the steps of $T$
that corresponds to (the permutation of) a prefix of $S$.
For this purpose, we consider a step to be the application 
of a certain rule at a certain position.
We do not take into account the source and the target of the steps
as these may change due to preceding steps being dropped (not selected).

\begin{defi}\label{def:permute:prefix}
  Let $S : s_0 \to^\alpha_\atrs s_\alpha$ and $T : t_0 \to^\beta_\atrs t_\beta$
  be parallel permutation equivalent with respect to the bijection $f : \alpha \to \beta$.
  Let $\kappa \le \alpha$, and define $S'$ as the prefix of $S$ of length $\kappa$.
  We define the \emph{permutation of $S'$ with respect to $f$}
  as the rewrite sequence obtained from $T$ 
  by selecting the subsequence of steps at indexes $\gamma < \beta$ for which $f^{-1}(\gamma) < \kappa$.
\end{defi}
As dropping (not selecting) a step changes all subsequent terms in the sequence, 
we need to show that the selected steps still form a rewrite sequence 
(the rules are applicable at the designated positions).

\begin{proof}[Proof of well-definedness of Definition~\ref{def:permute:prefix}]
  To prove that the selected subsequence of steps from $T$ forms a rewrite sequence, 
  we show that every non-selected (dropped) step is parallel to all subsequent selected steps.
  As a consequence, the dropped steps do not influence the applicability of the selected steps.
  Let us consider a step that is dropped,
  that is $\gamma < \beta$ with $f^{-1}(\gamma) \ge \kappa$,
  and a subsequent step that is selected,
  $\gamma'$ with $\gamma < \gamma' < \beta$ and $f^{-1}(\gamma') < \kappa$.
  From parallel permutation equivalence of $S$ and $T$
  it follows immediately that
  $\thepos{T}{\gamma} \parallel \thepos{T}{\gamma'}$
  since $f^{-1}(\gamma') < \kappa \le f^{-1}(\gamma)$ and $\gamma < \gamma'$.
\end{proof}

\begin{lem}\label{lem:permute:prefix}
  Let $S : s_0 \to^\alpha_\atrs s_\alpha$ and $T : t_0 \to^\beta_\atrs t_\beta$
  be parallel permutation equivalent with respect to the bijection $f : \alpha \to \beta$.
  Every prefix $S'$ of $S$ is parallel permutation equivalent 
  to the permutation of $S'$ with respect to $f$.
\end{lem}

\begin{proof}
  Follows immediately from parallel permutation equivalence of the rewrite sequences $S$ and $T$ 
  together with Definition~\ref{def:permute:prefix}
  since the order of the selected steps is preserved
  (both from $S$ to $S'$ as well as from $T$ to the subsequence of selected steps of $T$).
\end{proof}

The following lemma states that parallel permutation equivalent sequences
converge to the same target term. 
\begin{lem}\label{lem:same:target}
  If the rewrite sequences $S : s_0 \to^\alpha_\atrs s_\alpha$ and $T : t_0 \to^\beta_\atrs t_\beta$
  are parallel permutation equivalent, then we have $s_\alpha = t_\beta$.  
\end{lem}

\begin{proof}
  We prove the claim by induction on $\alpha$.
  Let $S : s_0 \to^\alpha_\atrs s_\alpha$ and $T : t_0 \to^\beta_\atrs t_\beta$
  be rewrite sequences that are parallel permutation equivalent.
  Let $f : \alpha \to \beta$ be such that the conditions of Definition~\ref{def:permute} are fulfilled.
  We distinguish cases for $\alpha$:
  \begin{enumerate}
    \item 
      If $\alpha = 0$, it follows that $\beta = 0$.
      Then $s_0 = t_0$ by definition of parallel permutation equivalence
      (the starting terms of the reductions must be equal).
    \item 
      If $\alpha$ is a successor ordinal $\alpha = \alpha' + 1$,
      we proceed as follows.
      
      Let $S'$ be the prefix $s \to^{\alpha'} s_{\alpha'}$ of $S$ of length $\alpha'$.
      In other words, $S'$ is the rewrite sequence 
      obtained from $S$ by dropping the last step $s_{\alpha'} \to s_\alpha$.
      
      Let $T'$ be the permutation of $S'$ with respect to $f$.
      So, $T'$ is the rewrite sequence $s \to^{\beta'} u_{\beta'}$ with $\beta' \le \beta$ obtained 
      from $T$ by dropping the step $t_{f(\alpha')} \to t_{f(\alpha')+1}$.
      Let
      \begin{align*}
      (\ell \to r) &= \therule{T}{f(\alpha')} &
      p &= \thepos{T}{f(\alpha')} &
      \sigma &= \thesub{T}{f(\alpha')}
      \end{align*}
      Recall that the step $t_{f(\alpha')} \to t_{f(\alpha')+1}$
      can be dropped from $T$ as its position
      is parallel to the positions of all subsequent steps in $T$:
      $\thepos{T}{f(\alpha')} \parallel \thepos{T}{\gamma}$
      for every $\gamma$ with $f(\alpha') < \gamma < \beta$.
      From this fact it also follows that 
      \begin{align}
        &&
        u_{\beta'}|_p = t_{f(\alpha')}|_p = \ell\sigma &&
        t_{\beta}|_p = t_{f(\alpha')+1}|_p &&
        t_{\beta} = u_{\beta'}[r\sigma]_p 
        \label{t:drop}
      \end{align}
      Observe that $S'$ is shorter than $S$ as we have removed the last step.
      In contrast, the length of $T'$ may be less or equal to the length of $T$.
      For example, dropping the $5$th step from an $\omega$-long sequence does not decrease its length.
      
      From parallel permutation equivalence of $S$ and $T$,
      it follows by Lemma~\ref{lem:permute:prefix}
      that $S'$ and $T'$ are parallel permutation equivalent.
      By induction hypothesis, we may conclude that $s_{\alpha'} = u_{\beta'}$. 
      From \eqref{t:drop} it follows that
      \begin{align*}
        s_{\alpha'}|_p &= \ell\sigma &
        s_{\alpha} &= s_{\alpha'}[r\sigma]_p 
      \end{align*}
      since 
      $\therule{S}{\alpha'} = \therule{T}{f(\alpha')} = \ell \to r$ and
      $\thepos{S}{\alpha'} = \thepos{T}{f(\alpha')} = p$.
      Then 
      \begin{align*}
        s_{\alpha} = s_{\alpha'}[r\sigma]_p = u_{\beta'}[r\sigma]_p = t_{\beta}\;.
      \end{align*}
    \item 
      Assume $\alpha$ is a limit ordinal.
      We use $s =_n t$ to denote that the terms $s$ and $t$ coincide up to depth $n$. 
      For $s_\alpha = t_\alpha$ it suffices to show that $s_\alpha =_n t_\alpha$ for every $n \in \nat$.
      
      Let $n \in \nat$ be arbitrary.
      By strong convergence of $S$ there exists 
      a strict prefix $S'$ of $S$ that contains all steps of $S$ at depth $\le n$. 
      So, if $\kappa$ is the length of $S'$,
      then all steps in $S$ at index an $\gamma \ge \kappa$ have depth $> n$.
      Let $T'$ be the permutation of $S'$ with respect to $f$.
      By Lemma~\ref{lem:permute:prefix}, $S'$ and $T'$ 
      are permutation equivalent, and by induction hypothesis, they have the same final term, say final term $u$.
      As all steps in $S$ after the last step of $S'$ are at depth $> n$ we obtain $u =_n s_\alpha$.
      Likewise, we get $u =_n t_\beta$ since $T'$ contains all steps of $T$ that are at depth $\le n$.
      Hence $s_\alpha =_n t_\beta$.
  \end{enumerate}
  This concludes the proof.
\end{proof}

\begin{prop}
  Parallel permutation equivalence is an equivalence relation.
\end{prop}

\proof
  We prove reflexivity, symmetry and transitivity:
  \begin{enumerate}[label=\({\alph*}]
    \item \emph{Reflexivity.}
      Parallel permutation equivalence of a rewrite sequence $S : s \to^\alpha t$ to itself
      is witnessed by choosing $f$ as the identity function on $\alpha$ in Definition~\ref{def:permute};
      then both conditions in the definition are trivially fulfilled.
    \item \emph{Symmetry.}
      Let $S : s \to^\alpha t$ be parallel permutation equivalent to $T : s \to^\beta t$
      with witnessing bijection $f : \alpha \to \beta$.
      Then parallel permutation equivalence of $T$ to $S$ is witnessed by $f^{-1}$;
      we check both conditions of Definition~\ref{def:permute}:
      \begin{enumerate}
        \item 
          $\therule{T}{\gamma} = \therule{T}{f(f^{-1}(\gamma))} = \therule{S}{f^{-1}(\gamma)}$ and \\
          $\thepos{T}{\gamma} = \thepos{T}{f(f^{-1}(\gamma))} = \thepos{S}{f^{-1}(\gamma)}$ 
          for every $\gamma < \beta$\\\hphantom{} \hfill (since $f^{-1}(\gamma) < \alpha$)
        \item 
          We have
          \begin{align*}
             \hspace{2cm}&\text{$\gamma_1 < \gamma_2 < \beta$ with $f^{-1}(\gamma_1) > f^{-1}(\gamma_2)$} \\
             &\quad\iff
              \text{$f^{-1}(\gamma_2) < f^{-1}(\gamma_1) < \alpha$ with $f(f^{-1}(\gamma_2)) = \gamma_2 > \gamma_1 = f(f^{-1}(\gamma_1))$}
          \end{align*}
          Hence $\thepos{T}{\gamma_1} \parallel \thepos{T}{\gamma_2}$
          whenever $\gamma_1 < \gamma_2 < \beta$ and $f^{-1}(\gamma_1) > f^{-1}(\gamma_2)$.
      \end{enumerate}
    \item \emph{Transitivity.}
      Let $S : s \to^\alpha t$ be parallel permutation equivalent to $T : s \to^\beta t$
      and $T : s \to^\beta t$ parallel permutation equivalent to $U : s \to^\delta t$
      witnessed by bijections $f : \alpha \to \beta$ and $g : \beta \to \delta$, respectively.      
      Then parallel permutation equivalence of $S$ to $U$ is witnessed by $g \funcomp f$;
      we check both conditions of Definition~\ref{def:permute}:
      \begin{enumerate}
        \item 
          $\therule{S}{\gamma} = \therule{T}{f(\gamma)} = \therule{U}{g(f(\gamma))}$ and\\ 
          $\thepos{S}{\gamma} = \thepos{T}{f(\gamma)} = \thepos{U}{g(f(\gamma))}$ 
          for every $\gamma < \alpha$, and
        \item 
          Assume that $\gamma_1 < \gamma_2 < \alpha$ with $g(f(\gamma_1)) > g(f(\gamma_2))$.
          Then either 
          $f(\gamma_1) > f(\gamma_2)$ or
          $f(\gamma_1) < f(\gamma_2) \;\wedge\; g(f(\gamma_1)) > g(f(\gamma_2))$.
          In the former case, 
          $\thepos{S}{\gamma_1} \parallel \thepos{S}{\gamma_2}$
          follows from parallel permutation equivalence of $S$ to $T$.
          In the latter case,
          it follows from parallel permutation equivalence of $T$ to $U$.
      \qedhere
      \end{enumerate}
  \end{enumerate}\smallskip

\noindent The following lemma implies that the witnessing function $f$ in the definition of
parallel permutation equivalence is unique (for fixed rewrite
sequences $S$ and $T$).
\newpage

\begin{lem}\label{lem:permute:injective}
  Let $S : s \to^\alpha t$ and $T : s \to^\beta t$ be rewrite sequences.
  Let $\alpha' \le \alpha$ and 
  let $f : \alpha' \to \beta$ be an an injective function with the properties
  \begin{enumerate}
    \item \label{lem:permute:injective:i}
      $\therule{S}{\gamma} = \therule{T}{f(\gamma)}$ and 
      $\thepos{S}{\gamma} = \thepos{T}{f(\gamma)}$ 
      for every $\gamma < \alpha'$, and
    \item \label{lem:permute:injective:ii}
      $\thepos{S}{\gamma_1} \parallel \thepos{S}{\gamma_2}$
      whenever $\gamma_1 < \gamma_2 < \alpha'$ and $f(\gamma_1) > f(\gamma_2)$.
    \item \label{lem:permute:injective:iii}
      $\thepos{T}{\gamma_1} \parallel \thepos{T}{\gamma_2}$
      whenever $\gamma_1 < \gamma_2 < \beta$, $\gamma_1$ is not in the image of $f$, but $\gamma_2$ is,
  \end{enumerate}
  Then $f$ is unique with these properties (\,for fixed $S$, $T$, $\alpha'$).
  
  Moreover, among the ordinals~$\le \alpha$,
  there exists a largest ordinal $\alpha'$ such that a function~$f$ with these properties exists.
\end{lem}
Before we prove the lemma, let us give some intuition for the conditions~\ref{lem:permute:injective:i}--\ref{lem:permute:injective:iii}.
Item~\ref{lem:permute:injective:i} ensures that $f$ correctly embeds $S$ into $T$ in the sense it respects redex position and the applied rule.
Condition~\ref{lem:permute:injective:ii} guarantees that $f$ only swaps steps that are at parallel positions.
Finally, condition~\ref{lem:permute:injective:iii} requires that steps in $T$ that are not in the image of $f$
must be parallel to all subsequent steps in $T$ that are in the image of $f$.
\begin{proof}[Proof of Lemma~\ref{lem:permute:injective}]
  We prove uniqueness of $f$ by induction on $\alpha'$

  The base case $\alpha' = 0$ is trivial.

  For $\alpha'$ a successor ordinal, $\alpha' = \alpha'' + 1$, we argue as follows.
  Let $f, g : \alpha' \to \beta$ be injective functions fulfilling the 
  properties~\ref{lem:permute:injective:i}--\ref{lem:permute:injective:iii}.
  Then by induction hypothesis, the functions $f|_{\alpha''}$ and $g|_{\alpha''}$ coincide.
  Assume, for a contradiction, $f(\alpha'') \ne g(\alpha'')$.
  Without loss of generality we may assume $f(\alpha'') < g(\alpha'')$. 
  Then $f(\alpha'') < g(\alpha'') < \beta$ 
  and $f(\alpha'')$ is not in the image of $g$, but $g(\alpha'')$ is.
  However,
  $\thepos{T}{f(\alpha'')} = \thepos{S}{\alpha''} = \thepos{T}{g(\alpha'')}$,
  contradicting property \ref{lem:permute:injective:iii} for the function $g$.
  Hence $f$ and $g$ coincide.
  
  Let $\alpha'$ be a limit ordinal and 
  let $f,g : \alpha' \to \beta$ be injective functions 
  fulfilling the properties~\ref{lem:permute:injective:i}--\ref{lem:permute:injective:iii}. 
  Assume, for a contradiction, that $f(\alpha'') \ne g(\alpha'')$ for some $\alpha'' < \alpha'$.
  However, we have that $\alpha'' + 1 < \alpha'$ since $\alpha'$ is a limit ordinal.
  By induction hypothesis we obtain
  $f|_{\alpha''+1}$ coincides with $g|_{\alpha''+1}$.
  Thus $f(\alpha'') = g(\alpha'')$, contradicting our assumption. 
  Hence, $f$ and $g$ coincide.
  This concludes the proof of uniqueness of $f$.
  
  We write $P(\alpha')$
  if there exists an injective function $f : \alpha' \to \beta$ with 
  properties~\ref{lem:permute:injective:i}--\ref{lem:permute:injective:iii}.
  It remains to be shown that there exists a largest ordinal $\alpha' \le \alpha$
  for which $P(\alpha')$ holds.
  Let $\xi$ be the supremum (union) of all ordinals $\xi' \le \alpha$ for which $P(\xi')$;
  note that this set is non-empty since always $P(0)$ holds.
  If $P(\xi)$, then $\xi$ is the largest of these ordinals.
  Thus assume $\neg P(\xi)$.
  Then $\xi$ is a limit ordinal.
  Note that $\xi'' < \xi'$ and $P(\xi')$ imply $P(\xi'')$.
  As a consequence we have $P(\xi')$ for every $\xi' < \xi$.
  As the length of every rewrite sequence is countable~\cite{tere:2003},
  it follows that $\alpha$ and thus $\xi$ are countable.
  Every countable limit ordinal has cofinality $\omega$.
  Thus there exist ordinals $\xi_1 < \xi_2 < \xi_3 < \ldots$, each of which $< \xi$, 
  such that $\xi = \bigcup \,\{\, \xi_i \mid i \in \nat \,\}$.
  Then $P(\xi_i)$ for every $i \in \nat$.
  For every $i \in \nat$, 
  there exists an injective function $f_i : \xi_i \to \beta$ fulfilling 
  properties~\ref{lem:permute:injective:i}--\ref{lem:permute:injective:iii} for $\alpha' = \xi_i$.
  From the uniqueness (shown above), it follows that $f_i$ coincides with $f_j|_{\xi_i}$ for every $i < j$.
  As a consequence, we can define $f : \xi \to \beta$ by 
  \begin{align*}
    \text{$f(\xi') = f_i(\xi')$ whenever $i \in \nat$ and $\xi' < \xi_i$}\,.
  \end{align*}
  We claim that $f$ is injective and has the properties~\ref{lem:permute:injective:i}--\ref{lem:permute:injective:iii}
  and hence $P(\xi)$ holds; contradicting the above assumption.
  Injectivity and property~\ref{lem:permute:injective:i} are immediate.
  Property~\ref{lem:permute:injective:ii} follows from the fact that
  for $\gamma_1 < \gamma_2 < \xi$
  there exists $i \in \nat$ such that $\gamma_1, \gamma_2 < \xi_i$
  (since $\xi$ is the supremum of the $\xi_i$'s).
  Then $f_i$ fulfilling property~\ref{lem:permute:injective:ii} for $\gamma_1, \gamma_2$
  implies $f$ fulfilling property~\ref{lem:permute:injective:ii} for $\gamma_1, \gamma_2$.
%
%
  Analogously, property~\ref{lem:permute:injective:iii} follows from the following observation:
  whenever $\gamma_1 < \gamma_2 < \beta$, 
  if $\gamma_2$ is in the image of $f$ and $\gamma_1$ is not, then there exists some $i$ such that $\gamma_2$ is in the image of $f_i$ (because $f$ is the union of all $f_j$),
  while $\gamma_1$ is not (since the image of $f$ contains that of all $f_j$).
\end{proof}

\begin{lem}\label{lem:equivalent:form}
  Let $S : s \to^\alpha t$ be a rewrite sequence.
  Then $S$ is of the form
  \begin{align*}
    S = S_1 \relcomp S_2 \relcomp S_3 \relcomp \cdots \relcomp S_{2n+1}
  \end{align*}
  for some $n \in \nat$ such that for every $i \in \{\,1,\ldots,2n\,\}$ we have:
  \begin{enumerate}
    \item for odd $i$, $S_i : s_i \to^\alpha s_{i+1}$ is a reduction below the root,
    \item for even $i$, $S_i : s_i \rstep s_{i+1}$ is a root step. 
  \end{enumerate}
  If $S$ is parallel permutation equivalent to $T : s \to^\beta t$,
  then $T$ is of the form:
  \begin{align*}
    T = T_1 \relcomp T_2 \relcomp T_3 \relcomp \cdots \relcomp T_{2n+1}
  \end{align*}
  such that $S_i$ is parallel permutation equivalent to $T_i$ for every $i \in \{\,1,\ldots,2n+1\,\}$.
\end{lem}

\begin{proof}
  From the definition of parallel permutation equivalence,
  it follows that 
  \begin{itemize}[label=($\star$)]
    \item [($\star$)] no steps can swap the order with a root step.
  \end{itemize}
  As a consequence, $T$ contains the same root steps as $S$, in the same order. 
  Hence $T$ is of the form $T = T_1 \relcomp T_2 \relcomp T_3 \relcomp \cdots \relcomp T_{2n+1}$
  such that for every $i \in \{\,1,\ldots,2n\,\}$ we have:
  \begin{enumerate}
    \item for odd $i$, $T_i : t_i \to^\alpha t_{i+1}$ is a reduction below the root,
    \item for even $i$, $T_i : t_i \rstep t_{i+1}$ is a root step, the same as $S_i$. 
  \end{enumerate}
  From ($\star$) it moreover follows that
  $S_i$ is parallel permutation equivalent to $T_i$ for every $i \in \{\,1,\ldots,2n+1\,\}$.
  The reason is that there cannot be a step $s \in S_i$ with $f(s) \in T_j$ where $i,j$ are odd and $j \neq i$.
  For otherwise, this would imply a swap of $s$ with respect to the root step $S_{i+1}$ (if $i < j$) or the root step $S_{i-1}$ (if $i > j$).
\end{proof}

The next definition extracts the order in which certain rule applications occur in a rewrite sequence.
A \emph{rule application} is formally a pair $\pair{\rho}{p} \in \atrs \times \nat^*$ consisting of a rewrite rule and a position.
A \emph{step} in a rewrite sequence of length $\alpha$ is a triple $\pair{\beta}{\rho,p}$ where $\beta < \alpha$ is an index and $\pair{\rho}{p}$ is a rule application.
Given a rewrite sequence $S$, consider the sequence of rule applications that take place at each step in $S$.
We are interested in the subsequence of all those
$\pair{\rho}{p}$
that fall in a given set $P \subseteq \atrs \times \nat^*$.
\begin{defi}
  Let $S : s \to^\alpha t$ be a rewrite sequence.
  The \emph{rule application sequence of $S$} is the sequence
  $\rulapp(S)\colon \alpha \to \atrs \times \nat^*$ given by
  $\rulapp(S)(\beta) = \pair{\therule{S}{\beta}}{\thepos{S}{\beta}}$ for all $\beta < \alpha$.
  Given a set $P \subseteq \atrs \times \nat^*$ of \emph{rule applications},
  we define the \emph{$P$-projection of $S$} as the
  subsequence $\proj{P}{S}$ of $\rulapp(S)$ obtained by picking all rule applications in $P$.

  In other words, $\proj{P}{S}$ is a function $\proj{P}{S}: \beta \to P$ for some ordinal $\beta \le \alpha$ 
  together with an embedding $f : \beta \to \alpha$ such that:
  \begin{enumerate}
    \item \textit{$f$ is an embedding:}
      $\proj{P}{S}(\gamma) = \pair{\,\therule{S}{f(\gamma)}}{\,\thepos{S}{f(\gamma)}\,}$ for every $\gamma < \beta$, 
    \item \textit{$f$ is increasing:} 
      $f(\gamma_1) < f(\gamma_2)$ whenever $\gamma_1 < \gamma_2 < \beta$, and
    \item \textit{$f$ selects all rule applications in $P$:}
      for all steps $\pair{\gamma,\therule{S}{\gamma}}{\,\thepos{S}{\gamma}\,}$ in $S$, if $\pair{\,\therule{S}{\gamma}}{\,\thepos{S}{\gamma}\,} \in P$ then $\gamma$ is in the image of $f$.
  \end{enumerate}
  We say that rewrite sequences $S_1$ and $S_2$ have the \emph{same order of rule applications in $P \subseteq \atrs \times \nat^*$} if $\proj{P}{S_1} = \proj{P}{S_2}$ (here we mean point-wise equality).
\end{defi}
It should be clear that for given $S$ and $P$, $\proj{P}{S}$ is well-defined.

\begin{exa}
Consider a rewrite system with rules
$\begin{array}[t]{clcl}
  (\rho_1) & \fun{a}(x) & \to & \fun{a}(\fun{a}(x))\\
  (\rho_2) & \fun{b}(x) & \to & \fun{b}(\fun{b}(x))\\
  (\rho_3) & \fun{f}(x,y) & \to & \fun{f}(\fun{a}(\fun{c}),\fun{b}(\fun{c}))
\end{array}$\\
Then we can rewrite
\[\begin{array}[t]{lcll}
S: \fun{f}(\fun{a}(\fun{c}),\fun{b}(\fun{c})) & 
 \to & \fun{f}(\fun{a}(\fun{a}(\fun{c})),\fun{b}(\fun{c}))       & \text{using } \pair{\rho_1}{0}\\
& \to & \fun{f}(\fun{a}(\fun{a}(\fun{a}(\fun{c}))),\fun{b}(\fun{c}))    & \text{using } \pair{\rho_1}{0}\\
& \to & \fun{f}(\fun{a}(\fun{a}(\fun{a}(\fun{c}))),\fun{b}(\fun{b}(\fun{c}))) & \text{using } \pair{\rho_2}{1}\\
& \to & \fun{f}(\fun{a}(\fun{c}),\fun{b}(\fun{c}))          & \text{using } \pair{\rho_3}{\epsilon}\\
& \to & \fun{f}(\fun{a}(\fun{c}),\fun{b}(\fun{b}(\fun{c})))       & \text{using } \pair{\rho_2}{1}\\
& \to & \fun{f}(\fun{a}(\fun{a}(\fun{c})),\fun{b}(\fun{b}(\fun{c})))    & \text{using } \pair{\rho_1}{0}\\
\end{array}
\]
Taking $P = \{\pair{\rho_1}{0}, \pair{\rho_2}{1}\}$, we get that 
$\proj{P}{S} =  \pair{\rho_1}{0}, \pair{\rho_1}{0}, \pair{\rho_2}{1}, \pair{\rho_2}{1}, \pair{\rho_1}{0}$,
and, in particular, $\beta=5$.
\end{exa}

The following lemma states, for a given proof tree, 
the order of non-parallel rule applications is the same in every rewrite sequence corresponding to the proof tree.
\begin{lem}\label{lem:order:non:parallel}
  Let $\delta$ be a proof tree (Definition~\ref{def:ired:restrict})
  and let $P \subseteq \atrs \times \nat^*$ be a set of rule applications
  such that for every $\pair{\rho_1}{p_1}, \pair{\rho_2}{p_2} \in P$ 
  we have $p_1 \nparallel p_2$. 
  If rewrite sequences $S$ and $T$ both correspond to $\delta$, then $\proj{P}{S} = \proj{P}{T}$.
\end{lem}

\begin{proof}
  It suffices to consider the case that $P$ consists of at most $2$ elements since:
  \begin{align*}
    \proj{P}{S} = \proj{P}{T} \;\;\Longleftrightarrow\;\; \forall Q \subseteq P.\, ({{|}Q{|}} \le 2 \;\implies\; \proj{Q}{S} = \proj{Q}{T}) \;.
  \end{align*}
We prove that
for all proof trees $\delta$ and all $P \subseteq \atrs \times \nat^*$ with ${|}P{|}\leq 2$,
if rewrite sequences $S$ and $T$ correspond to $\delta$, then
$\proj{P}{S} = \proj{P}{T}$. 
So let $\delta, S, T,P$ be as stated, and assume that
$P = \{\pair{\rho_1}{p_1},\pair{\rho_2}{p_2}\}$; we allow $\pair{\rho_1}{p_1}=\pair{\rho_2}{p_2}$.
Without loss of generality, we assume that $p_1 \le p_2$.
The proof is by well-founded induction on the length of $p_2$.
First we note that if $\delta$ is a single root step or a single application of $\rid$, the result is immediate. 

\emph{Base case:} Here we have $p_2=p_1=\varepsilon$ which means that $P$ contains only root step applications.
We distinguish cases according to the root of $\delta$.
If the root of $\delta$ is obtained from a $\rsplit$-application,
then the root steps in $S$ and $T$ must occur in the same order,
hence $\proj{P}{S} = \proj{P}{T}$. 
In case the root of $\delta$ is obtained from a $\rlift$-application,
then $S$ and $T$ cannot contain any root steps, hence we are also done.

\emph{Induction step:} 
Assume that $p_2 = ip_2'$ and define
\begin{align*}
  P' = \begin{cases}
    \{\pair{\rho_2}{p_2'}\}, &\text{if } p_1 = \varepsilon \\
    \{\pair{\rho_1}{p_1'}, \pair{\rho_2}{p_2'}\},  &\text{if } p_1 = ip_1'
  \end{cases}
\end{align*}
By induction, we may assume that the property holds for the set of rule applications $P'$.

For $(\rho,p) \in (\atrs \times \nat^*)$, we define $i\cdot (\rho,p) = (\rho,ip)$.
For functions $h : \beta \to \atrs \times \nat^*$, 
we define $(i \cdot h) : \beta \to \atrs \times \nat^*$ by 
by $(i \cdot h)(\gamma) = i \cdot h(\gamma)$.

We distinguish cases for the shape of $\delta$.
If $\delta$ is of the form:
\begin{align*}
        \infer=[\rlift]
        {f(s_1,\ldots,s_n) \ireddownmfin f(t_1,\ldots,t_n)}
        {\delta_1 : s_1 \ired t_1 \quad \delta_2 : s_2 \ired t_2 \quad\cdots\quad \delta_n : s_n \ired t_n}
\end{align*}
then $S$, respectively $T$, is an interleaving of 
rewrite sequences $S_1,\ldots,S_n$, respectively $T_1, \ldots, T_n$,
corresponding to $\delta_1,\ldots,\delta_n$.
In particular, $S_i \colon s_i \ired t_i$ and $T_i \colon s_i \ired t_i$ correspond to $\delta_i$.
Since $S$ cannot contain root steps, we have that
$\pair{\rho_1}{\varepsilon}$ is not in the image of $\proj{P}{S}$,
so if $p_1=\varepsilon$ then
$\proj{P}{S} = i\cdot \proj{P'}{S_i}$.
If $p_1=i p_i'$, then 
all $\rho_1$- and $\rho_2$-applications in $P$ are on the $i$'th subterm,
and since interleaving preserves order of steps in this subterm,
we again have that 
$\proj{P}{S} = i\cdot \proj{P'}{S_i}$.
Similarly, $\proj{P}{T} = i\cdot \proj{P'}{T_i}$.
Now, $\proj{P}{S}=\proj{P}{T}$ follows from the induction hypothesis.

If $\delta$ is of the form
\begin{align*}
        \infer=[\rsplit]
        {s \ired t}
        {\delta_1 \quad \delta_2  \quad\cdots\quad \delta_n}
\end{align*}
then $S$, respectively $T$, is a concatenation of $S_1, \ldots, S_n$, respectively $T_1, \ldots, T_n$,
corresponding to the trees $\delta_1,\ldots,\delta_n$.
Hence it suffices to show that $\proj{P}{S_i} = \proj{P}{T_i}$ for all $i=1,\ldots,n$.
For all $\delta_i$'s that are a single node or an $\rid$-application, this holds.
The remaining $\delta_i$'s must have a root that is obtained from a $\rlift$-application,
and the result now follows from the $\rlift$-case.
\end{proof}

Next, we characterise the correspondence of canonical proof trees and rewrite sequences.

\begin{thm}\label{thm:canonical}
  We have the following facts about canonical proof trees:
  \begin{enumerate}
    \item \label{thm:canonical:a}
      Every rewrite sequence corresponds to  
      precisely one canonical proof tree.
    \item \label{thm:canonical:b}
      Rewrite sequences $S : s_0 \to^\alpha_\atrs s_\alpha$ and $T : t_0 \to^\beta_\atrs t_\beta$ 
      are represented by the same canonical proof tree 
      if and only if they are parallel permutation equivalent.  
  \end{enumerate}
\end{thm}

\begin{proof}
  First, note that every rewrite sequence corresponds to a canonical proof tree.
  This follows by an inspection of the proof of Theorem~\ref{thm:ired:equiv}:
  for every rewrite sequence $S : s \to^\alpha t$,
  the proof tree $\mktree{S}$ is canonical and corresponds to $S$.

  Second, we show that canonical proof trees coincide if they correspond to 
  rewrite sequences that are parallel permutation equivalent.
  As a direct consequence, we obtain that every rewrite sequence corresponds to precisely one canonical proof tree,
  establishing \ref{thm:canonical:a}.

  \emph{Notation:} 
  For a strongly convergent rewrite sequence $R : s \redord^\alpha t$,
  we write $\mathfrak{T}_R$ to denote \emph{an arbitrary} canonical proof tree $\mathfrak{T}_R : s \ired t$ corresponding to $R$.
  Note that, a priori, the tree $\mathfrak{T}_R$ can be different from $\mathfrak{T}(R)$.
  If $R$ is a rewrite sequence below the root, 
  then we moreover write $\mathfrak{T}'^{(<)}_{R}$ for a canonical proof tree $\mathfrak{T}'^{(<)}_{R} : s \ireddownmfin t$ corresponding to $R$.

  Let $S : s \redord^\alpha t$ and $T : s \redord^\beta t$ be rewrite sequences 
  that are parallel permutation equivalent.
  We show that $\mathfrak{T}_S = \mathfrak{T}_T$ by coinduction, that is:
  if $S$ and $T$ are parallel permutation equivalent,
  then $\mathfrak{T}_S$ and  $\mathfrak{T}_T$ have the same root 
  and all subtrees arise again from parallel permutation equivalent rewrite sequences.
  As $S$ and $T$ are parallel permutation equivalent, 
  from Lemma~\ref{lem:equivalent:form} it follows that
  $S$ and $T$ are of the forms:
  \begin{gather}
    \begin{aligned}
    S &= S_1 \relcomp S_2 \relcomp S_3 \relcomp \cdots \relcomp S_{2n+1}\\
    T &= T_1 \relcomp T_2 \relcomp T_3 \relcomp \cdots \relcomp T_{2n+1}
    \end{aligned}
    \label{eq:form:ST}
  \end{gather}
  for some $n \in \nat$ such that for every $i \in \{\,1,\ldots,2n+1\,\}$ we have:
  \begin{enumerate}[label=\({\alph*}]
    \item \label{it:form:ST:a} for odd $i$, $S_i$ and $T_i$ are reductions below the root,
    \item \label{it:form:ST:b} for even $i$, $S_i$ and $T_i$ are a root steps, and
    \item \label{it:form:ST:c} $S_i$ is parallel permutation equivalent to $T_i$ for every $i \in \{\,1,\ldots,2n+1\,\}$.
  \end{enumerate}
  In particular, by Lemma~\ref{lem:same:target}, $S_i$ and $T_i$ have the same source and target term for every $i \in \{\,1,\ldots,2n\,\}$,
  say source $u_i$ and target $u_{i+1}$.

  We consider the root of the proof trees $\mathfrak{T}_S : s \ired t$ and $\mathfrak{T}_T : s \ired t$.
  The only way to derive $\ired$ is by an application of the split-rule.
  Due to \eqref{eq:form:ST}, \ref{it:form:ST:a} and \ref{it:form:ST:b}, 
  and the form of the canonical split-rule, it follows that $\mathfrak{T}_{S}$ and $\mathfrak{T}_{T}$ must be of the form:
  \begin{align*}
    \raisebox{2ex}{$\mathfrak{T}_{S}$} &\;\raisebox{2ex}{$=$}\; \infer=[\rsplit]{u_1 \ired u_{2n+2}}{
      \mathfrak{T}'^{<}_{S_1}
      & u_2 \rstep u_3 
      & \mathfrak{T}'^{<}_{S_3}
      & u_4 \rstep u_5
      & \cdots 
      & u_{2n} \rstep u_{2n+1}
      & \mathfrak{T}'_{S_{2n+1}}
    } \\
    \raisebox{2ex}{$\mathfrak{T}_{T}$} &\;\raisebox{2ex}{$=$}\; \infer=[\rsplit]{u_1 \ired u_{2n+2}}{
      \mathfrak{T}'^{<}_{T_1}
      & u_2 \rstep u_3 
      & \mathfrak{T}'^{<}_{T_3}
      & u_4 \rstep u_5
      & \cdots 
      & u_{2n} \rstep u_{2n+1}
      & \mathfrak{T}'_{T_{2n+1}}
    } 
  \end{align*}
  As a consequence $\mathfrak{T}_{S}$ and $\mathfrak{T}_{T}$ have the same root
  and, due to (c), the subtrees arise from rewrite sequences below the root that are again parallel permutation equivalent.
  
  Let $S : s \redord^\alpha t$ and $T : s \redord^\beta t$ be rewrite sequences below the root
  that are parallel permutation equivalent.
  We consider the root of the proof trees $\mathfrak{T}'^{(<)}_S$ and $\mathfrak{T}'^{(<)}_T$
  and observe that $\ireddownmfin$ can only be derived using the id-rule or the lift-rule.
  In case one of the trees is derived using the canonical id-rule, 
  it follows that $s = t = x$ for some variable $x \in \avars$,
  and there is only one possible proof tree deriving $x \ireddownmfin x$:
  \begin{align*}
    \raisebox{2ex}{$\mathfrak{T}'^{(<)}_S = \mathfrak{T}'^{(<)}_T =$} \;\;
    \infer=[\rid]{x \ireddownmfin x}{}
  \end{align*}
  Thus, assume that both $\mathfrak{T}'^{(<)}_S$ and $\mathfrak{T}'^{(<)}_T$ are derived using the lift-rule.
  Then $s = f(s_1,\ldots,s_n)$ and $t = f(t_1,\ldots,t_n)$ for some $f \in \Sigma$ of arity $n$ 
  and terms $s_1,\ldots,s_n,t_1,\ldots,t_n$.
  The proof trees must be of the following forms:
  \begin{align*}
    \raisebox{2ex}{$\mathfrak{T}'^{(<)}_{S}$} &\;\raisebox{2ex}{$=$}\; \infer=[\rlift]{f(s_1,\ldots,s_n) \ireddownmfin f(t_1,\ldots,t_n)}{
      \mathfrak{T}_{S_1} & \ldots & \mathfrak{T}_{S_n}
    } \\
    \raisebox{2ex}{$\mathfrak{T}'^{(<)}_{T}$} &\;\raisebox{2ex}{$=$}\; \infer=[\rlift]{f(s_1,\ldots,s_n) \ireddownmfin f(t_1,\ldots,t_n)}{
      \mathfrak{T}_{T_1} & \ldots & \mathfrak{T}_{T_n}
    } 
  \end{align*}
  Where, for $i \in \{\,1,\ldots,n\,\}$,
  $S_i$ is the subsequence of $S$ on the $i$-th argument of $f$,
  and
  $T_i$ is the subsequence of $T$ on the $i$-th argument of $f$.
  Since $S$ and $T$ are parallel permutation equivalent, 
  it follows that $S_i$ and  $T_i$ are, for every $i \in \{\,1,\ldots,n\,\}$.
  As a consequence, $\mathfrak{T}'^{(<)}_{S}$ and $\mathfrak{T}'^{(<)}_{T}$ have the same root
  and the subtrees arise from rewrite sequences that are parallel permutation equivalent.
  This concludes the coinduction, hence $\mathfrak{T}_S = \mathfrak{T}_T$.
  
  It remains to be shown that rewrite sequences that are not parallel permutation equivalent
  have different canonical proof trees. 
  Let $S : s \redord^\alpha t$ and $T : s \redord^\beta t$ be rewrite sequences 
  that are not parallel permutation equivalent.
  Note that, due to strong convergence, $S$ and $T$ contain only a finite number of steps at every depth $n$.
  Moreover, we may assume that:
  \begin{itemize}[label=($\star$)]
    \item [($\star$)] For every rule $\rho \in \atrs$ and position $\apos \in \nat^*$,
      $S$ and $T$ contain the same number of steps arising from an application of $\rho$ at position $\apos$.
  \end{itemize}
  If ($\star$) was violated, then $S$ and $T$ cannot correspond to the same proof tree.
  The reason is that, from a given proof tree one can derive 
  the steps at position~$p$ with respect to rule~$\rho$.
  
  By Lemma~\ref{lem:permute:injective} there exists a largest ordinal $\alpha' \le \alpha$
  such that there exists an injective function $f : \alpha' \to \beta$ with the properties
  \begin{enumerate}
    \item 
      $\therule{S}{\gamma} = \therule{T}{f(\gamma)}$ and 
      $\thepos{S}{\gamma} = \thepos{T}{f(\gamma)}$ 
      for every $\gamma < \alpha'$, and
    \item 
      $\thepos{S}{\gamma_1} \parallel \thepos{S}{\gamma_2}$
      whenever $\gamma_1 < \gamma_2 < \alpha'$ and $f(\gamma_1) > f(\gamma_2)$.
    \item 
      $\thepos{T}{\gamma_1} \parallel \thepos{T}{\gamma_2}$
      whenever $\gamma_1 < \gamma_2 < \beta$, $\gamma_1$ is not in the image of $f$, but $\gamma_2$ is,
  \end{enumerate}
  and this function $f$ is unique.
  As $S$ and $T$ are not parallel permutation equivalent, 
  it follows that either
  \begin{enumerate}[label=\({\alph*}]
    \item $\alpha' < \alpha$, or
    \item $\alpha'= \alpha$ and there exists $\beta' < \beta$ that is not in the image of $f$.
  \end{enumerate}
  Note that (b) would contradict ($\star$) as all of the steps of $S$ are in the domain of $f$
  but the image of $f$ does not contain all of the steps of $T$.
  Then $T$ would have to contain more steps than $S$ at some position $p \in \nat^*$.

  Thus assume that $\alpha' < \alpha$.
  Then there exists some index $\gamma < \beta$ such that $\gamma$ is not in the image of~$f$,
  $\therule{S}{\alpha'} = \therule{T}{\gamma}$ and $\thepos{S}{\alpha'} = \thepos{T}{\gamma}$.
  For otherwise, if $\gamma$ would not exist,
  we had a contradiction with ($\star$) as then
  $S$ contained more steps than $T$ 
  at position $\thepos{S}{\alpha'}$ with respect to rule $\therule{S}{\alpha'}$.
  We define $g : \alpha' + 1 \to \beta$ by $g(\xi) = f(\xi)$ for every $\xi < \alpha'$ and $g(\alpha') = \gamma$.   
  Then $g$ is injective and satisfies property~\ref{lem:permute:injective:i}.
  For property~\ref{lem:permute:injective:ii} of $g$, we only need to consider the case that $\gamma_2 = \alpha'$,
  so $\gamma_1 < \gamma_2 = \alpha' < \alpha'+1$ 
  and $g(\gamma_2) = \gamma < g(\gamma_1)$.
  Then $\gamma < f(\gamma_1)$ and $\gamma$ is not in the image of $f$, while $f(\gamma_1)$ is.
  Then it follows from property~\ref{lem:permute:injective:iii} of $f$ 
  that 
  $\thepos{T}{\gamma} \parallel \thepos{T}{f(\gamma_1)}$.
  Hence $\thepos{T}{g(\gamma_2)} \parallel \thepos{T}{g(\gamma_1)}$
  and $\thepos{S}{\gamma_2} \parallel \thepos{S}{\gamma_1}$.
  Thus $g$ satisfies property~\ref{lem:permute:injective:ii}.
  However, due to the choice of $\alpha'$, 
  $g$ cannot satisfy all properties~\ref{lem:permute:injective:i}--~\ref{lem:permute:injective:iii};
  hence property~\ref{lem:permute:injective:iii} has to fail.
  For property~\ref{lem:permute:injective:iii} of $g$, 
  it suffices to consider the case that $\gamma_2 = \gamma$
  since $\gamma$ is the only element in the image of $g$ that is not in the image of $f$.
  So, due to failure of property~\ref{lem:permute:injective:iii} for $g$, we have: 
  \begin{itemize}[label=($\dagger$)]
    \item [($\dagger$)]
      There exist $\gamma_1 < \gamma$ such that
      $\gamma_1$ is not in the image of~$g$
      and $\thepos{T}{\gamma_1} \nparallel \thepos{T}{\gamma}$.
      Note that, in particular, $\gamma > 0$ since otherwise property (iii) would hold.
  \end{itemize}
  Define a set of rule applications $P = \{\, \chi_1,\, \chi_2 \,\}$, where
 \begin{align*}
   \chi_1 &= \pair{\therule{T}{\gamma_1}}{\thepos{T}{\gamma_1}}\\
  \chi_2 &= \pair{\therule{T}{\gamma}}{\thepos{T}{\gamma}}
\end{align*}
We compare the order of rule applications in $P$ 
  in the rewrite sequences $S$ and $T$.
  From properties~\ref{lem:permute:injective:ii} and~\ref{lem:permute:injective:iii}
  it follows that $f$ maps
  the $i$-th $\chi_1$-step in $S$ to the $i$-th  $\chi_1$-step in $T$,
  and likewise 
  the $i$-th $\chi_2$-step in $S$ to the $i$-th  $\chi_2$-step in $T$.
  Then it follows from ($\dagger$) that $S$ and $T$ 
  do not contain the same order of rule applications in $P$.
  By Lemma~\ref{lem:order:non:parallel} it follows that $S$ and $T$ 
  cannot arise from the same proof tree.
\end{proof}

\section{A Formalization in Coq}\label{sec:coq}

The standard definition of infinitary rewriting,
using ordinal length rewrite sequences and strong convergence at limit ordinals,
is difficult to formalize. 
The coinductive framework we propose, is easy to formalize and work with in theorem provers.
We discuss the important steps of the formalisation of infinitary rewriting
and the compression lemma.

\subsection{Formalisation of Relations and Vectors}

We have formalised binary relations
and properties of relations, such as reflexivity, transitivity, inclusion and equality, as follows.

\begin{formalisation}[Relations and properties of relations]\mbox{}
\begin{itemize}[leftmargin=4ex]
\item []
{\small
\begin{verbatim}
Variables A : Type.

Definition relation := A -> A -> Prop.

Definition reflexive  (R : relation) := forall x, R x x.
Definition transitive (R : relation) := forall x y z, R x y -> R y z -> R x z.
Definition subrel  (R R' : relation) := forall x y, R x y -> R' x y.
Definition eqrel   (R R' : relation) := subrel R R' /\ subrel R' R.
\end{verbatim}}
\end{itemize}
Note that we have formalised relations as functions.
\end{formalisation}

We have moreover formalised operations on relations
such as composition, union and the reflexive transitive closure. 
\begin{formalisation}[Operations on relations]\mbox{}
\begin{itemize}[leftmargin=4ex]
\item []
{\small
\begin{verbatim}
Definition compose (R S : relation) : relation := 
  fun a c =>  exists b, R a b /\ S b c.

Definition Runion  (R S : relation) : relation :=
  fun a b =>  R a b \/ S a b.

Inductive refl_trans_close (R : relation) : relation :=
  | refl_trans_step  : subrel R (refl_trans_close R)
  | refl_trans_refl  : reflexive (refl_trans_close R) 
  | refl_trans_trans : transitive (refl_trans_close R).

Notation "R ;; S" := (compose R S) (right associativity).
Notation "R (+) S" := (Runion R S) (right associativity).
Notation "R *" := (refl_trans_close R) (left associativity).
\end{verbatim}}
\end{itemize}
\end{formalisation}

\noindent Our formalisation of vectors is based on the formalisation of vectors by Pierre Boutillier
(in the Coq standard library).
Thereby a vector $v$ of length $n$ with elements from $A$ is a function $v : \texttt{Fin n} \to A$.
Here \texttt{Fin n} is an $n$-element set;
for example \texttt{Fin 4} consists of the following elements
\begin{align*}
  \texttt{Fin 4} = \{\; 
    &\texttt{First 3}, \\
    &\texttt{Next (First 2)}, \\
    & \texttt{Next (Next (First 1))} \\
    & \texttt{Next (Next (Next (First 0)))} \;\} \\
\end{align*}

\begin{formalisation}[Vectors and a map operation]\mbox{}
\begin{itemize}[leftmargin=4ex]
\item []
{\small
\begin{verbatim}
Inductive Fin : nat -> Type :=
  | First : forall n, Fin (S n)
  | Next  : forall n, Fin n -> Fin (S n).

Definition vector (n : nat) :=  Fin n -> A.
  
Definition vmap (n : nat) (f : A -> B) : vector A n -> vector B n :=
  fun v i =>  f (v i).
\end{verbatim}}
\end{itemize}
\end{formalisation}

\subsection{Formalisation of Infinite Terms}

For the formalisation of terms, we begin with the set of variables and the signature.
\begin{formalisation}[Variables and Signature]\mbox{}
\begin{itemize}[leftmargin=4ex]
\item []
{\small
\begin{verbatim}
Record variables : Type := Variables {
  variable :> Type;
  beq_var : variable -> variable -> bool;
  beq_var_ok : forall x y, beq_var x y = true <-> x = y
}.

Record signature : Type := Signature {
  symbol :> Type;
  arity : symbol -> nat;
  beq_symb : symbol -> symbol -> bool;
  beq_symb_ok : forall x y, beq_symb x y = true <-> x = y
}.  
\end{verbatim}}
\end{itemize}
Notably, next to the set of variables and functions symbols itself,
our formalisation includes functions \texttt{beq\_var} and \texttt{beq\_symb}
for a \emph{decidable equality} on the variables and function symbols, respectively.
The functions \texttt{beq\_var\_ok} and \texttt{beq\_symb\_ok} guarantee that
the decidable equality coincides with the standard equality in Coq.
\end{formalisation} 
\noindent

Recall that a term is either
\begin{enumerate}
  \item a variable $x \in \avars$, or
  \item a function symbol $f \in \Sigma$ together with a vector of terms of length $\arity{f}$. 
\end{enumerate} 
The inductive interpretation of this principle yields the finite terms (\texttt{finite\_term} in Coq), 
the coinductive interpretation gives rise to the finite and infinite terms (\texttt{term} in Coq).
\begin{formalisation}[Finite and Infinite Terms]\mbox{}
\begin{itemize}[leftmargin=4ex]
\item []
{\small
\begin{verbatim}
Variable F : signature.
Variable X : variables.

Inductive finite_term : Type :=
  | FVar : X -> finite_term
  | FFun : forall f : F, vector finite_term (arity f) -> finite_term.

CoInductive term : Type :=
  | Var : X -> term
  | Fun : forall f : F, vector term (arity f) -> term.
\end{verbatim}}
\end{itemize}
\end{formalisation}

\noindent In Coq, there is no extensional equality, that is, $\forall x.\; f(x) = g(x)$ does not imply $f = g$.
Similarly, infinite terms $s,t$ that coincide on every position, are not necessarily equal
with respect to the standard equality $=$ in Coq. 
As a consequence, equality $=$ is not suitable to work with infinite terms in Coq.
We therefore work with bisimilarity $\xbis$, as is common practice in coalgebra.
Terms are bisimilar if and only if they coincide on every position.
\begin{formalisation}[Bisimilarity on terms]\mbox{}
\begin{itemize}[leftmargin=4ex]
\item []
{\small
\begin{verbatim}
CoInductive term_bis : term -> term -> Prop :=
  | Var_bis : forall x, term_bis (Var x) (Var x)
  | Fun_bis : forall f v w, (forall i, term_bis (v i) (w i)) ->
                term_bis (Fun f v) (Fun f w).

Infix " [~] " := term_bis (no associativity, at level 70).
\end{verbatim}}
\end{itemize}
In the sequel, we write $\xbis$ for bisimilarity of terms in Coq.
\end{formalisation}

\subsection{Formalisation of Term Rewriting Systems}
Next, we formalise rewrite rules and term rewriting systems as lists of rules. 
In the following definition, we leave \texttt{is\_var}, \texttt{vars}, \texttt{inc} and \texttt{list} implicit;
the function 
\begin{enumerate}
  \item \texttt{is\_var $:$ term $\to$ bool} returns \texttt{true} if and only if the given term is a variable,
  \item \texttt{vars $:$ finite\_term $\to$ (X $\to$ Prop)} returns the set of variables in a term,
  \item \texttt{inc} stands for set inclusion, and
  \item \texttt{list} is the implementation of lists in the standard library of Coq. 
\end{enumerate}

\begin{formalisation}[Rules and term rewriting systems]
These are formalized as
\begin{itemize}[leftmargin=4ex]
\item []
{\small
\begin{verbatim}
Record rule : Type := Rule {
  lhs     : finite_term;
  rhs     : term;
  rule_wf : is_var lhs = false /\ inc (vars rhs) (vars lhs)
}.

Definition trs := list rule.
\end{verbatim}}
\end{itemize}
A rule consists of a \emph{finite} left-hand side (\texttt{lhs}) and \emph{finite or infinite} right-hand side (\texttt{rhs}), and 
a proposition \texttt{rule\_wf} stating that the \texttt{lhs} is not a variable,
and the variables in \texttt{rhs} are a subset of the variables in \texttt{lhs}.
\end{formalisation}

\begin{rem}
  We note that the definition of \texttt{rules}
  could easily be generalised to infinite left-hand sides.
  This is not a restriction of our coinductive framework for infinitary rewriting.
  In the literature, \emph{infinitary term rewriting systems}
  are typically defined to have finite left-hand sides
  to keep matching (with respect to left-linear rules) decidable.
  We have chosen to adopt this restriction as
  our goal was a formalisation of the Compression Lemma
  which only holds for term rewrite systems with finite (and linear) left-hand sides.
\end{rem}

We introduce substitutions as maps from variables to finite or infinite terms. 
\begin{formalisation}[Substitution]\mbox{}
\begin{itemize}[leftmargin=4ex]
\item []
{\small
\begin{verbatim}
Definition substitution := X -> term.

CoFixpoint substitute (sigma : substitution) (t : term) : term :=
  match t with
  | Var x      => sigma x
  | Fun f args => Fun f (vmap (substitute sigma) args)
  end.
\end{verbatim}}
\end{itemize}
The \texttt{substitute} function applies a substitution to a finite or infinite term.
\end{formalisation}

\subsection{Formalisation of Infinitary Rewriting}

The closure of the rules under substitutions gives rise to root steps
on finite and infinite terms. 
\begin{formalisation}[Root steps on finite and infinite terms]\mbox{}
\begin{itemize}[leftmargin=4ex]
\item []
{\small
\begin{verbatim}
Variable system : trs.

Inductive root_step : relation term :=
  | Root_step : 
      forall (s t : term) (r : rule) (u : substitution),
        In r system ->
        substitute u (lhs r) [~] s ->
        substitute u (rhs r) [~] t ->
          root_step s t.
\end{verbatim}}
\end{itemize}
As discussed above, we must include bisimilarity $\xbis$
since the equality $=$ in Coq is too strict.
Here `\texttt{In r system}' checks whether \texttt{r} (a rule) occurs in \texttt{system} (a list of rules).
\end{formalisation}

We formalise the lifting operation $\down{R}$ (Definition~\ref{def:lifting}) as follows.
\begin{formalisation}[Lifting]\mbox{}
\begin{itemize}[leftmargin=4ex]
\item []
{\small
\begin{verbatim}
Inductive lift (R : relation term) : relation term :=
  | lift_id : forall s t, s [~] t -> lift R s t
  | lift_step :
      forall (f : F) (s t : vector term (arity f)),
             (forall i, R (s i) (t i)) -> 
             forall fs ft, Fun f s [~] fs -> Fun f t [~] ft -> lift R fs ft.
\end{verbatim}}
\end{itemize}
Again, we include bisimilarity $\xbis$ instead of the standard equality $=$ in Coq.
\end{formalisation}

We use the root step rewrite relation and the lifting operation to introduce infinitary strongly convergent reductions.
Our coinductive definition (Definition~\ref{def:ired:fixedpoint})
is based on mixed induction and coinduction:
\begin{align*}
  {\ired} \;\;\defd\;\; \lfp{R}{\gfp{S}{(\rstep \cup \mathrel{\down{R}})^*\relcomp \down{S}}} \,.
\end{align*}
Unfortunately, Coq has no support for mutual inductive and coinductive definitions.

To overcome this problem, we employ the fact that the greatest fixed point $\gfp{S}{F(S)}$ 
is the union of all $S$ for which $S \subseteq F(S)$ (under the condition that $F : L \to L$ is monotone);
see further Section~\ref{sec:coinduction}.
In other words, 
$\gfp{S}{F(S)}$
is the smallest relation $T$ such that
$S \subseteq T$ whenever $S \subseteq F(S)$.
Hence, we have
\begin{align}
  {\ired} \;\;\defd\;\; 
    \lfp{R}{\;\Bigg(\;\;
  \begin{aligned}
      &\text{the smallest $T$ such that: } \\[-1ex]
      &\text{for all relations $S$, } \\[-1ex]
      &\quad S \subseteq (\rstep \cup \mathrel{\down{R}})^*\relcomp \down{S} \implies S \subseteq T
  \end{aligned}
    \;\;\Bigg)}
  \label{eq:ired:union}
\end{align} 
This definition lends itself to a formalisation in Coq.
\begin{formalisation}\label{form:ired:wrong}
Formalisation of~\eqref{eq:ired:union}:
\begin{itemize}[leftmargin=4ex]
\item []
{\small
\begin{verbatim}
Inductive ired : relation (term F X) :=
  | Ired : 
      forall S : relation (term F X),
      subrel S ((root_step (+) lift ired)* ;; lift S) ->
      subrel S ired.
\end{verbatim}}
\end{itemize}
Here 
\verb=;;= is relation composition,
\verb=(+)= is the union, and
\verb=*= the reflexive-transitive closure.
The statement \texttt{Inductive ired} in the formalisation corresponds to $\slfp{R}$ in~\eqref{eq:ired:union},
and thus \texttt{ired} corresponds to $R$.
\end{formalisation}

While Formalisation~\ref{form:ired:wrong} is correct, 
it turns out that Coq is not able to generate a good induction
principle from the definition.
The generated induction principle is:
\begin{itemize}[leftmargin=4ex]
\item []
{\small
\begin{verbatim}
ired_ind :  forall P : term -> term -> Prop,
            ( forall S : relation term, 
              subrel S ((root_step (+) lift ired)* ;; lift S) -> subrel S P )
            -> subrel ired P
\end{verbatim}}
\end{itemize}
In mathematical notation this reads as follows:
\begin{align*}
  \forall P.\; \Big(\; 
    \forall S.\; \Big(\;S \subseteq (\rstep \cup \mathrel{\down{\texttt{ired}}})^*\relcomp \down{S} \;\Big) \;\to\; S \subseteq P
  \;\Big) \;\;\to\;\; \texttt{ired} \subseteq P
\end{align*}
Note that, in particular we have $\texttt{ired} = (\rstep \cup \mathrel{\down{\texttt{ired}}})^*$.
As a consequence, we already have to show $\texttt{ired} \subseteq P$ as part of the induction step.
Hence, this induction principle is void.

To overcome this problem, we adapt~\eqref{eq:ired:union} and Formalisation~\ref{form:ired:wrong} as follows:
\begin{align}
  {\ired} \;\;\defd\;\; 
    \lfp{R}{\;\Bigg(\;\;
  \begin{aligned}
      &\text{the smallest $T$ such that: } \\[-1ex]
      &\text{for all relations $S$ and $I$, }\\[-1ex]
      &\quad I \subseteq R \implies
       S \subseteq (\rstep \cup \mathrel{\down{I}})^*\relcomp \down{S} \implies S \subseteq T
  \end{aligned}
    \;\;\Bigg)}
  \label{eq:ired}
\end{align} 
To help Coq generate a good induction principle,
we introduce an auxiliary relation $I$, and
we replace the occurrences of $R$ in the body of the definition by $I$.
To preserve the semantics of the definition, 
we require $\texttt{I} \subseteq \texttt{ired}$.
(In other words, $I$ is a lower-approximant of $R$.)

To see that~\eqref{eq:ired:union} and~\eqref{eq:ired} 
give rise to the same relation $\ired$, we argue as follows.
If we were to replace $I \subseteq R$ by $ I = R$ in~\eqref{eq:ired}, 
then both definitions would obviously coincide.
However, both definitions also coincide without the replacement
since sets $I \subsetneq R$ do not contribute due to monotonicity.

We formalise~\eqref{eq:ired} as follows.
\begin{formalisation}[Strongly convergent rewrite relation]\footnote{%
  Note that canonical proof trees could be formalised by a small adaptation of this formalization.
  On the one hand, the restricted (canonical) proof trees can simplify proofs whose source is (a proof tree for) an infinite reduction.
  On the other hand, the restriction might complicate proofs whose target is (a proof tree for) an infinite reduction.
  To combine the advantages of both choices,
  it would be interesting to formally prove that every proof tree can be transformed into an equivalent canonical proof tree.
}\mbox{}
\begin{itemize}[leftmargin=4ex]
\item []
{\small
\begin{verbatim}
Inductive ired : relation term :=
  | Ired : 
      forall S I : relation term,
      subrel I ired ->
      subrel S ((root_step (+) lift I)* ;; lift S) ->
      subrel S ired.
\end{verbatim}}
\end{itemize}  
\end{formalisation}
\noindent
\begin{rem}\label{rem:ired:induction}
  For this definition, Coq generates the following good induction principle:
  \begin{itemize}[leftmargin=4ex]
  \item []
  {\small
  \begin{verbatim}
  ired_ind : forall P : term -> term -> Prop,
             ( forall S I : relation (term F X),
               subrel I ired ->
               subrel I P ->                
               subrel S ((root_step (+) lift I)* ;; lift S) ->
               subrel S P) ->
             subrel ired P
  \end{verbatim}}
  \end{itemize}
  Thus in order to prove ${\ired} \subseteq P$, it suffices to show
  \begin{align*}
    I \subset {\ired} 
    \;\implies\; I \subseteq P 
    \;\implies\; S \subseteq {(\rstep \cup \mathrel{\down{I}})^*\relcomp \down{S}} 
    \;\implies\; S \subseteq P
  \end{align*}
  for every relation $I$ and $S$.
  Below, we will see in several examples, 
  that this is a useful induction principle that is easy to work with.
\end{rem}

\subsection{Formalisation of Omega Rewriting}
In order to formalise the Compression Lemma, we need rewrite sequences of length $\le \omega$.
To formalise $\to^{\le \omega}$, we use
\begin{align}
  \to^{\le \omega} \;\;\defd\;\; \gfp{R}{({\to^*} \relcomp {\down{R}})} \label{eq:ored:good}
\end{align}
That is, a rewrite sequence of length $\le \omega$ 
is a finite rewrite sequence
followed by rewrite sequences of length $\le \omega$ on subterms (below the root).
Using dovetailing, this gives rise to the usual concept of rewrite sequences indexed by an ordinal $\le \omega$.

As above for $\ired$, we have:
\begin{align}
  {\to^{\le \omega}} \;\;\defd\;\; 
 \begin{aligned} 
      &\text{the smallest $T$ such that: } \\[-1ex]
      &\text{for all relations $S$, } S \subseteq ({\to^*} \relcomp {\down{S}}) \implies S \subseteq T
  \end{aligned} 
  \label{eq:ored}
\end{align}
This definition can be formalised in Coq as follows.
\begin{formalisation}[Rewrite sequences of length at most $\omega$]\mbox{}\label{form:ored}
\begin{itemize}[leftmargin=4ex]
\item []
{\small
\begin{verbatim}
Inductive ored : relation term :=
  | Ored : 
      forall S : relation term,
      subrel S (mred ;; lift S) ->
      subrel S ored.
\end{verbatim}}
\end{itemize}  
Here \verb=mred= are finite rewrite sequences $\to^*$.
\end{formalisation}

To keep the proof of compression in Coq as simple as possible,  
we have chosen to define the finite rewrite relation $\to^*$ 
in a non-standard way.
We introduce \texttt{mred} as the smallest relation $S$ that fulfils the following conditions:
\begin{enumerate}
  \item \texttt{mred\_bis}: if $s \mathrel{S} t$, $s \xbis s'$ and $t \xbis t'$, then $s' \mathrel{S} t'$,
  \item \texttt{mred\_refl}: $\mathrel{S}$ is reflexive,
  \item \texttt{mred\_trans}: $\mathrel{S}$ is transitive,
  \item \texttt{mred\_root}: ${\rstep} \subseteq {\mathrel{S}}$, and
  \item \texttt{mred\_fun}: if $u_1 \mathrel{S} v_1$, \ldots, $u_n \mathrel{S} v_n$, then $f(u_1,\ldots,u_n) \mathrel{S} f(v_1,\ldots,v_n)$.
\end{enumerate}
Let us check that \texttt{mred} indeed is the finite rewrite relation $\to^*$:
\begin{enumerate}[label=\({\alph*}]
  \item 
    For ${\texttt{mred}} \subseteq {\to^*}$ note that the finite rewrite relation $\to^*$ fulfils all these criteria.
  \item 
    For ${\to^*} \subseteq {\texttt{mred}}$ we argue as follows.
    We have ${\rstep} \subseteq {\texttt{mred}}$ by \texttt{mred\_root}.
    By \texttt{mred\_fun} together with \texttt{mred\_refl} 
    it follows that \texttt{mred} is closed under contexts. 
    Thus ${\to} \subseteq {\texttt{mred}}$.
    By \texttt{mred\_refl} and \texttt{mred\_trans} we obtain that ${\to^*} \subseteq {\texttt{mred}}$.
\end{enumerate}
Hence ${\texttt{mred}} = {\to^*}$.

\begin{formalisation}[Finite rewriting relation on infinite terms]\mbox{}
\begin{itemize}[leftmargin=4ex]
\item []
{\small
\begin{verbatim}
Inductive mred : relation term :=
  | mred_bis :
      forall s s' t t', s [~] s' -> t [~] t' -> mred s t -> mred s' t'
  | mred_refl :
      forall s, mred s s
  | mred_trans :
      forall s t u, mred s t -> mred t u -> mred s u
  | mred_root : 
      forall (s t : term F X) , root_step s t -> mred s t
  | mred_fun : 
      forall (f : F) (u v : vector (term F X) (arity f)), 
      (forall i : Fin (arity f), mred (u i) (v i)) -> 
      mred (Fun f u) (Fun f v).
\end{verbatim}}
\end{itemize}  
\end{formalisation}

\begin{nota}
  In the remainder of this section, we will write
  \begin{center}
    $\cmred$ for \texttt{mred}\;,\quad
    $\cired$ for \texttt{ired}\;,\quad
    $\cored$ for \texttt{ored}\;, and\quad
    $\crstep$ for \texttt{root\_step}.
  \end{center} 
\end{nota}

We have formalised Lemma~\ref{lem:ored:ok} to show that
Formalisation~\ref{form:ored} of \texttt{ored} indeed corresponds to Equation~\ref{eq:ored:good}.
We first prove two auxiliary facts.

\begin{lem}[\texttt{ored\_subrel}]\label{lem:ored:subrel}
  We have ${\cored} \subseteq {\cmred} \relcomp {\down{\cored}}$.
\end{lem}

\begin{proof}
  Let $s \cored t$.
  By definition of $\cored$ there exists a relation $S \subseteq {\cored}$ 
  such that $S \subseteq {\cmred} \relcomp {\down{S}}$ and $s \mathrel{S} t$.
  As $s \mathrel{{\cmred} \relcomp {\down{S}}} t$ and $S \subseteq {\cored}$, we have $s \cmred \relcomp \mathrel{\down{\cored}} t$.
\end{proof}

\begin{lem}[\texttt{subrel\_lift\_ored\_ored}]\label{lem:ored:subrel:lift:ored:ored}
  We have ${\down{\cored}} \subseteq {\cored}$.
\end{lem}

\begin{proof}
  Let $s \mathrel{\down{\cored}} t$. 
  Define $S = {{\cored} \cup {\pair{s}{t}}}$.
  Then we have $s \mathrel{S} t$.
  To prove $s \cored t$, it suffices (by definition of $\cored$) to show that $S \subseteq {{\cmred} \relcomp {\down{S}}}$.
  By Lemma~\ref{lem:ored:subrel},
  we have ${\cored} \subseteq {\cmred} \relcomp {\down{\cored}}$,
  and hence ${\cored} \subseteq {\cmred} \relcomp {\down{S}}$.
  Moreover,
  we have $\pair{s}{t} \in {\cmred} \relcomp {\down{S}}$
  as a consequence of reflexively of ${\cmred}$,
  and since $s \mathrel{\down{\cored}} t$ and ${\cored} \subseteq S$ imply $s \mathrel{\down{S}} t$.
  This concludes the proof.
\end{proof}

\begin{lem}[\texttt{ored\_ok}]\label{lem:ored:ok}
  We have ${\cored} = {\cmred} \relcomp {\down{\cored}}$.
\end{lem}

\begin{proof}
  By Lemma~\ref{lem:ored:subrel} it suffices to show that ${\cmred} \relcomp {\down{\cored}} \subseteq {\cored}$.
  Let $s \cmred t \mathrel{\down{\cored}} u$.
  Then, by Lemma~\ref{lem:ored:subrel:lift:ored:ored}, we have $s \cmred t \cored u$.
  By definition of $\cored$ there exists a relation $S \subseteq {\cored}$ 
  such that $S \subseteq {\cmred} \relcomp {\down{S}}$ and $t \mathrel{S} u$.
  Thus $t \cmred t' \mathrel{\down{S}} u$ for some $t'$.
  Define $S' = S \cup {\pair{s}{u}}$.
  Then $s \mathrel{S'} u$ and $S' \subseteq {\cmred} \relcomp {\down{S'}}$ 
  since $s \cmred t' \mathrel{\down{S'}} u$ 
  by \texttt{mred\_trans} and $S \subseteq S'$.
  Hence $s \cored u$ by definition of $\cored$ using $S'$.
\end{proof}

\subsection{Formalisation of the Compression Lemma}
Using the above definitions, we will now prove the Compression Lemma.
The proof realises a transformation of $\ired$ proof trees into $\ored$ proof trees.

\begin{lem}[Compression Lemma]
  Let $\atrs$ be a left-linear term rewriting system with finite left-hand sides.
  Whenever there is an infinite reduction from $s$ to $t$ ($s \ired t$) 
  then there exists a reduction of length at most $\omega$ from $s$ to~$t$ ($s \to^{\le \omega} t$).
\end{lem}

We have formalised the Compression Lemma as follows:
\begin{lem}\label{lem:compression}
  For left-linear \texttt{trs}'s we have ${\cired} \subseteq {\cored}$.
\end{lem}
\noindent
The condition \emph{finite left-hand sides} is
part of the formalisation of \texttt{trs}'s, see above. 
In the remainder of this section we tacitly assume that the underlying \texttt{trs} is \emph{left-linear}.
This is a necessary condition for all lemmas~\ref{lem:ored:match}--\ref{lem:ored:liftR}.

We present the proof of Lemma~\ref{lem:compression} as close as possible to to our formalisation in Coq.
We begin with a few auxiliary lemmas.
\begin{lem}[\texttt{ored\_match}]\label{lem:ored:match}
  If $s \cored \ell\sigma$ for a finite, linear term $\ell$, 
  then $s \cmred \ell\tau$ for some substitution $\tau$ with
  $\forall x.\;\tau(x) \cored \sigma(x)$. 
\end{lem}

\begin{proof}
  The proof proceeds by induction on $\ell$.
  If $s \cored \ell\sigma$, then there exists a term $t$ 
  such that $s \cmred t \mathrel{\down{\cored}} \ell\sigma$.
  Assume that $\ell$ is a variable, say $\ell = x$.
  Then define $\tau(x) = t$ and  $\tau(y) = \sigma(y)$ for every $y \ne x$.
  We have $\tau(x) = t \mathrel{\down{\cored}} \ell\sigma = \sigma(x)$.
  Then $s \cmred \ell\tau$,
  and we have $\forall x.\;\tau(x) \cored \sigma(x)$
  since ${\down{\cored}} \subseteq {\cored}$ by Lemma~\ref{lem:ored:subrel:lift:ored:ored} and $\cored$ is reflexive. 
  
  If $\ell$ is not a variable, let $\ell = f(\ell_1,\ldots,\ell_n)$.
  Then $t \mathrel{\down{\cored}} \ell\sigma$ implies that
  $t = f(t_1,\ldots,t_n)$ for some terms $t_1,\ldots,t_n$,
  and we have $t_i \cored \ell_i\sigma$ for every $i \in \{1,\ldots,n\}$.
  Then by induction hypothesis, we have
  $t_i \cmred \ell_i\tau_i$ and $\forall x.\;\tau_i(x) \cored \sigma(x)$
  for every $i \in \{1,\ldots,n\}$.
  From linearity of $\ell$ it follows that $\ell_1,\ldots,\ell_n$ do not share any variables.
  Consequently, we can define the substitution $\tau$ as follows:
  $\tau(x) = \tau_i(x)$ if $x \in \vars{\ell_i}$ 
  and $\tau(x) = \sigma(x)$ if $x \not\in \vars{\ell}$.
  Then $\ell_i\tau = \ell_i\tau_i$ for every $i \in \{1,\ldots,n\}$
  and $\ell\tau = f(\ell_1\tau_1,\ldots,\ell_n\tau_n)$.
  Moreover we have $\forall x.\;\tau(x) \cored \sigma(x)$ by definition of $\tau$.
  It remains to show $s \cmred \ell\tau$.
  Using \texttt{mred\_fun} we get 
  $t = f(t_1,\ldots,t_n) \cmred f(\ell_1\tau_1,\ldots,\ell_n\tau_n) = \ell\tau$.
  By \texttt{mred\_trans} we obtain $s \cmred \ell\tau$.
\end{proof}

\begin{lem}[\texttt{subsitution\_ored}]\label{lem:subsitution:ored}
  Let $t$ be a (finite or infinite) term and $\sigma, \tau$ substitutions such that
  $\sigma(x) \cored \tau(x)$ for every $x \in \vars{t}$.
  Then $t\sigma \cored t\tau$.
\end{lem}

\begin{proof}
  By definition of $\cored$
  it suffices to find a relation $S$ such that
  $t\sigma \mathrel{S} t\tau$
  and $S \subseteq {\cmred \relcomp \mathrel{\down{S}}}$.
  We define 
  $S = {\cored \cup \mathrel{\{\; \pair{u\sigma}{u\tau} \mid \text{$u$ a term} \;\}}}$.
  We show  $S \subseteq {\cmred \relcomp \mathrel{\down{S}}}$.
  We have 
  \begin{align}
    {\cored} 
    \;\;\subseteq\;\; {\cmred \relcomp \mathrel{\down{\cored}}}
    \;\;\subseteq\;\; {\cmred \relcomp \mathrel{\down{S}}}
    \label{eq:subsitution:ored:ored}
  \end{align}
  by Lemma~\ref{lem:ored:ok}.
  
  So consider $u\sigma \mathrel{S} u\tau$ for some term $u$.
  If $u$ is a variable, then $u\sigma = \sigma(u) \cored \tau(u) = u\tau$.
  Then $u\sigma \cmred \relcomp \mathrel{\down{S}} u\tau$ follows from \eqref{eq:subsitution:ored:ored}.
  Thus, let $u = f(u_1,\ldots,u_n)$
  for some symbol $f$ and terms $u_1,\ldots,u_n$.
  Then we have $u\sigma = f(u_1\sigma,\ldots,u_n\sigma) \mathrel{\down{S}}  f(u_1\tau,\ldots,u_n\tau) = u\tau$,
  and clearly $\down{S} \subseteq {\cmred \relcomp \mathrel{\down{S}}}$.
  Hence $S \subseteq {\cmred \relcomp \mathrel{\down{S}}}$.
\end{proof}

\begin{lem}[\texttt{ored\_rstep}]\label{lem:ored:rstep}
  We have ${\cored \relcomp \crstep} \subseteq {\cored}$.
\end{lem}

\begin{proof}
  Let $\ell \to r \in \atrs$ be a rule, $\sigma$ a substitution and consider $s \cored \ell\sigma \crstep r\sigma$.
  Then by Lemma~\ref{lem:ored:match} we have
  $s \cmred \ell\tau$ for some substitution $\tau$ with $\forall x.\;\tau(x) \cored \sigma(x)$.
  We have $r\tau \cored r\sigma$ by Lemma~\ref{lem:subsitution:ored},
  and thus $r\tau \cmred t \mathrel{\down{\cored}} r\sigma$ for some $t$ by Lemma~\ref{lem:ored:ok}.
  Then also $s \cmred \ell\tau \crstep r\tau \cmred t \mathrel{\down{\cored}} r\sigma$
  and $s \cmred t \mathrel{\down{\cored}} r\sigma$ by \texttt{mred\_root} and \texttt{mred\_trans}.
  Hence we have $s \cored r\sigma$ by Lemma~\ref{lem:ored:ok}.
\end{proof}

\begin{lem}[\texttt{ored\_mred}]\label{lem:ored:mred}
  We have ${\cored \relcomp \cmred} \subseteq {\cored}$.
\end{lem}

\begin{proof}
  The proof proceeds by induction on the definition of $\cmred$ in $s \cored t \cmred u$.
  The induction step is trivial for \texttt{mred\_bis}, \texttt{mred\_refl} and \texttt{mred\_trans}.
  For \texttt{mred\_root} the claim follows from Lemma~\ref{lem:ored:rstep}.
  Here we only consider the case of \texttt{mred\_fun}.
  Then $t = f(t_1,\ldots,t_n)$ and $u = f(u_1,\ldots,u_n)$
  for some symbol~$f$ and $t_i \cmred u_i$ for every $i \in \{1,\ldots,n\}$.
  From $s \cored t$ it follows that
  $s \cmred s' \mathrel{\down{\cored}} t = f(t_1,\ldots,t_n)$
  for some $s'$.
  Hence $s' = f(s_1,\ldots,s_n)$ with $s_i \cored t_i$ for every $i \in \{1,\ldots,n\}$.
  By the induction hypothesis, we obtain 
  $s_i \cored u_i$ for every $i \in \{1,\ldots,n\}$,
  and, consequently, $s \cmred s' \mathrel{\down{\cored}} u$.
  Finally, $s \cored u$ by Lemma~\ref{lem:ored:ok}.
\end{proof}

\begin{lem}[\texttt{ored\_ored}]\label{lem:ored:ored}
  We have ${\cored \relcomp \cored} \subseteq {\cored}$ and 
  ${\cored \relcomp \mathrel{\down{\cored}}} \subseteq {\cored}$.
\end{lem}

\begin{proof}
  Define $S = (\cored \relcomp \cored) \cup (\cored \relcomp \mathrel{\down{\cored}})$.
  We show $S \subseteq {\cored}$.
  By definition of $\cored$ it suffices that $S \subseteq {\cmred \relcomp \mathrel{\down{S}}}$.
  We have:
  \begin{align}
    {\cored \relcomp \cored} 
    \;\subseteq\;
    {\cored \relcomp \mathrel{(\cmred \relcomp \mathrel{\down{\cored}})}}
    \;\subseteq\;
    {\cored \relcomp \mathrel{\down{\cored}}}
  \end{align}
  by Lemma~\ref{lem:ored:ok} and Lemma~\ref{lem:ored:mred},
  and
  \begin{align*}
    {\cored \relcomp \mathrel{\down{\cored}}}
    \;\subseteq\;
    {\mathrel{(\cmred \relcomp \mathrel{\down{\cored}})} \relcomp \mathrel{\down{\cored}}}
    \;\subseteq\;
    {\cmred \relcomp \mathrel{(\mathrel{\down{\cored}} \relcomp \mathrel{\down{\cored}})}}
    \;\subseteq\;
    {\cmred \relcomp \mathrel{\down{(\cored \relcomp \cored)}}}
    \;\subseteq\;
    {\cmred \relcomp \mathrel{\down{S}}}
  \end{align*}
  by Lemma~\ref{lem:ored:ok}, associativity, definition of \,$\down{\;\cdot\;}$.
  This proves $S \subseteq {\cmred \relcomp \mathrel{\down{S}}}$.
\end{proof}

\begin{lem}[\texttt{ored\_mix}]\label{lem:ored:mix}
  We have ${(\,{\crstep} \cup {\down{\cored}}\,)^*} \subseteq {\cored}$.
\end{lem}

\proof
  Define $S = (\,{\crstep} \cup {\down{\cored}}\,)$. 
  We prove $S^* \subseteq {\cored}$ by induction on (the definition of) the reflexive transitive closure:
  \begin{enumerate}
    \item \texttt{refl\_trans\_step}:
      For $s \mathrel{S} t$ we have either $s \crstep t$ or $s \mathrel{\down{\cored}} t$.
      Then $s \cored t$ as a consequence of either Lemma~\ref{lem:ored:rstep} or Lemma~\ref{lem:ored:ok}, respectively.
    \item \texttt{refl\_trans\_refl}:
      For $s \xbis t$ we $s \cmred t$ by \texttt{mred\_bis} and \texttt{mred\_refl}.
      Hence $s \cored t$ since ${\cmred} \relcomp {\down{\cored}} \subseteq \cored$ by Lemma~\ref{lem:ored:ok}
      and $\down{\cored}$ is reflexive by \texttt{lift\_id}.
    \item \texttt{refl\_trans\_trans}:
      For $s \mathrel{S^*} u \mathrel{S^*} t$
      we obtain $s \cored u$ and $u \cored t$ by induction hypothesis.
      Then $s \cored t$ follows from Lemma~\ref{lem:ored:ored}.
      \qedhere
  \end{enumerate}
\medskip

\begin{lem}[\texttt{ored\_liftR}]\label{lem:ored:liftR}
  We have $R \subseteq {\cored \relcomp \mathrel{\down{R}}}$ implies $R \subseteq {\cored}$.
\end{lem}

\begin{proof}
  Define $S = (\cored \relcomp \mathrel{R}) \cup (\cored \relcomp \mathrel{\down{R}})$.
  We show $S \subseteq (\cmred \relcomp \mathrel{\down{S}})$:
  \begin{align*}
    {\cored \relcomp \mathrel{R}} 
    \;\;\subseteq\;\; \cored \relcomp \mathrel{({\cored \relcomp \mathrel{\down{R}}})}
    \;\;\subseteq\;\; \cored \relcomp \mathrel{\down{R}}
  \end{align*}
  by $R \subseteq {\cored \relcomp \mathrel{\down{R}}}$ and Lemma~\ref{lem:ored:ored}, and
  \begin{align*}
    \cored \relcomp \mathrel{\down{R}}
    \;\;\subseteq\;\; \mathrel{(\cmred \relcomp \mathrel{\down{\cored}})} \relcomp \mathrel{\down{R}}
    \;\;\subseteq\;\; \cmred \relcomp \mathrel{\down{(\cored \relcomp \mathrel{R})}}
    \;\;\subseteq\;\; \cmred \relcomp \mathrel{\down{S}}
  \end{align*}
  by Lemma~\ref{lem:ored:ok}.
  Hence $S \subseteq (\cmred \relcomp \mathrel{\down{S}})$
  and consequently $S \subseteq {\cored}$.
  Note that $R \subseteq {\cored \relcomp \mathrel{\down{R}}} \subseteq S$. 
  Thus $R \subseteq S \subseteq {\cored}$.
\end{proof}

We are ready to prove compression.
\begin{proof}[Proof of Lemma~\ref{lem:compression}]
  The proof proceeds by induction on $\cired$.
  By the induction principle discussed in Remark~\ref{rem:ired:induction}
  we have to show 
  \begin{align*}
    I \subset {\cired} 
    \;\implies\; I \subseteq {\cored} 
    \;\implies\; S \subseteq {(\crstep \cup \mathrel{\down{I}})^*\relcomp \down{S}} 
    \;\implies\; S \subseteq {\cored}
  \end{align*}
  for every relation $I$ and $S$.
  So let $I \subseteq ({\cired} \cap {\cored})$ and $S \subseteq {(\crstep \cup \mathrel{\down{I}})^*\relcomp \down{S}}$.
  Since $I \subseteq {\cored}$ we get 
  $(\crstep \cup \mathrel{\down{I}})^* \subseteq (\crstep \cup \mathrel{\down{\cored}})^* \subseteq {\cored}$ 
  by Lemma~\ref{lem:ored:mix}.
  Consequently, $S \subseteq {\cored \relcomp \mathrel{\down{S}}}$
  and, by Lemma~\ref{lem:ored:liftR}, we obtain $S \subseteq {\cored}$.
  Hence ${\cired} \subseteq {\cored}$.
\end{proof}

\noindent

To the best of our knowledge this is the first formal proof of this well-known lemma.
The formalization is available at~\cite{compression}.

\section{Conclusion}\label{sec:conclusion}
We have proposed a coinductive framework which gives rise to several
natural variants of infinitary rewriting in a uniform way:
  \begin{enumerate}[label=\({\alph*}]
    \item infinitary equational reasoning
        ${\ieq} \;\;\defd\;\; \gfp{y}{(\leftarrow_{\varepsilon} \cup \rstep \cup \mathrel{\down{y}})^*}$,
  \smallskip
    \item bi-infinite rewriting
        ${\ibi} \;\;\defd\;\; \gfp{y}{(\rstep \cup \mathrel{\down{y}})^*}$, and
  \smallskip
    \item infinitary rewriting
        ${\ired} \;\;\defd\;\; \lfp{x}{\gfp{y}{(\rstep \cup \mathrel{\down{x}})^*\relcomp \down{y}}}$\,.
  \end{enumerate}
We believe that (a) and (b) are novel.
As a consequence of the coinduction over the term structure,
these notions have the strong convergence built-in,
and thus can profit from the well-developed techniques (such as tracing)
in infinitary rewriting.

We have given a mixed inductive/coinductive definition of infinitary rewriting 
and established a bridge between infinitary rewriting and coalgebra.
Both fields are concerned with infinite objects and we would like %
to understand their relation better. 
In contrast to previous coinductive treatments, 
the framework presented here captures rewrite sequences of arbitrary ordinal length,
and paves the way for formalizing infinitary rewriting in theorem provers
(as illustrated by our proof of the Compression Lemma in Coq).

Concerning proof trees/terms for infinite reductions, let us mention that 
an alternative approach has been developed in parallel by Lombardi, R\'{\i}os and de~Vrijer~\cite{lomb:rios:vrij:2014}.
While we focus on proof terms for the reduction relation and abstract from the order of steps in parallel subterms, 
they use proof terms for modeling the fine-structure of 
the infinite reductions themselves.
Another difference is that our framework allows for non-left-linear systems.
We believe that both approaches are complementary.
Theorems for which the fine-structure of rewrite sequences is crucial,
must be handled using~\cite{lomb:rios:vrij:2014}.
(But note that we can capture standard reductions by a restriction on proof trees 
and prove standardization using proof tree transformations, see~\cite{endr:polo:2012b}). 
If the fine-structure is not important, as for instance for proving confluence,
then our system is more convenient to work with due to simpler proof terms.

Our work lays the foundation for several directions of future research:
\begin{enumerate}
  \item 
    The coinductive treatment of infinitary $\lambda$-calculus~\cite{endr:polo:2012b}
    has led to elegant, significantly simpler proofs~\cite{czaj:2014,czaj:2015}
    of some central properties of the infinitary $\lambda$-calculus.
    The coinductive framework that we propose
    enables similar developments for infinitary term rewriting
    with reductions of arbitrary ordinal length.

  \item 
    The concepts of bi-infinite rewriting is novel, and the theory of infinitary equational reasoning is still underdeveloped.
    It would be interesting to study these concepts.
    Is there an equivalent of ordinal-indexed rewrite sequences for bi-infinite rewriting
    (maybe using Conway's surreal numbers~\cite{conw:2000})?
    Is it possible to establish some sort of Compression Lemma for bi-infinite rewriting?
    
    Moreover, it would be fruitful to compare the 
    Church--Rosser properties 
    \begin{align*}
      {\ieq}  \;\subseteq\; {\ired \relcomp \iredi} 
      &&\text{ and }&& 
      {(\iredi \relcomp \ired)^*} \;\subseteq\; {\ired \relcomp \iredi} \;\,.
    \end{align*}

  \item 
    The formalization of the proof of the Compression Lemma in Coq
    is just the first step towards the formalization of all major theorems in infinitary rewriting.

    It would also be interesting to formalise infinitary equational reasoning $\ieq$ and bi-infinite rewriting $\ibi$ in Coq. 
    We expect that it is straightforward to adapt our Coq formalization 
    of infinitary rewriting $\ired$ to equational reasoning $\ieq$ and bi-infinite rewriting $\ibi$.
    The latter two concepts have significantly simpler definitions in the fixed point calculus.
    
  \item It is interesting to investigate how the coinductive framework
    should be extended to incorporate the infinitary analysis of
    meaningless
    terms. \cite{bahr:2010,bahr:2010b,bahr:2012,endr:hend:klop:2012}
    This would be the natural stepping-stone to the formalization of
    confluence theorems in infinitary rewriting extended with $\bot$-reduction.
    
  \item 
    We believe that the coinductive definitions will ease the development of new techniques for automated reasoning 
    about infinitary rewriting.
    For example, methods
    for proving (local) productivity~\cite{endr:grab:hend:2009,endr:hend:2011,zant:raff:2010},
    for (local) infinitary normalization~\cite{zant:2008,endr:grab:hend:klop:vrij:2009,endr:vrij:wald:10},
    for (local) unique normal forms~\cite{endr:hend:grab:klop:oost:2014},
    and for analysis of infinitary reachability and infinitary confluence.
    Due to the coinductive definitions, the implementation and formalization of these techniques 
    could make use of circular coinduction~\cite{gogu:lin:rosu:2000,endr:hend:bodi:2013}.  
\end{enumerate}

\subsection*{Acknowledgments}
We thank Patrick Bahr, Jeroen Ketema, and Vincent van Oostrom for fruitful discussions 
and comments on earlier versions of this paper.
We are thankful to the reviewers for pointing out mistakes and suggesting several improvements.
%


\bibliographystyle{plain}
\bibliography{main}

\begin{thebibliography}{10}

\bibitem{baad:nipk:1998}
F.~Baader and T.~Nipkow.
\newblock {\em {Term Rewriting and All That}}.
\newblock Cambridge Univ. Press, 1998.

\bibitem{bahr:2010b}
P.~Bahr.
\newblock {Abstract Models of Transfinite Reductions}.
\newblock In {\em Proc.\ Conf.\ on Rewriting Techniques and Applications (RTA
  2010)}, volume~6 of {\em Leibniz International Proceedings in Informatics},
  pages 49--66. Schloss Dagstuhl, 2010.

\bibitem{bahr:2010}
P.~Bahr.
\newblock {Partial Order Infinitary Term Rewriting and B{\"o}hm Trees}.
\newblock In {\em Proc.\ Conf.\ on Rewriting Techniques and Applications (RTA
  2010)}, volume~6 of {\em Leibniz International Proceedings in Informatics},
  pages 67--84. Schloss Dagstuhl, 2010.

\bibitem{bahr:2012}
P.~Bahr.
\newblock {Infinitary Term Graph Rewriting is Simple, Sound and Complete}.
\newblock In {\em Proc.\ Conf.\ on Rewriting Techniques and Applications (RTA
  2012)}, volume~15 of {\em Leibniz International Proceedings in Informatics},
  pages 69--84. Schloss Dagstuhl, 2012.

\bibitem{bare:1977}
H.P. Barendregt.
\newblock {The Type Free Lambda Calculus}.
\newblock In {\em Handbook of Mathematical Logic}, pages 1091--1132.
  Nort-Holland Publishing Company, Amsterdam, 1977.

\bibitem{bare:klop:2009}
H.P. Barendregt and J.W. Klop.
\newblock {Applications of Infinitary Lambda Calculus}.
\newblock {\em Information and Computation}, 207(5):559--582, 2009.

\bibitem{conw:2000}
J.H. Conway.
\newblock {\em {On Numbers and Games}}.
\newblock Ak Peters Series. Taylor \& Francis, 2000.

\bibitem{coqu:1996}
C.~Coquand and Th. Coquand.
\newblock {On the Definition of Reduction for Infinite Terms}.
\newblock {\em Comptes Rendus de l'Acad\'emie des Sciences. S\'erie I},
  323(5):553--558, 1996.

\bibitem{coqu:1994}
Th. Coquand.
\newblock Infinite objects in type theory.
\newblock In Henk Barendregt and Tobias Nipkow, editors, {\em Types for Proofs
  and Programs, International Workshop TYPES'93, Nijmegen, The Netherlands, May
  24--28, 1993, Selected Papers}, volume 806 of {\em LNCS}, pages 62--78.
  Springer, 1994.

\bibitem{czaj:2014}
\L. Czajka.
\newblock {A Coinductive Confluence Proof for Infinitary Lambda-Calculus}.
\newblock In {\em Rewriting and Typed Lambda Calculi (RTA-TLCA 2014)}, volume
  8560 of {\em Lecture Notes in Computer Science}, pages 164--178. Springer,
  2014.

\bibitem{czaj:2015}
{\L}.~{Czajka}.
\newblock {Coinductive Techniques in Infinitary Lambda-Calculus}.
\newblock {\em ArXiv e-prints}, 2015.

\bibitem{ders:kapl:plai:1991}
N.~Dershowitz, S.~Kaplan, and D.A. Plaisted.
\newblock {Rewrite, Rewrite, Rewrite, Rewrite, Rewrite,\dots}.
\newblock {\em Theoretical Computer Science}, 83(1):71--96, 1991.

\bibitem{endr:vrij:wald:10}
J.~Endrullis, R.~C. de~Vrijer, and J.~Waldmann.
\newblock {Local Termination: Theory and Practice}.
\newblock {\em Logical Methods in Computer Science}, 6(3), 2010.

\bibitem{endr:grab:hend:2009}
J.~Endrullis, C.~Grabmayer, and D.~Hendriks.
\newblock {Complexity of Fractran and Productivity}.
\newblock In {\em Proc.\ Conf.\ on Automated Deduction (CADE~22)}, volume 5663
  of {\em LNCS}, pages 371--387, 2009.

\bibitem{endr:grab:hend:klop:vrij:2009}
J.~Endrullis, C.~Grabmayer, D.~Hendriks, J.W. Klop, and R.C de~Vrijer.
\newblock {Proving Infinitary Normalization}.
\newblock In {\em Postproc.\ Int.\ Workshop on Types for Proofs and Programs
  (TYPES 2008)}, volume 5497 of {\em LNCS}, pages 64--82. Springer, 2009.

\bibitem{endr:hans:hend:polo:silv:2015}
J.~Endrullis, H.~H. Hansen, D.~Hendriks, A.~Polonsky, and A.~Silva.
\newblock {A Coinductive Framework for Infinitary Rewriting and Equational
  Reasoning}.
\newblock In {\em Proc.\ Conf.\ on Rewriting Techniques and Applications (RTA
  2015)}, Leibniz International Proceedings in Informatics. Schloss Dagstuhl,
  2015.

\bibitem{compression}
J.~Endrullis, H.~H. Hansen, D.~Hendriks, A.~Polonsky, and A.~Silva.
\newblock {A Formalization of the Compression Lemma}, 2016.
\newblock Available at \url{http://joerg.endrullis.de/coq/compression}.

\bibitem{endr:hend:2011}
J.~Endrullis and D.~Hendriks.
\newblock {Lazy Productivity via Termination}.
\newblock {\em Theoretical Computer Science}, 412(28):3203--3225, 2011.

\bibitem{endr:hend:bodi:2013}
J.~Endrullis, D.~Hendriks, and M.~Bodin.
\newblock {Circular Coinduction in Coq Using Bisimulation-Up-To Techniques}.
\newblock In {\em Proc. Conf. on Interactive Theorem Proving (ITP)}, volume
  7998 of {\em LNCS}, pages 354--369. Springer, 2013.

\bibitem{endr:hend:grab:klop:oost:2014}
J.~Endrullis, D.~Hendriks, C.~Grabmayer, J.W. Klop, and V.~van Oostrom.
\newblock Infinitary term rewriting for weakly orthogonal systems: Properties
  and counterexamples.
\newblock {\em Logical Methods in Computer Science}, 10(2:7):1--33, 2014.

\bibitem{endr:hend:klop:2012}
J.~Endrullis, D.~Hendriks, and J.W. Klop.
\newblock {Highlights in Infinitary Rewriting and Lambda Calculus}.
\newblock {\em Theoretical Computer Science}, 464:48--71, 2012.

\bibitem{endr:polo:2012b}
J.~Endrullis and A.~Polonsky.
\newblock {Infinitary Rewriting Coinductively}.
\newblock In {\em Proc.\ Types for Proofs and Programs (TYPES 2012)}, volume~19
  of {\em Leibniz International Proceedings in Informatics}, pages 16--27.
  Schloss Dagstuhl, 2013.

\bibitem{gogu:lin:rosu:2000}
J.~Goguen, K.~Lin, and G.~Ro\c{s}u.
\newblock {Circular Coinductive Rewriting}.
\newblock In {\em Proc. of Automated Software Engineering}, pages 123--131.
  IEEE, 2000.

\bibitem{jaco:rutt:2011}
B.~Jacobs and J.J.M.M. Rutten.
\newblock {An Introduction to (Co)Algebras and (Co)Induction}.
\newblock In {\em Advanced Topics in Bisimulation and Coinduction}, pages
  38--99. Cambridge University Press, 2011.

\bibitem{joac:2004}
F.~Joachimski.
\newblock {Confluence of the Coinductive Lambda Calculus}.
\newblock {\em Theoretical Computer Science}, 311(1-3):105--119, 2004.

\bibitem{kahr:2013}
S.~Kahrs.
\newblock {Infinitary Rewriting: Closure Operators, Equivalences and Models}.
\newblock {\em Acta Informatica}, 50(2):123--156, 2013.

\bibitem{kenn:vrie:2003}
J.R. Kennaway and F.-J. de~Vries.
\newblock {\em {Infinitary Rewriting}}, chapter~12.
\newblock Cambridge University Press, 2003.
\newblock in~\cite{tere:2003}.

\bibitem{kenn:klop:slee:vrie:1995a}
J.R. Kennaway, J.W. Klop, M.R. Sleep, and F.-J. de~Vries.
\newblock {Transfinite Reductions in Orthogonal Term Rewriting Systems}.
\newblock {\em Information and Computation}, 119(1):18--38, 1995.

\bibitem{kete:simo:2013}
J.~Ketema and J.G. Simonsen.
\newblock {Computing with Infinite Terms and Infinite Reductions}.
\newblock Unpublished manuscript.

\bibitem{klop:1992}
J.W. Klop.
\newblock {Term Rewriting Systems}.
\newblock In {\em Handbook of Logic in Computer Science}, volume~II, pages
  1--116. Oxford University Press, 1992.

\bibitem{klop:vrij:2005}
J.W. Klop and R.C de~Vrijer.
\newblock {Infinitary Normalization}.
\newblock In {\em We Will Show Them: Essays in Honour of Dov Gabbay (2)}, pages
  169--192. {College Publications}, 2005.

\bibitem{lomb:rios:vrij:2014}
C.~Lombardi, A.~R{\'{\i}}os, and R.C de~Vrijer.
\newblock {Proof Terms for Infinitary Rewriting}.
\newblock In {\em Rewriting and Typed Lambda Calculi (RTA-TLCA 2014)}, volume
  8560 of {\em Lecture Notes in Computer Science}, pages 303--318. Springer,
  2014.

\bibitem{Milius:CIA}
S.~Milius.
\newblock Completely iterative algebras and completely iterative monads.
\newblock {\em ic}, 196:1--41, 2005.

\bibitem{simo:2004}
J.G. Simonsen.
\newblock {On Confluence and Residuals in Cauchy Convergent Transfinite
  Rewriting}.
\newblock {\em Information Processing Letters}, 91(3):141--146, 2004.

\bibitem{tere:2003}
Terese.
\newblock {\em {Term Rewriting Systems}}, volume~55 of {\em Cambridge Tracts in
  Theoretical Computer Science}.
\newblock Cambridge University Press, 2003.

\bibitem{verm:2010}
M.~Vermaat.
\newblock {Infinitary Rewriting in Coq}.
\newblock Available at url \url{http://martijn.vermaat.name/master-project/}.

\bibitem{zant:2008}
H.~Zantema.
\newblock {Normalization of Infinite Terms}.
\newblock In {\em Proc.\ Conf.\ on Rewriting Techniques and Applications (RTA
  2008)}, number 5117 in LNCS, pages 441--455, 2008.

\bibitem{zant:raff:2010}
H.~Zantema and M.~Raffelsieper.
\newblock {Proving Productivity in Infinite Data Structures}.
\newblock In {\em Proc.\ Conf.\ on Rewriting Techniques and Applications (RTA
  2010)}, volume~6 of {\em Leibniz International Proceedings in Informatics},
  pages 401--416. Schloss Dagstuhl, 2010.

\end{thebibliography}

\end{document}